\documentclass[11pt,b5paper,notitlepage]{article}
\usepackage[b5paper, margin={0.5in,0.65in}]{geometry}
\usepackage{amsmath,amscd,amssymb,amsthm,mathrsfs,amsfonts,layout,indentfirst,graphicx,caption,mathabx, stmaryrd,appendix,calc,imakeidx,upgreek} 
\usepackage[dvipsnames]{xcolor}
\usepackage{palatino}  
\usepackage{slashed} 
\usepackage{mathrsfs} 
\usepackage{extarrows} 
\usepackage{enumitem} 

\usepackage{csquotes} 
\usepackage{epigraph}   

\usepackage{fancyhdr} 

\usepackage{ulem}  

\usepackage{relsize} 

\usepackage{verbatim}  

\usepackage{halloweenmath} 

\usepackage{tcolorbox}
\tcbuselibrary{theorems}

\usepackage{pdfrender}

\usepackage{tikz}
\usetikzlibrary{fadings}
\usetikzlibrary{patterns}
\usetikzlibrary{shadows.blur}
\usetikzlibrary{shapes}

\usepackage{tikz-cd}
\usepackage[nottoc]{tocbibind}   

\usepackage{float}

\usepackage{lipsum}  
\let\OLDthebibliography\thebibliography
\renewcommand\thebibliography[1]{
	\OLDthebibliography{#1}
	\setlength{\parskip}{0pt}
	\setlength{\itemsep}{2pt} 
}

\allowdisplaybreaks  
\usepackage{latexsym}
\usepackage{chngcntr}
\usepackage[colorlinks,linkcolor=blue,anchorcolor=blue, linktocpage,
]{hyperref}
\hypersetup{ urlcolor=cyan,
	citecolor=[rgb]{0,0.5,0}}

\counterwithin{figure}{section}

\theoremstyle{definition}
\newtheorem{df}{Definition}[section]
\newtheorem{eg}[df]{Example}

\newtheorem{rem}[df]{Remark}

\newtheorem{cv}[df]{Convention}

\newtheorem{que}[df]{Question}
\newtheorem{prob}[df]{Problem}

\theoremstyle{plain}
\newtheorem{thm}[df]{Theorem}

\newtheorem{pp}[df]{Proposition}
\newtheorem{co}[df]{Corollary}
\newtheorem{lm}[df]{Lemma}

\newtheorem{mrs}[df]{Main Result}

\usepackage{calligra}
\DeclareMathOperator{\shom}{\mathscr{H}\text{\kern -3pt {\calligra\large om}}\,}
\DeclareMathOperator{\sext}{\mathscr{E}\text{\kern -3pt {\calligra\large xt}}\,}
\DeclareMathOperator{\Rel}{\mathscr{R}\text{\kern -3pt {\calligra\large el}~}\,}
\DeclareMathOperator{\sann}{\mathscr{A}\text{\kern -3pt {\calligra\large nn}}\,}
\DeclareMathOperator{\send}{\mathscr{E}\text{\kern -3pt {\calligra\large nd}}\,}
\DeclareMathOperator{\stor}{\mathscr{T}\text{\kern -3pt {\calligra\large or}}\,}

\usepackage{aurical}
\DeclareMathOperator{\VVir}{\text{\Fontlukas V}\text{\kern -0pt {\Fontlukas\large ir}}\,}

\newcommand{\fk}{\mathfrak}
\newcommand{\mc}{\mathcal}
\newcommand{\wtd}{\widetilde}
\newcommand{\wht}{\widehat}

\newcommand{\ovl}{\overline}

\newcommand{\End}{\mathrm{End}} 
\newcommand{\idt}{\mathbf{1}}
\newcommand{\id}{\mathrm{id}}
\newcommand{\Hom}{\mathrm{Hom}}

\newcommand{\op}{\mathrm{op}}
\newcommand{\Rep}{\mathrm{Rep}}

\newcommand{\RepGA}{\mathrm{Rep}^G(\mathcal A)}

\newcommand{\Diff}{\mathrm{Diff}}

\newcommand{\DiffS}{\mathrm{Diff}^+(\mathbb S^1)}
\newcommand{\PSU}{\mathrm{PSU}(1,1)}
\newcommand{\Vir}{\mathrm{Vir}}

\newcommand{\bk}[1]{\langle {#1}\rangle}

\newcommand{\GA}{\mathscr G_{\mathcal A}}

\newcommand{\scr}{\mathscr}

\newcommand{\Jtd}{\widetilde{\mathcal J}}
\newcommand{\gk}{\mathfrak g}

\newcommand{\Ad}{\mathrm{Ad}}

\newcommand{\im}{\mathbf{i}}
\newcommand{\Co}{\complement}

\newcommand{\RepA}{\mathrm{Rep}(\mathcal A)}

\newcommand{\mbb}{\mathbb}
\newcommand{\mbf}{\mathbf}

\newcommand{\Cbb}{\mathbb C}

\newcommand{\Zbb}{\mathbb Z}

\newcommand{\Rbb}{\mathbb R}

\renewcommand{\Bbb}{\mathbb B}

\newcommand{\Sbb}{{\mathbb S}}
\newcommand{\Gc}{\mathscr G_c}

\newcommand{\bpr}{{}^\backprime}

\newcommand{\cl}{\mathrm{cl}}

\newcommand{\rk}{\mathfrak r}

\newcommand{\Jbb}{\mathbb J}

\newcommand{\Aut}{\mathrm{Aut}}

\newcommand{\Int}{\mathrm{Int}}

\newcommand{\eps}{\varepsilon}
\newcommand{\MA}{\mathcal A}

\newcommand{\MH}{\mathcal H}
\newcommand{\MJ}{\mathcal J}
\newcommand{\MU}{\mathcal U}
\newcommand{\MK}{\mathcal K}
\newcommand{\MM}{\mathcal M}
\newcommand{\MO}{\mathcal O}
\newcommand{\MN}{\mathcal N}
\newcommand{\SG}{\mathscr G}
\newcommand{\Maa}{{\mathcal A^{\otimes 2}}}
\newcommand{\sqz}{{\sqrt 0}}
\newcommand{\sbsc}[1]{_{\raisebox{-0.6ex}{$\scriptstyle {#1}$}}}

\newcommand{\MHcl}{\mathcal H_{\mathrm{cl}}}
\newcommand{\Qbf}{\mathbf Q}
\newcommand{\Roo}{\mathbb R^{1,1}}
\newcommand{\Rcl}{\mathbb R^{1,1}_{\mathrm{cl}}}
\newcommand{\Rop}{\mathbb R^{1,1}_{\mathrm{op}}}
\newcommand{\Rpp}{\mathbb R^{1,1}_{(-\pi,\pi)}}
\newcommand{\Rzp}{\mathbb R^{1,1}_{(0,\pi)}}
\newcommand{\Bcl}{\mathcal B_{\mathrm{cl}}}
\newcommand{\Opp}{\mathcal O_{(-\pi,\pi)}}
\newcommand{\Ozp}{\mathcal O_{(0,\pi)}}
\newcommand{\ZOpp}{\sqrt{\mathcal O_{(-\pi,\pi)}}}
\newcommand{\Hcl}{\mathcal H_{\mathrm{cl}}}
\newcommand{\SL}{\mathscr L}
\newcommand{\Ucl}{U_{\mathrm{cl}}}
\newcommand{\SA}{\scr A}
\newcommand{\SB}{\scr B}
\newcommand{\Std}{\widetilde {\mathbb S}^1}
\newcommand{\Jcl}{\mathfrak J_{\mathrm{cl}}}
\newcommand{\Bop}[1]{\mathcal B_{\mathrm{op}({#1})}}
\newcommand{\Bopi}{\mathcal B_{\mathrm{op}(i)}}
\newcommand{\Bopj}{\mathcal B_{\mathrm{op}(j)}}
\newcommand{\Hop}[1]{\mathcal H_{\mathrm{op}({#1})}}
\newcommand{\Hopij}{\mathcal H_{\mathrm{op}(i,j)}}
\newcommand{\Hopji}{\mathcal H_{\mathrm{op}(j,i)}}
\newcommand{\Hopii}{\mathcal H_{\mathrm{op}(i,i)}}
\newcommand{\opij}{\mathrm{op}(i,j)}
\newcommand{\opii}{\mathrm{op}(i,i)}
\newcommand{\opjj}{\mathrm{op}(j,j)}

\newcommand{\hqed}{\hfill\qedsymbol}

\usepackage{tipa} 

\usepackage{tipx}

\numberwithin{equation}{section}

\title{Minkowskian open/closed conformal field theory possibly without vacuum: the Cardy case}
\author{{\sc Bin Gui}
}
\date{}
\begin{document}\sloppy 
	\pagenumbering{arabic}


	\maketitle

\begin{abstract}
For any conformal net, not necessarily rational, we construct the associated Cardy-type conformal field theory on the Minkowski spacetimes $\Rcl=(\Rbb/2\pi\Zbb)\times\Rbb$ for closed strings and $\Rop=[0,\pi]\times\Rbb$ for open strings within the framework of algebraic quantum field theory. In addition to verifying some of their basic properties, we prove three forms of Haag duality for multi-double-cones and boundary intervals, interpreted respectively as the Minkowskian versions of modular invariance, the Cardy consistency condition, and the Morita equivalence of boundary field algebras.
\end{abstract}

\tableofcontents

\section{Introduction}

\subsection{CFT beyond discrete-type}

One of the ultimate goals in the mathematical study of two-dimensional conformal field theory (CFT) is the rigorous construction of the theory describing both open and closed strings, known as an \textbf{open/closed CFT} or \textbf{full/boundary CFT}. Like ordinary quantum field theories, a CFT can be formulated either on Euclidean or on Minkowskian spacetimes.  

In the Euclidean setting, the theory is defined on compact Riemann surfaces with physical boundaries, together with parametrized boundary circles and intervals, and is required to satisfy Segal's axioms (cf. \cite{Seg88,Seg04}). In the Minkowskian setting, particularly in the framework of algebraic quantum field theory (AQFT), the theory is defined on the Minkowski spacetimes $\Rcl=(\Rbb/2\pi\Zbb)\times\Rbb$ and $\Rop=[0,\pi]\times\Rbb$ and is required to satisfy e.g. the axioms of a Haag--Kastler net in \cite{Haag96}.\footnote{In the standard AQFT literature, the spacetime for the open CFT is usually taken to be the half-plane $[0,+\infty)\times\Rbb$ rather than the strip $[0,\pi]\times\Rbb$ . We will discuss this difference in more detail in Sec. \ref{lb159}.}

Over the past three decades, substantial progress has been made in the rigorous construction of \textit{rational} CFTs, both in the Euclidean and Minkowskian settings. In the Euclidean setting, the construction of rational CFTs using vertex operator algebras (VOAs) has been achieved in genus 0 and 1 in \cite{HK07,Kong07,Kong08a,Kong08b}; see \cite{KR10,Kong11} for surveys. For Riemann surfaces of arbitrary genus, a complete construction remains open, but a topological field theory (TFT) approach has been developed in \cite{FFFS02,FRS02,FRS05,FFRS06}. In the Minkowskian setting, the construction of rational CFTs using conformal nets has been carried out in \cite{LR95,KL04,LR04,LR09,CKL13,BKL15}; see \cite{KR10,Kaw15} for reviews. In all these approaches, the method of Q-systems ($\approx$ Frobenius algebras) plays a crucial role.

Beyond rational CFTs, recent progress has also been made in the construction of Euclidean closed CFTs (cf. \cite{Mor20,Mor23}) and Minkowskian closed CFTs (cf. \cite{AGT23,AGT25}). The theories constructed in these works, while not necessarily rational, are still \textit{discrete} in the sense that their state spaces are given by finite or infinite direct sums, rather than direct integrals, of irreducible representations of the tensor product of the chiral and antichiral algebras.

However, the frameworks developed in the previously mentioned approaches do not encompass (possibly) non-discrete CFTs, e.g., those whose closed-string state space $\MH_\cl$ admits a direct integral
\begin{align*}
\MH_\cl\simeq\int_X^\oplus \MH_x\otimes_\Cbb\MK_x~d\mu(x)
\end{align*}
where $(X,\mu)$ is a measure space and $\MH_x,\MK_x$ are (irreducible) representations of the chiral and antichiral algebras. The most prominent example is Liouville CFT, which has recently been constructed rigorously in the Euclidean setting using probabilistic methods; see \cite{KRV20,GKRV24,GKRV21} for the closed theory, and \cite{RZ22,Wu22,ARS25,ARSZ23,GRW24} for the open theory.

Liouville CFT is regarded as a Cardy-type open/closed CFT associated with the chiral algebra $\Vir_c$, the unitary Virasoro algebra (or Virasoro conformal net) with central charge $c=1+6Q^2$ where $Q=b+b^{-1}$ (and hence $c\geq 25$). Its closed-string state space decomposes, as a $\Vir_c\otimes\Vir_c$-module, as
\begin{align}\label{eq83}
\Hcl\simeq\int_{\Rbb_{\geq0}}^\oplus \MH_{\frac{Q^2}4+P^2}\otimes\ovl{\MH_{\frac{Q^2}4+P^2}}~dP
\end{align}
where, for each $h\geq0$, $\MH_h$ denotes the Hilbert-space completion of the irreducible positive-energy representation of $\Vir_c$ with lowest weight $h$. Besides its continuous spectrum, a remarkable feature of Liouville CFT is that \uwave{the vacuum vector does not belong to the state space $\Hcl$}. In fact, the vacuum vector (which has conformal weight $0$)  does not even exist as a distributional state in $\Hcl$, since the lowest weight of $\Hcl$ is strictly positive.

By contrast, in all previous approaches based on conformal nets or VOAs, the state space contains a vacuum vector. This raises the natural question of whether Liouville CFT can be accommodated within either framework. Since both conformal nets and VOAs have traditionally served as frameworks for general CFTs rather than for particular classes of examples, one is led to ask:

\begin{que}\label{lb152}
Using either the theory of conformal nets or VOAs, and in either the Minkowskian or Euclidean setting, can one develop a general framework for open/closed CFTs in which both discrete-type CFTs (such as rational CFTs) and continuous-type CFTs (whose state spaces may not contain a vacuum vector) can be studied within a unified setting?
\end{que}

\subsection{A general framework for Minkowskian open/closed CFT}\label{lb163}

The goal of the present paper is to provide a positive answer to Question \ref{lb152} for Minkowskian CFTs using the conformal net method, at least for Cardy-type CFTs. See \cite{HT26} for an approach to Euclidean CFTs of Cardy-type, also based on conformal nets. We expect that, by extending the results for conformal nets to nonlocal chiral (i.e. one-dimensional) nets, the framework developed here can eventually be generalized to non-Cardy-type Minkowskian CFTs.

To introduce the framework of this paper and its main results, we first sketch what a Minkowskian open/closed CFT (not necessarily of Cardy-type) should look like, adapting part of the general picture of \cite{FRS02,FRS05,FFRS06} to the AQFT setting.

In a Minkowskian CFT, boundary conditions are labeled by symbols $i,j,k,\dots$. The world sheet for the \text{open string} $[0,\pi]$ is the Minkowski space
\begin{align*}
\Rop=[0,\pi]\times\Rbb
\end{align*}
The Hilbert space for the open string with left boundary condition $i$ and right boundary condition $j$ is denoted by $\pmb{\Hopij}$. 

The boundary operators on $\Rbb$ preserving the boundary condition $i$ form a (possibly) non-local chiral net $\pmb{\Bopi}$, called the \textbf{boundary net}. Since the boundary fields on $\partial_+\Rop:=\{0\}\times\Rbb$ preserving the boundary condition $i$ act on $\Hopij$, the Hilbert space $\Hopij$ carries a left action of $\Bopi$. See Thm. \ref{lb114} for the basic properties expected of nonlocal chiral nets, and Thm. \ref{lb119} for the properties that the action of $\Bopi$ on $\Hopij$ is expected to satisfy.

Similarly, since the boundary fields on $\partial_-\Rop:=\{\pi\}\times\Rbb$ preserving the boundary condition $j$ act on $\Hopij$, the Hilbert space $\Hopij$ admits a right action of $\Bopj$, in other words, a left action of the dual net $\Bopj'$ of $\Bopj$; cf. Thm. \ref{lb114}-(6) for the meaning of dual nets. 

The left action of $\Bopi$ and the right action of $\Bopj$ on $\Hopij$ are, in an appropriate sense, mutual commutants. More precisely, they satisfy the \textbf{boundary-boundary Haag duality}, as formulated in Thm. \ref{lb119}. It is in this sense that $\Bopi$ and $\Bopj$ are viewed as Morita equivalent.

If $\ovl{\Hopij}$ denotes the conjugate of $\Hopij$, one expects a canonical equivalence
\begin{subequations}\label{eq80}
\begin{align}
\Hopji\simeq\ovl{\Hopij}
\end{align}
Moreover, one expects a canonical equivalence
\begin{align}
\Hop{i,k}\simeq\Hopij\boxtimes_{\Bopj}\Hop{j,k}
\end{align}
\end{subequations}
where the RHS is the (yet to be defined) Connes fusion product of the right $\Bopj$-module $\Hopij$ and the left $\Bopj$-module $\Hop{j,k}$.

The world sheet for the \text{closed string} $\Sbb^1=\Rbb/2\pi\Zbb$ is the Minkowski space
\begin{align*}
\Rcl=\Sbb^1\times\Rbb
\end{align*}
commonly known as the Einstein cylinder. Let $\pmb{\Hcl}$ denote the state space of the closed string. The bulk fields on $\Rcl$ then form a net of algebras acting on $\Hcl$, denoted by $\pmb{\Bcl}$ and called the \textbf{bulk net}. See Thm. \ref{lb72} for the properties satisfied by this net. 

Besides the properties mentioned in Thm. \ref{lb72}, the net $\Bcl$ should also satisfy the \textbf{bulk-bulk Haag duality}, namely, Haag duality for multi-double-cones in $\Rcl$. See Thm. \ref{lb115} for the precise statement. For unions of two or more double-cones, this duality is equivalent, in the terminology of \cite{BKL15}, to the statement that the $\mu$-index of $\Bcl$ is $1$. At least in the rational case, this condition is equivalent to (genus-one) modular invariance in the Euclidean picture. Indeed, by \cite[Prop. 6.6]{BKL15} and \cite[Thm. 6.7]{Kong08b}, both conditions admit the same formulation in terms of Q-systems (or Frobenius algebras). See also \cite{BCKLM24} for a review of this topic from the perspective of entropies.

Given left and right boundary conditions $i$ and $j$, the bulk fields on the interior $\Int\Rop:=(0,\pi)\times\Rbb$ act on the open-string state space $\Hopij$. Thus, the restriction of the bulk net to $\Int\Rop$ acts on $\Hopij$. See Thm. \ref{lb149} for a list of properties expected of this action. In particular, it should satisfy the \textbf{bulk-boundary Haag duality}; see Thm. \ref{lb149}-(5) and Thm. \ref{lb153} for the precise statement. This form of Haag duality should be regarded as the Minkowskian analogue of the Cardy consistency condition; see \cite[Sec. 4.2]{KR10} for a related discussion.

The interpretations of the three forms of Haag duality discussed above are summarized in Table \ref{tb1}.

\begin{table}[h]
\centering
\begin{tabular}{|c|c|c|c|}
\hline Haag duality& Spacetime & $\begin{array}{c}
\rule{0pt}{0.42cm}\text{Pictorial}\\
\text{illustration}
\end{array}$ & Interpretation   \\
\hline Bulk-bulk   & $\Rcl$ & $\begin{array}{c}
\rule{0pt}{0.42cm}\text{Fig. \ref{fig3}}\\
\text{in Sec. \ref{lb145}}
\end{array}$  & $\begin{array}{c}
\rule{0pt}{0.42cm}\text{Modular invariance}\\
\text{in Euclidean spacetimes}
\end{array}$   \\
\hline Boundary-boundary  & $\Rop$ & $\begin{array}{c}
\rule{0pt}{0.42cm}\text{Fig. \ref{fig4}}\\
\text{near Sec. \ref{lb130}}
\end{array}$ & $\begin{array}{c}
\rule{0pt}{0.42cm}\text{Morita equivalence}\\
\text{of boundary nets}
\end{array}$   \\
\hline Bulk-boundary  & $\Rop$ & $\begin{array}{c}
\rule{0pt}{0.42cm}\text{Fig. \ref{fig5}}\\
\text{in Sec. \ref{lb154}}
\end{array}$ & $\begin{array}{c}
\rule{0pt}{0.42cm}\text{Cardy consistency condition}\\
\text{in Euclidean spacetimes}
\end{array}$  \\
\hline
\end{tabular}
\caption{Three forms of Haag duality}\label{tb1}
\end{table}

\subsection{Main results}

We now consider the case where the CFT is of \textbf{Cardy-type}, meaning that there exists a boundary condition, denoted by $0$, such that, upon setting
\begin{align*}
\MA:=\Bop{0}\qquad\MH_0:=\Hop{0,0}
\end{align*}
$\MA$ is a local chiral net (i.e. a conformal net) with vacuum module $\MH_0$ containing a vacuum vector $\Omega$. Each boundary condition $i$ then determines a right $\MA$-module (also called solitonic $\MA$-module)
\begin{align*}
\MH_i:=\Hop{i,0}
\end{align*}
As in \cite{FFFS02}, we restrict for simplicity to those boundary conditions $i$ for which $\MH_i$ is a local $\MA$-module, simply called an $\MA$-module. Then \eqref{eq80} implies $\Hop{0,j}\simeq\ovl{\MH_j}$, and hence
\begin{align}\label{eq81}
\Hopij\simeq\MH_i\boxtimes_\MA\ovl{\MH_j}
\end{align}

The main results of this paper are summarized as follows:

\begin{mrs}\label{lb157}
Let $\MA$ be any conformal net with vacuum module $\MH_0$. 
\begin{enumerate}[label=(\arabic*)]
\item We construct the closed-string state space $\Hcl$ as an $\Maa$-module. See Def. \ref{lb155}. Moreover, we construct the net $\Bcl$ on the spacetime $\Rcl$, acting on the Hilbert space $\Hcl$. See Def. \ref{lb64} and \ref{lb74}. 
\item We prove that the pair $(\Bcl,\Hcl)$ satisfies the usual properties required of a Haag--Kastler net, except for the existence of a vacuum. See Thm. \ref{lb72}.
\item We prove that $\Bcl$ satisfies the bulk-bulk Haag duality. See Thm. \ref{lb115}.
\end{enumerate}
\end{mrs}

\begin{mrs}\label{lb158}
For each pair of $\MA$-modules $\MH_i,\MH_j$, define the $\MA$-module $\Hopij$ by \eqref{eq81}, i.e., $\Hopij:=\MH_i\boxtimes_\MA\ovl{\MH_j}$.
\begin{enumerate}[label=(\arabic*)]
\item We construct the non-local chiral net $\Bopi$ on the spacetime $\Rbb$, acting on the state space $\Hop{i,i}$. See Def. \ref{lb112}. We prove that the pair $(\Bopi,\Hopii)$ satisfies the usual properties required of a Haag--Kastler net, except for the existence of a vacuum. See Thm. \ref{lb114}.

\item We construct the left and right actions of $\Bopi$ and $\Bopj$ on $\Hopij$. See Def. \ref{lb118}. We prove that these actions satisfy the usual requirements for representations of Haag--Kastler nets, and we prove the boundary-boundary Haag duality. See Thm. \ref{lb119}.

\item We construct the action of $\Bcl$ (restricted to $(0,\pi)\times\Rbb$) on $\Hopij$. See Def. \ref{lb120} and Thm. \ref{lb149}. We show that the joint actions of $\Bopi$, $\Bopj$, and $\Bcl$ on $\Hopij$ satisfy the usual requirements for representations of Haag--Kastler nets, and we prove the bulk-boundary Haag duality. See Thm. \ref{lb149} and \ref{lb153}.
\end{enumerate}
\end{mrs}

Moreover, a brief discussion of the nets of algebroids of boundary operators changing boundary conditions is given in Sec. \ref{lb156}.

\subsection{Proof strategies}\label{lb159}

Since Main Result \ref{lb158} is relatively easier to obtain than Main Result \ref{lb157}, we focus here on the main idea behind the proof of Main Result \ref{lb157}, namely, the construction of the closed CFT. We also explain why the existing techniques from operator algebras and AQFT cannot be applied directly, and how we overcome the resulting difficulties.

Adapting the construction of \cite{LR04}, it is not difficult to construct the action of the bulk net $\Bcl$ (restricted to $(0,\pi)\times\Rbb$) on $\Hop{0,0}=\MH_0$ in terms of relative commutants. Indeed, this construction is simply a reformulation of the bulk-boundary Haag duality for the left and right actions of $\Bop{0}=\MA$ together with the action of $\Bcl$ on $\Hop{0,0}$. 

Already at this stage, however, an important difference emerges between our approach and those taken in most of the AQFT literature. In almost all previous AQFT works, the spacetime for the open-string theory is taken to be $[0,+\infty)\times\Rbb$. In other words, the open string is modeled by the interval $[0,+\infty)$, which has a single boundary component $\{0\}$, rather than by the interval $[0,\pi]$, which has two boundary components. There are several reasons for adopting the latter viewpoint.

First, it is more compatible with the techniques developed in this paper, in particular the theory of crossed categorical extensions of conformal nets. Second, it is consistent with the Euclidean CFT literature \cite{FFFS02,FRS02,FRS05,FFRS06,HK07,Kong07,Kong08a,Kong08b}, where open strings are likewise modeled by intervals with two boundary components. Third, only by working with strings having two boundary components can one formulate the Morita equivalence of non-local nets in terms of boundary-boundary Haag duality, thereby avoiding the machinery of Q-systems/Frobenius algebras.

This brings us to a fundamental difficulty encountered by previous AQFT approaches when moving beyond discrete-type CFTs. Those approaches rely heavily on Q-systems, which work perfectly for rational CFTs, and in some cases also for irrational discrete-type CFTs or QFTs (see, for example, \cite{Mas00,DVG18}). In particular, the bulk net $\Bcl$ on the entire spacetime $\Rcl$, together with the state space $\Hcl$, was constructed in \cite{LR04} using the Longo--Rehren Q-system introduced in \cite{LR95} and later generalized in \cite{Reh00}. Even the construction of $\Bcl$ and $\Hcl$ in \cite{LR09}, which does not explicitly use Q-systems, still relies on the existence of a conditional expectation between a pair of subfactors. However, since the closed-string state space of a general CFT needs not contain a vacuum vector, there is no reason to expect such a conditional expectation to exist.

Our approach in this paper is based on the following simple observation. Let $\fk H$ be a faithful right $\MN$-module where $\MN$ is a von Neumann algebra. Let $\mc N'$ denote the commutant of $\MN$, acting on the left of $\fk H$. By \cite[Sec. IX.3]{Tak03}, the Hilbert space $\fk H\boxtimes_\MN\ovl{\fk H}$ carries a canonical structure of $\MN'$-$\MN'$ bimodule that is unitarily equivalent to $L^2(\MN')$. Moreover, although not stated explicitly there, this unitary equivalence pulls back the modular conjugation on $L^2(\MN')$ to an involutive antiunitary operator $\Theta$ on $\fk H\boxtimes_\MN\ovl{\fk H}$, which, roughly speaking, sends $\xi\boxtimes\ovl\eta$ to $\eta\boxtimes\ovl\xi$. It follows from the Tomita--Takesaki theory that the left and right actions of $\MN'$ on $\fk H\boxtimes_\MN\ovl{\fk H}$ are mutual commutants, and that $\Ad_\Theta$ interchanges the left and right actions.

The construction of the bulk net $\Bcl$ on $\Hcl$ from the action of the bulk net (restricted to $(0,\pi)\times\Rbb$) on $\Hop{0,0}$ resembles the construction of the $\MN'$-$\MN'$ bimodule $\fk H\boxtimes_\MN\ovl{\fk H}$ from the $\MN'$-$\MN$ bimodule $\fk H$. This analogy can be made precise by taking $\MN=\Maa(\wtd I)$ and $\fk H=\MH_\sqz$. Here, $\wtd I$ is an arg-valued interval (cf. Def. \ref{lb160}), and $\MH_\sqz$ is the canonical permutation-twisted $\Maa$-module associated with the $\MA$-module $\MH_0$ (cf. Def. \ref{lb52}), whose sector-theoretic counterpart was first introduced in \cite{LX04}. A crucial difficulty, however, is that $\Hcl$ should be independent of the choice of $\wtd I$, and the action of $\Bcl$ on $\Hcl$ must be compatible with the underlying geometry of both $\Sbb^1$ and $\Rcl$.

This is precisely why, unlike previous approaches in the AQFT literature, we work within the framework of \textbf{categorical extensions of conformal nets}, first introduced in \cite{Gui21a} and later generalized to the twisted setting in \cite{MS26a,MS26b}. This framework is particularly well suited to handling the geometry of $\Sbb^1$, whereas the sector-theoretic approach often requires removing a point from $\Sbb^1$, which is inconvenient in the present context.

However, even with the geometric flexibility provided by the (crossed) categorical extensions of conformal nets, the resulting net $\Bcl$ is initially defined only on $(-\pi,\pi)\times\Rbb$, and the corresponding action of $\Bcl$ on $\Hcl$ satisfies only some local versions of conformal covariance on $(-\pi,\pi)\times\Rbb$, see Sec. \ref{lb144}. The most nontrivial part of the proof of Main Result \ref{lb157}, carried out in Sec. \ref{lb161} and \ref{lb100}, is to extend this construction to a net $\Bcl$ on  the entire spacetime $\Rcl$, acting on $\Hcl$ and satisfying global conformal covariance. In particular, if $O$ is a double cone contained in $(-\pi,\pi)\times\Rbb$, and if, for each $s\in\Rbb$, we let $\Ucl(\varrho_c(s),\varrho_c(-s))$ denote the unitary operator on $\Hcl$ implementing translation by $s$ units along the $x$-direction (as in the main body of this paper), then we should have
\begin{align}\label{eq82}
\Ad_{\Ucl(\varrho_c(2\pi),\varrho_c(-2\pi))}\big(\Bcl(O)\big)=\Bcl(O)
\end{align}

In the AQFT literature, when dealing with rational conformal nets, global conformal covariance such as \eqref{eq82} is typically established by proving $\Ucl(\varrho_c(2\pi),\varrho_c(-2\pi))=1$, either by first obtaining an explicit irreducible decomposition of the $\Maa$-module $\Hcl$ and then verifying the equation componentwise, or by invoking the Bisognano--Wichmann property (cf. \cite[Prop. 2.1]{KL04} or \cite[Thm. A.5]{MT19}). In our approach, neither method is available, due to the absence of a discrete decomposition of the state space and the absence of a vacuum vector. We therefore leave the following problem open.

\begin{prob}\label{lb162}
Let $\MA$ be any conformal net. Prove that  $\Ucl(\varrho_c(2\pi),\varrho_c(-2\pi))$, the unitary operator on $\Hcl$ implementing translation by $2\pi$ along the $x$-axis, is the identity.
\end{prob}

Although we are unable to solve Problem \ref{lb162}, we can nevertheless prove \eqref{eq82} by exploiting the Tomita--Takesaki-type result for $\Theta$ mentioned above, together with the relationship between $\Theta$ and the translation group. In fact, \eqref{eq82} arises as a byproduct of our proof of the PCT symmetry and Haag duality for the double cone $O_\rightarrow$ and its spacelike complement $O_\leftarrow$, illustrated in Fig. \ref{fig2}. The proof, given in Sec. \ref{lb161}, of the PCT symmetry and Haag duality for $O_\rightarrow$ and $O_\leftarrow$, together with the resulting implication that these regions satisfy \eqref{eq82} (cf. Lem.~\ref{lb92}), is, in our view, the most novel part of the construction of open/closed CFT presented in this paper, although it is by no means the most technical.

Having established the results in Sec. \ref{lb161} and proved the global conformal covariance on $\Rcl$, one can readily deduce the bulk-bulk Haag duality for a single double-cone (and its causal complement). Once this case is established, the extension to multi-double-cones follows straightforwardly from the split property of $\MA$, see Sec. \ref{lb145}. In other words, the main difficulty in proving bulk-bulk Haag duality for closed CFT lies in the case of a single double-cone. For rational closed CFTs, this single double-cone Haag duality can still be proved even without modular invariance, because the state space contains a vacuum vector and hence the Bisognano--Wichmann theorem can be applied. In our setting, however, the state space $\Hcl$ does not in general admit a vacuum vector, so the Bisognano--Wichmann theorem is not available. Therefore, the crucial argument in Sec. \ref{lb161} is indispensable.

\subsection{Future directions}

It is, of course, desirable to generalize the constructions in this paper to Minkowskian open/closed CFTs that are not of Cardy type. To achieve this, one needs to develop a theory of Connes fusion for modules over non-local chiral nets, together with the corresponding theory of categorical extensions. However, unlike the case of conformal nets, the vacuum module of a non-local net (for example, $\Hopii$ for $\Bopi$) generally does not contain a vacuum vector, just as the closed state space does not. It is therefore desirable to develop a theory of Connes fusion that is compatible with the geometry of the circle, as was accomplished for conformal nets in \cite{Gui21a,MS26a,MS26b}.

In addition, one should show that Connes fusion defined using Morita equivalent non-local nets yields equivalent theories. This expectation is suggested by the discussion in Sec. \ref{lb163}, where $\Bopi$ and $\Bopj$ are Morita equivalent via the bimodule $\Hopij$, and hence
\begin{align*}
\Hop{k,i}\boxtimes_{\Bopi}\Hop{i,l}\simeq \Hop{k,l}\simeq \Hop{k,j}\boxtimes_{\Bopj}\Hop{j,l}
\end{align*}
This suggests that the corresponding Connes fusion theories should be equivalent, with the equivalence implemented by
\begin{gather*}
\Hop{k,i}\mapsto\Hop{k,j}\qquad \Hop{i,l}\mapsto\Hop{j,l}
\end{gather*}
that is, implemented by
\begin{gather*}
\Hop{k,i}\mapsto\Hop{k,i}\boxtimes_{\Bopi}\Hopij\qquad \Hop{i,l}\mapsto \Hopji\boxtimes_{\Bopi}\Hop{i,l}
\end{gather*}
We plan to investigate these topics in a forthcoming paper.

Moreover, one should establish that two non-local nets are Morita equivalent (in the sense of satisfying boundary-boundary Haag duality) if and only if the closed CFTs generated by them are equivalent. This would generalize the result that two (special symmetric) Frobenius algebras in a modular fusion category are Morita equivalent if and only if their centers are equivalent; cf. \cite{KR08}. In fact, the theory of Minkowskian open/closed CFTs should be viewed as a highly infinite analogue of the theory of Q-systems or Frobenius algebras, much as measure theory is a highly infinite generalization of finite-dimensional commutative algebras. Just as, in the latter setting, one must distinguish between the roles of $L^\infty$ and $L^2$, in the former one should distinguish between nets of algebras and Hilbert spaces.

Returning to Cardy-type CFTs, another future direction is to relate the construction in this paper to those arising from the probabilistic approach. For example, one should show that when $\MA$ is the Virasoro net with central charge $c=1+6Q^2$, where $Q=b+b^{-1}$, the $\Maa$-module $\Hcl$ decomposes as in \eqref{eq83}. This would identify the closed-string state space constructed in this paper with the one appearing in the literature on Liouville CFTs.

Moreover, as suggested by \eqref{eq81}, determining the open-string state space $\Hopij$ requires computing the Connes fusion $\MH_i\boxtimes_\MA\ovl{\MH_j}$ for Virasoro modules $\MH_i,\MH_j$, a problem proposed and studied in \cite{Tes01,Tes09}. Note also that, since $\Hcl$ is defined in this paper as the Connes fusion $\MH_\sqz\boxtimes_{\Maa}\ovl{\MH_\sqz}$, proving \eqref{eq83} may itself be viewed as a problem of computing Connes fusion.

After identifying the state spaces, the next step would be to relate the bulk and boundary nets, together with their actions on the state spaces, to the correlation functions in the Liouville CFT literature. For example, one should show that these nets and their actions coincide with those obtained by smearing the Wightman bulk and boundary fields in Liouville CFTs.

More generally, Toda CFTs are expected to be Cardy-type CFTs associated with the principal $W(\gk)$-algebra at Toda central charge, where $\gk$ is a simple Lie algebra. Liouville CFT is the special case corresponding to $\gk=\fk{sl}_2$. Toda CFTs, especially for $\gk=\fk{sl}_3$, have recently been studied rigorously from the probabilistic perspective \cite{CRV23,CH22,Cer24,Cer25,CH24a,CH24b,CH25}. On the AQFT side, the conformal net associated with $W_3=W(\fk{sl}_3)$ has been constructed in \cite{CTW23,CTW22}. Thus, at least for such conformal nets, one may ask whether the state spaces of the open/closed CFT constructed in this paper agree with those appearing in the probabilistic literature, and whether the bulk and boundary nets constructed here can likewise be realized by smearing the corresponding Wightman fields.

Note that, in Toda open CFTs, one should also allow boundary conditions corresponding to twisted $\MA$-modules, especially those arising from automorphisms of the Dynkin diagram (cf. \cite{CH24a,CH24b,CH25}). For simplicity, we have considered only untwisted $\MA$-modules as boundary conditions in this paper, although most of our results extend straightforwardly to twisted boundary conditions, or more generally to solitonic $\MA$-modules (cf. Def. \ref{lb3}) whose restrictions to the Virasoro subnet are genuine modules.

Another interesting example is the rank-$n$ Heisenberg conformal net, whose closed-string state space is expected to admit the decomposition
\begin{align*}
\Hcl\simeq\int_{\Rbb^n}^\oplus\MH_p\otimes\ovl{\MH_p}~dp
\end{align*}
where $\MH_p$ denotes the irreducible positive-energy representation of the Heisenberg conformal net with momentum $p$, cf. \cite[Sec. B.1]{Luk07} for instance. As mentioned earlier in this section, obtaining such a decomposition requires computing the Connes fusion of certain twisted modules over the Heisenberg conformal net. Significant progress on this problem has recently been made in \cite{MS25,MS26c}. It would then be desirable to compare the bulk and boundary nets constructed in this paper with those obtained by smearing the corresponding Wightman fields.

\subsection{Outline}

This paper is organized as follows. Chapter 2 collects the preliminaries on conformal nets used throughout the paper. Section 2.1 recalls the definition of conformal nets and lists the standard structural properties that will be used later. Section 2.2 recalls arg-valued intervals and twisted, untwisted, and solitonic modules of conformal nets, together with their conformal covariance. Section 2.3 reviews the category $\RepGA$ of $G$-twisted $\MA$-modules and gives a detailed description of the action of $G$ on $\RepGA$. Section 2.4 reviews the notion of $G$-crossed categorical extensions, while Section 2.5 develops several techniques concerning categorical extensions that play a crucial role throughout the paper. Finally, Sections 2.6 and 2.7 adapt the Tomita--Takesaki-type results mentioned in Sec. \ref{lb159} to the setting of $G$-crossed categorical extensions.

Chapter 3 is devoted to the construction of the closed CFT, thereby establishing Main Result \ref{lb157}. Section 3.1 reviews the permutation-twisted $\Maa$-module $\MH_\sqz$ and uses it to define the closed-string state space $\Hcl$. Section 3.2 establishes the geometric setup. Section 3.3 defines the bulk net $\Bcl$, first on double cones compactly contained in $(-\pi,\pi)\times\Rbb$ and then on all double cones of the Einstein cylinder $\Rcl$, and states its main Haag--Kastler properties. The proofs of these properties occupy Sections 3.4--3.10. Finally, the bulk-bulk Haag duality is established in Section 3.11.

Chapter 4 constructs the open CFT, proving Main Result \ref{lb158}. Section 4.1 introduces the geometric setup. Section 4.2 defines the boundary chiral nets and establishes their Haag--Kastler properties. Section 4.3 constructs the actions of the boundary nets on the boundary state spaces, proves the corresponding representation-theoretic properties, and establishes the boundary-boundary Haag duality. Sections 4.4 and 4.5 construct the action of the bulk net (restricted to $(0,\pi)\times\Rbb$) on the open-string state spaces, establish the basic properties of this action together with the actions of the boundary nets, and prove the bulk-boundary Haag duality. Section 4.6 briefly discusses nets of boundary algebroids whose operators change boundary conditions, extending several of the preceding constructions from boundary-preserving operators to boundary-changing ones.

\subsection{AI disclosure and methodology}

The only results obtained with the assistance of AI (ChatGPT 5.4 Plus and 5.5 Pro) are the proof of Lem. \ref{lb106}, Lem. \ref{lb97}, and the proof of Lem. \ref{lb125}. All other mathematical ideas in this work were developed without AI assistance.

\subsubsection*{Proof of Lem. \ref{lb106}}

The present project began after Adri\`a Mar\'in-Salvador informed the author that he had completed the works \cite{MS26a,MS26b}, which establish crossed categorical extensions of conformal nets, and shared them with the author. The author had previously recognized that crossed categorical extensions are necessary for open/closed Minkowski CFT. On the other hand, \textit{assuming that bulk-bulk Haag duality has been established for single double-cones} (cf. Thm. \ref{lb89}), its proof for multi-double-cones readily reduces to the computation of a relative commutant in Lem. \ref{lb106}. Therefore, proving Lem. \ref{lb106} was the natural next step.

The idea behind the proof of Lem. \ref{lb106} should be elementary to experts. However, at that time, the author did not immediately see why the lemma was true, partly because he was unfamiliar with some of the arguments involving the split property, and partly because he had understood the split property in a different---external/spatial rather than intrinsic---way. More precisely, the author viewed the unitary maps implementing the split as arising from charged fields from the vacuum module $\MH_0\otimes\MH_0$ of the tensor product conformal net $\MA\otimes\MA$ to the twisted module $\MH_{\sqrt0}$, which in the rational case can be obtained from smeared intertwining operators between the corresponding untwisted and twisted modules of the associated tensor product VOA $V\otimes V$. (In particular, the split property is related to the strong integrability of the twisted $V\otimes V$ module corresponding to $\MH_{\sqrt0}$.)

After learning the proof of Lem. \ref{lb106} from ChatGPT 5.4 Plus on April 14, 2026, the author gained initial confidence that the present project was worth pursuing and consequently devoted his full energy to it. The author became fully convinced that the project could succeed when he found the proof of bulk-bulk Haag duality for single double-cones, Thm. \ref{lb89}, without any AI assistance. (In fact, after finding his own proof, the author asked ChatGPT 5.5 Pro to prove Thm. \ref{lb89} independently, in order to test the capabilities of the AI model at that time. The AI-generated proof was incorrect, relying on some naive arguments.)

\subsubsection*{Lem. \ref{lb97}}

At the final stage of proving conformal covariance for closed string CFT, Thm. \ref{lb98}, the author had already established the crucial (global) rotation covariance, Lem. \ref{lb96}, as well as the local conformal covariance properties. It was natural to expect that these results could be combined to obtain full conformal covariance. However, the author's original argument for deriving Thm. \ref{lb98} from these properties contained a gap. At that point, the author lacked the patience to resolve the gap himself and therefore consulted ChatGPT 5.5 Pro. The AI suggested a lemma essentially equivalent to Lem. \ref{lb97}, which completed the proof of conformal covariance, although the author subsequently modified its formulation to obtain the version presented in this paper.

\subsubsection*{Proof of Lem. \ref{lb125}}

The author initially proved Lem. \ref{lb125} independently, but later discovered a gap in his proof. He then consulted ChatGPT 5.5 Pro. The final proof of Lem. \ref{lb125} presented in this article, which uses one-parameter subgroups of diffeomorphism groups, differs from ChatGPT's proof, which proceeds by expressing a diffeomorphism as a product of “small” diffeomorphisms. Although the author was uncertain about the correctness of ChatGPT's proof (ChatGPT 5.5 Pro was still prone to hallucinations compared with later models, and the author therefore did not fully trust many of its results at that time), that proof inspired him to find the argument presented in this article.

\subsubsection*{Writing}

The article was written by the author. AI (ChatGPT 5.5) was used only to find references needed for the introduction and, after the manuscript was completed, to improve grammar and check for typos and mistakes.

\subsection*{Acknowledgment}

I would like to thank Andr\'e Henriques for encouraging me to think about understanding Liouville CFT from an operator-algebraic perspective, and for many inspiring discussions related to the topics of this paper. I am also grateful to Yi-Zhi Huang, Liang Kong, Adri\`a Mar\'in-Salvador, Yuto Moriwaki, Eveliina Peltola, Xin Sun, J\"org Teschner, Baojun Wu, and Ziyun Xu for many helpful discussions. This work was supported by NSFC Grant 12401159.

\section{Preliminaries on conformal nets}

For any Hilbert spaces $\MH_1,\MH_2$, we let $\fk L(\MH_1,\MH_2)$ denote the set of bounded linear operators $\MH_1\rightarrow\MH_2$, and abbreviate $\fk L(\MH,\MH)$ as $\fk L(\MH)$. Similarly, we let $\MU(\MH_1,\MH_2)$ be the set of unitary maps $\MH_1\rightarrow\MH_2$ and write $\MU(\MH,\MH)$ as $\MU(\MH)$. For each Hilbert space $\MH$, we understand its inner product $\bk{\cdot|\cdot}$ to be linear on the right variable $|\cdot\rangle$ and antilinear on the left variable $\langle\cdot|$. For a collection of von Neumann algebras $(\mc M_i)_{i\in\scr I}$ acting on a common Hilbert space, we let $\bigvee_{i\in\scr I}\mc M_i$ be the von Neumann algebra generated by this collection, i.e., $\bigvee_{i\in\scr I}\mc M_i=\big(\bigcup_{i\in\scr I}\mc M_i\big)''$.

We write $X\Subset Y$ if $X\subset Y$, and if $X$ has compact closure in $Y$.

Throughout this paper, we adopt the conventions \ref{lb4}, \ref{lb9}, and \ref{lb10}, which will be introduced later in this chapter.

\subsection{Conformal nets}

Let $\DiffS$ be the orientation-preserving diffeomorphism group of the unit circle $\Sbb^1$. Let $\MJ$ be the set of (non-empty non-dense open) intervals in $\Sbb^1$. If $I\in\MJ$, we let $I'$ be the interior of $\Sbb^1\setminus I$. We let $\Diff_I(\Sbb^1)$ be the subgroup consisting of all $g\in\DiffS$ fixing points in the closure of $I'$. 

Let $\PSU$ be the subgroup of M\"obius transforms preserving $\Sbb^1$, i.e., $\PSU$ consists of $g\in\DiffS$ of the form
\begin{align*}
z\in\Sbb^1\mapsto\frac{az+b}{\ovl bz+\ovl a}
\end{align*}
where $a,b\in\Cbb,|a|^2-|b|^2=1$. Let
\begin{align}\label{eq3}
\varrho:\Rbb\rightarrow\DiffS\qquad \varrho(t)z=e^{\im t}z
\end{align}
(where $z\in\Sbb^1$) be the one-parameter rotation group.

In this paper, all conformal nets are assumed to be irreducible, defined as follows.

\begin{df}\label{lb2}
An (irreducible non-trivial) \textbf{conformal net} denotes a tuple $(\MA,\MH_0,U_0,\Omega)$ (abbreviated to $(\MA,\MH_0)$ or simply $\MA$) where $\MH_0$ is a separable Hilbert space with $\dim\MH_0>1$, and $\MA$ is a map sending each $I\in\MJ$ to a von Neumann algebra $\MA(I)$ on $\MH_0$ satisfying the following conditions.
\begin{enumerate}[label=(\alph*)]
\item (\textbf{Isotony}) If $I_1\subset I_2\in\MJ$, then $\MA(I_1)$ is a von Neumann subalgebra of $\MA(I_2)$.
\item (\textbf{Locality}) If $I_1,I_2\in\MJ$ are disjoint, then $[\MA(I_1),\MA(I_2)]=0$.
\item (\textbf{Conformal covariance}) $U_0$ is a strongly-continuous projectively unitary representation of $\DiffS$ on $\MH_0$ such that the relation
\begin{align*}
V\MA(I)V^*=\MA(gI)
\end{align*}
holds for any $g\in\DiffS$, any $I\in\MJ$, any unitary $V\in\MU(\MH_0)$ representing $U_0(g)$. Moreover, if $g\in\Diff_I(\Sbb^1)$ and $V$ represents $U_0(g)$, then $V\in\MA(I)$.
\item $\Omega\in\MH_0$, called the \textbf{vacuum vector}, is the unique (up to scalar multiplication) unit vector such that $V\Omega\in\Cbb\Omega$ for each $V\in\MU(\MH_0)$ representing some $g\in\PSU$. (In particular, $U_0$ restricts to a strongly-continuous unitary representation $U_0$ of $\PSU$ on $\MH$ satisfying $U_0(g)\Omega=\Omega$ for each $g\in\PSU$.) Moreover, $\Omega$ is cyclic under the action of $\bigvee_{I\in\MJ}\MA(I)$. 
\item (Positive energy) The self-adjoint generator of the one-parameter unitary group $U_0\circ\varrho$ has spectrum in $\Rbb_{\geq0}$.
\end{enumerate}
\end{df}

\begin{rem}\label{lb1}
A conformal net $\MA$ satisfies the following properties, cf. \cite[Prop. 1.1, 1.2]{GL96} and the references therein.
\begin{enumerate}[label=(\arabic*)]
\item (\textbf{Additivity}) $\MA(I)=\bigvee_{\alpha}\MA(I_\alpha)$ if $(I_\alpha)$ is a collection in $\MJ$ whose union is $I\in\MJ$.
\item (\textbf{Haag duality}) $\MA(I)'=\MA(I')$ for each $I\in\MJ$.
\item (\textbf{Reeh--Schlieder property}) $\MA(I)\Omega$ is dense in $\MH_0$ for each $I\in\MJ$.
\item (\textbf{Irreducibility}) $\bigvee_{I\in\MJ}\MA(I)=\fk L(\MH_0)$.
\item Each $\MA(I)$ is a type III factor. 
\end{enumerate}
Moreover, by \cite{MTW18}, $\MA$ satisfies the \textbf{split property}, which means that for each $I,J\in\MJ$ with $I\Subset J$, there is a type I factor between $\MA(I)$ and $\MA(J)$. By Haag duality (applied to $I_1=I$ and $I_2=J'$), this is equivalent to the statement that for any $I_1,I_2\in\MJ$ with disjoint closures, there is a unitary map $\Phi:\MH_0\rightarrow\MH_1\otimes\MH_2$ where the Hilbert spaces $\MH_1$, $\MH_2$ carry normal representations $\pi_1$, $\pi_2$ of $\MA(I_1)$ and $\MA(I_2)$, respectively, such that
\begin{align*}
\Ad_\Phi|_{\MA(I_1)}=\pi_1\otimes\id_{\MH_2}\qquad \Ad_\Phi|_{\MA(I_2)}=\id_{\MH_1}\otimes\pi_2
\end{align*}
By property (5), one may in fact take $\MH_1=\MH_2=\MH_0$, with $\pi_1,\pi_2$ given by the inclusion maps $\MA(I_1)\hookrightarrow\fk L(\MH_0),\MA(I_2)\hookrightarrow\fk L(\MH_0)$, respectively.
\end{rem}

Next, we review some useful extensions of $\DiffS$ and $\Diff_I(\Sbb^1)$. 
\begin{df}
Let
\begin{align}\label{eq1}
\pmb{\SG}:=\wtd\DiffS\rightarrow\DiffS
\end{align}
be the universal cover of $\DiffS$. Viewing $\Sbb^1$ as $\Rbb/2\pi\Zbb$, the group $\SG$ consists of smooth functions $F:\Rbb\rightarrow\Rbb$ satisfying for each $x\in\Rbb$ that
\begin{align}\label{eq2}
F(x+2\pi)=F(x)+2\pi\qquad F'(x)>0
\end{align}
\end{df}

\begin{df}\label{lb123}
For each $I\in\MJ$, let $\pmb{\SG(I)}$ be the identity component of the preimage of $\Diff_I(\Sbb^1)$ under the map \eqref{eq1}. Equivalently, $\SG(I)$ consists of smooth functions $F:\Rbb\rightarrow\Rbb$ satisfying \eqref{eq2} and fixing pointwise preimage of $I'$ under the covering map $\Rbb\rightarrow\Rbb/2\pi\Zbb$. (In fact, by \eqref{eq2}, it suffices to assume fixing pointwise any prescribed connected component of the preimage of $I'$.) See Rem. \ref{lb129} for an alternative description.
\end{df}

Lift the rotation group \eqref{eq3} to a one-parameter group
\begin{align*}
\pmb{\varrho}:\Rbb\rightarrow\SG
\end{align*}
Thus, for each $t\in\Rbb$, the function $\varrho(t):\Rbb\rightarrow\Rbb$ sends $x$ to $x+t$.

The \textbf{central extension of $\SG$ associated to $\MA$} is defined to be the topological group
\begin{align*}
\pmb{\GA}=\{(g,V)\in\SG\times\MU(\MH_0):V\text{ represents }U_0(g)\}
\end{align*}
or more precisely, the surjective group homomorphism $\GA\rightarrow\SG$ defined by the projection $\SG\times\MU(\MH_0)\rightarrow\SG$. The kernel of this homomorphism is $1\times\Sbb^1\simeq\Sbb^1$.

\begin{rem}\label{lb39}
The topological group $\GA$ depends only on the central charge of $\MA$. Indeed, the net $\Vir_\MA:I\in\MJ\mapsto U_0(\Diff_I(\Sbb^1))''$ acting on $\MK_0:=\ovl{\Vir_\MA\Omega}$ is a conformal subnet of $\MA$; in particular, the representation of $\PSU$ on $\MH_0$ restricts to that on $\MK_0$. Therefore, by the irreducibility in Rem. \ref{lb1}, $\MK_0$ is irreducible as a strongly-continuous projectively-unitary positive-energy representation of $\DiffS$, and hence is integrated from the unitary module $L(c,0)$ (for some $c\geq0$, called the \textbf{central charge} of $\MA$) of the Virasoro algebra, cf. \cite[Thm. A.1]{Car04}. Thus $\Vir_\MA$ can be identified with the Virasoro net $\Vir_c$ with central charge $c$, and $\MK_0$ can be viewed as the Hilbert space completion of $L(c,0)$. Set
\begin{align*}
\pmb{\Gc}:=\SG_{\Vir_c}
\end{align*}
Then the following map defines an isomorphism of topological groups $\GA\simeq\Gc$.
\begin{align*}
\GA\xlongrightarrow{\simeq}\Gc \qquad (g,V)\mapsto (g,V|_{\MK_0})
\end{align*}
We will freely identify $\GA$ and $\Gc$ throughout the rest of this paper.
\end{rem}

\begin{df}
If $\wtd g=(g,V)\in\GA$, we write $V\in\MU(\MH_0)$ as $\pmb{U_0(\wtd g)}$, that is,
\begin{align*}
\wtd g=(g,U_0(\wtd g))
\end{align*}
Then $U_0$ is a strongly-continuous unitary representation of $\GA\simeq\Gc$ on $\MH_0$. 
\end{df}

For each $I\in\MJ$, we let $\pmb{\GA(I)}$ be the preimage of $\SG(I)$ under the map $\GA\rightarrow\SG$. In particular, we let $\pmb{\Gc(I)}=\SG_{\Vir_c}(I)$. By the conformal covariance in Def. \ref{lb2}, we have
\begin{align*}
U_0(\GA(I))\subset\MA(I)
\end{align*}


\begin{rem}\label{lb40}
From the definition of conformal nets, we have a unitary representation $U_0$ of $\PSU$ rather than merely a projectively unitary one. Thus, we have a one parameter unitary group
\begin{align*}
\pmb{\varrho_\MA}:\Rbb\rightarrow\GA\qquad \varrho_\MA(t)=(\varrho(t),U_0(\varrho(t)))
\end{align*}
called the \textbf{rotation group in $\GA$}. We let $\pmb{\varrho_c}$ denote the rotation group $\varrho_{\Vir_c}$ of $\Vir_c$. The isomorphism $\GA\simeq\Gc$ clearly identifies $\varrho_\MA$ with $\varrho_c$.
\end{rem}

\subsection{Twisted modules of conformal nets}

Fix a conformal net $\MA$. 

\begin{df}
An \textbf{automorphism} of $\MA$ denotes a unitary operator $\Phi\in\MU(\MH_0)$ fixing $\Omega$ and satisfying $\Phi\MA(I)\Phi^{-1}=\MA(I)$ for each $I\in\MJ$. 
\end{df}

We fix a group homomorphism $G\rightarrow\Aut(\MA)$ where $G$ is a discrete group, and $\Aut(\MA)$ is the group of automorphisms of $\MA$. We view each $\phi\in G$ as a unitary operator on $\MH_0$. Then
\begin{align*}
\phi\MA(I)\phi^{-1}=\MA(I)
\end{align*}
Hence $G$ gives rise to a group of automorphisms of each $\MA(I)$.

\begin{df}\label{lb160}
For each $I\in\MJ$, an \textbf{arg function} denotes a continuous function $\pmb{\arg_I}:I\rightarrow\Rbb$ such that $z=e^{\im\arg_I(z)}$ for all $z\in I$. An \textbf{arg-valued interval} denotes a pair
\begin{align*}
\wtd I=(I,\arg_I)
\end{align*} 
where $I\in\MJ$, and $\arg_I$ is an arg function of $I$. The set of arg-valued intervals is denoted by $\pmb{\Jtd}$.
\end{df}

Equivalently, $\wtd I$ is a connected component of the preimage of $I$ under the covering map $\Rbb\rightarrow \Rbb/2\pi\Zbb\simeq\Sbb^1$. Therefore, since the universal cover $\SG$ acts on $\Rbb$ (cf. \eqref{eq2}), it also acts on $\Jtd$. Through the quotient homomorphism $\GA\rightarrow\SG$, the group $\GA\simeq\Gc$ also acts on $\Sbb^1$, $\Rbb$, and $\Jtd$.

\begin{rem}\label{lb129}
Def. \ref{lb123} can be rephrased as follows: For each $I\in\MJ$, 
\begin{align*}
&\SG(I)=\{g\in\SG:g \text{ fixes pointwise } (I',\arg_{I'})\text{ for each }\arg_{I'}\}\\
=&\{g\in\SG:g \text{ fixes pointwise } (I',\arg_{I'})\text{ for some }\arg_{I'}\}
\end{align*}
Therefore,
\begin{align*}
&\GA(I)=\{g\in\GA:g \text{ fixes pointwise } (I',\arg_{I'})\text{ for each }\arg_{I'}\}\\
=&\{g\in\GA:g \text{ fixes pointwise } (I',\arg_{I'})\text{ for some }\arg_{I'}\}
\end{align*}
\end{rem}

If $\wtd I,\wtd J\in\Jtd$, we write $\pmb{\wtd I\subset\wtd J}$ if $I\subset J$ and $\arg_J|_I=\arg_I$. We say that $\wtd I$ and $\wtd J$ are \textbf{disjoint} if $I\cap J=\emptyset$. (Note that this is not the same as saying that $\wtd I$ and $\wtd J$ are disjoint as subsets of $\Rbb$.) We say that $\wtd J$ is \textbf{clockwise} to $\wtd I$ (equivalently, that $\wtd I$ is \textbf{anticlockwise} to $\wtd J$) if
\begin{align*}
\arg \zeta<\arg z<\arg\zeta+2\pi\qquad\text{for each }z\in I,\zeta\in J
\end{align*}

\begin{df}
We say that $\wtd I_1,\dots,\wtd I_n\in\Jtd$ are in \textbf{clockwise order} if $\wtd I_j$ is clockwise to $\wtd I_i$ whenever $j>i$.
\end{df}

\begin{df}\label{lb8}
The \textbf{clockwise complement $\pmb{\wtd I'}$} (resp. the \textbf{anticlockwise complement $\pmb{\bpr\wtd I}$}) of $\wtd I\in\Jtd$ is defined to be the interval $I'$ together with the arg function so that $\wtd I'$ is clockwise to $\wtd I$ (resp. $\bpr\wtd I$ is anticlockwise to $\wtd I$).
\end{df}

\begin{df}\label{lb3}
Let $\phi\in G$. A \textbf{$\phi$-twisted $\MA$-module} denotes a pair $(\MH_i,\pi_i)$ (or simply $\MH_i$) where $\MH_i$ is a separable Hilbert space, and $\pi_i$ associates to each $\wtd I\in\Jtd$ a normal representation $\pi_{i,\wtd I}$ of $\MA(I)$ on $\MH_i$ satisfying the following properties.
\begin{enumerate}[label=(\alph*)]
\item If $\wtd I,\wtd J\in\Jtd$ and $\wtd I\subset\wtd J$, then $\pi_{i,\wtd I}(x)=\pi_{i,\wtd J}(x)$ for each $x\in\MA(I)$.
\item For each $\wtd I\in\Jtd$ and $x\in\MA(I)$, we have
\begin{align}\label{eq8}
\pi_{i,\varrho(2\pi)\wtd I}(\phi x\phi^{-1})=\pi_{i,\wtd I}(x)
\end{align}
\end{enumerate}
If condition (b) is dropped, we say that $(\MH_i,\pi_i)$ is a \textbf{solitonic $\MA$-module}.
\end{df}

\begin{rem}\label{lb138}
The above definition of a $\phi$-twisted module follows \cite[Sec. 3.1]{MS26a}, except that the unit circle used there is obtained from ours by reflection across the $x$-axis. Consequently, the convention in \cite{MS26a,MS26b} uses $\varrho(-2\pi)$ in place of $\varrho(2\pi)$ in \eqref{eq8}. 
\end{rem}

\begin{df}\label{lb132}
We let $\pmb{\Rep^\phi(\MA)}$ be the $W^*$-category of $\phi$-twisted $\MA$-modules. For the identity element $e$ of $G$ we write $\Rep^e(\MA)$ as $\pmb{\Rep(\MA)}$. Objects in $\Rep(\MA)$ are called \textbf{untwisted $\MA$-modules}, or simply \textbf{$\MA$-modules}. 
\end{df}

\begin{rem}\label{lb133}
Let $(\MH_i,\pi_i)$ be a solitonic $\MA$-module. Assume that $\wtd J$ is clockwise to $\wtd I$. Then
\begin{align*}
[\pi_{i,\wtd I}(\MA(I)),\pi_{i,\wtd J}(\MA(J))]=0
\end{align*}
\end{rem}

\begin{proof}
By the additivity of $\MA$ and the normality of $\pi_{i,\wtd I},\pi_{i,\wtd J}$, it suffices to treat the case where $\wtd I,\wtd J\subset\wtd K$ for some $\wtd K\in\Jtd$. Then the commutativity follows by the locality of $\MA$ and part (a) of Def. \ref{lb3}.
\end{proof}

\begin{eg}
The action of $\MA$ on $\MH_0$ is clearly an untwisted $\MA$-module, called the \textbf{vacuum module} and is denoted by $\pmb{(\MH_0,\pi_0)}$ or simply by $\MH_0$.
\end{eg}

\begin{rem}\label{lb11}
Suppose that $x\in\MA(I)$ commutes with $\phi\in G$. Then by Def. \ref{lb3}-(b), $\pi_{i,\wtd I}(x)$ is independent of $\arg_I$. We thus write $\pi_{i,\wtd I}(x)$ as $\pmb{\pi_{i,I}(x)}$ in this case.

For example, if $g\in\GA$, then $U_0(g)$ commutes with any $\phi\in G$; cf. \cite[Rem. 3.14]{MS26a}. Hence, if $g\in\GA(I)$, we can 
\begin{align*}
\text{write $\pi_{i,\wtd I}(U_0(g))$ as $\pi_{i,I}(U_0(g))$}
\end{align*}
As another example, if $(\MH_i,\pi_i)\in\Rep(\MA)$, then $\pi_{i,\wtd I}(x)$ can be written as $\pi_{i,I}(x)$ for each $x\in\MA(I)$. \hqed
\end{rem}

\begin{thm}\label{lb12}
Any $\phi$-twisted $\MA$-module $(\MH_i,\pi_i)$ is \textbf{conformally covariant}, which means that there is a unique strongly-continuous unitary representation $\pmb{U_i}$ of $\GA\simeq\Gc$ on $\MH_i$ satisfying for each $\wtd I\in\Jtd,g\in\GA(I)$ the relation
\begin{align}
U_i(g)=\pi_{i,I}(U_0(g))
\end{align}
Moreover, for each $g\in\GA$, we have $U_i(g)\in\bigvee_{\wtd I\in\Jtd}\pi_{i,\wtd I}(\MA(I))$.
\end{thm}

\begin{proof}
This follows from \cite{Hen19}, as explained in \cite[Thm. 2.2]{Gui21a} (applied to the untwisted $\Vir_c$-module $(\MH_i,\pi_i|_{\Vir_c})$) or in \cite[Thm. 3.15]{MS26a}. In particular, the uniqueness follows from the following fact.
\end{proof}

\begin{lm}\label{lb14}
Let $\mc I$ be a collection of intervals covering $\Sbb^1$. Then $\bigcup_{I\in\mc I}\scr G(I)$ generates algebraically the group $\scr G$. Consequently, $\bigcup_{I\in\mc I}\GA(I)$ generates algebraically the group $\GA$.
\end{lm}

\begin{proof}
This is due to \cite[Lem. 17]{Hen19}.
\end{proof}

\begin{rem}
We abbreviate $U_i(g)$ to $\pmb{U(g)}$ or even $g$ when no confusion arises.
\end{rem}

\begin{co}\label{lb30}
For each $\phi$-twisted $\MA$-module $\MH_i$, $\wtd I\in\Jtd$, $x\in\MA(I)$, and $g\in\GA$, we have
\begin{align}\label{eq9}
U_i(g)\pi_{i,\wtd I}(x)U_i(g)^{-1}=\pi_{i,g\wtd I}(U_0(g)xU_0(g)^{-1})
\end{align}
\end{co}

\begin{proof}
See \cite[Cor. 3.16]{MS26a}. The idea is to check \eqref{eq9} using Thm. \ref{lb12} when $g$ is supported in any interval whose union with $I$ is not dense, and then conclude the general case by using Lem. \ref{lb14}.
\end{proof}

\subsection{The category $\RepGA$ and the $G$-action}

Fix a conformal net $\MA$ with central charge $c$, a discrete group $G$, and a group homomorphism $G\rightarrow\Aut(\MA)$. Let
\begin{align*}
\pmb{\Rep^G(\MA)}
\end{align*}
be the $W^*$-category whose objects are finite (orthogonal) direct sums of the form $\bigoplus_i \MH_i$ where each $\MH_i\in\Rep^{\omega_i}(\MA)$ for some $\omega_i\in G$. Objects in $\RepGA$ are called \pmb{$G$-twisted $\MA$-modules}. The morphisms in $\RepGA$ are defined in Def. \ref{lb137}.

\begin{df}\label{lb137}
If $\MH_i,\MH_j\in\RepGA$, we let $\pmb{\Hom_\MA(\MH_i,\MH_j)}$ denote $G$-twisted $\MA$-module morphisms from $\MH_i$ to $\MH_j$, that is, elements $T\in\fk L(\MH_i,\MH_j)$ satisfying
\begin{align*}
T\pi_{\wtd I,i}(x)=\pi_{\wtd I,j}(x)T
\end{align*}
for each $\wtd I\in\Jtd,x\in\MA(I)$. Elements of $\Hom_\MA(\MH_i,\MH_j)$ are called \textbf{(homo)morphisms of $G$-twisted $\MA$-modules} from $\MH_i$ to $\MH_j$.
\end{df}

\begin{cv}\label{lb4}
In this paper, objects of $\RepGA$ are denoted by symbols such as $\MH_i,\MH_j,\MH_k$, etc. We abbreviate $\pi_{i,\wtd I}$ to $\pmb{\pi_{\wtd I}}$ when no confusion arises. 
\end{cv}

\subsubsection{The crossed balanced $W^*$-tensor category $\RepGA$}

\begin{df}\label{lb58}
For each $\phi\in G$ and each $G$-twisted $\MA$-module $(\MH_i,\pi_i)$, we define a $G$-twisted $\MA$-module
\begin{align*}
\pmb{(\MH_{\phi i},\pi_{\phi i})}
\end{align*}
abbreviated to $\pmb{\MH_{\phi i}}$, where $\MH_{\phi i}$ is an abstract Hilbert space unitarily equivalent to $\MH_i$ via a prescribed unitary map
\begin{align*}
\pmb{\Gamma_\phi|_{\MH_i}}:\MH_i\xlongrightarrow{\simeq} \MH_{\phi i}
\end{align*}
and $\pi_{\phi i}$ is defined such that for each $\wtd I\in\Jtd,x\in\MA(I)$, we have
\begin{align*}
\Gamma_\phi|_{\MH_i}\circ \pi_{i,\wtd I}(x)=\pi_{\phi i,\wtd I}(\phi x\phi^{-1})\circ\Gamma_\phi|_{\MH_i}
\end{align*}
We abbreviate $\Gamma_\phi|_{\MH_i}$ to $\pmb{\Gamma_\phi}$ when no confusion arises. Then the above relation can be abbreviated to
\begin{align}\label{eq4}
\Gamma_\phi\circ\pi_{\wtd I}(x)=\pi_{\wtd I}(\phi x\phi^{-1})\circ\Gamma_\phi
\end{align}
For each $\omega\in G$, it is clear that $\MH_i\in\Rep^\omega(\MA)$ implies $\MH_{\phi i}\in\Rep^{\phi\omega\phi^{-1}}(\MA)$.
\end{df}

\begin{rem}
It is obvious that the pair $(\MH_{\phi i},\Gamma_\phi|_{\MH_i})$ is uniquely determined by $\MH_i$ and $\phi$ up to unique unitary maps. That is, if $(\wtd\MH_{\phi i},\wtd\Gamma_\phi|_{\MH_i})$ satisfies the same property, then there is a unique unitary map (indeed, a unique $G$-twisted $\MA$-module morphism) $\Psi:\wtd\MH_{\phi i}\rightarrow\MH_{\phi i}$ such that $\Gamma_\phi|_{\MH_i}=\Psi\circ\wtd\Gamma_\phi|_{\MH_i}$. We call $\Psi$ the (unique) \textbf{unitary morphism $(\wtd\MH_{\phi i},\wtd\Gamma_\phi|_{\MH_i})\rightarrow(\MH_{\phi i},\Gamma_\phi|_{\MH_i})$}. 
\end{rem}

\begin{eg}\label{lb6}
For the vacuum module $\MH_0$, one can choose $(\wtd\MH_{\phi 0},\wtd\Gamma_\phi|_{\MH_0})$ to be $(\MH_0,\phi)$ where $\phi$ is understood as in $\MU(\MH_0)$. The unitary morphism $(\MH_0,\phi)\rightarrow(\MH_{\phi0},\Gamma_\phi|_{\MH_0})$ is given by $\Gamma_\phi|_{\MH_0}\circ \phi^{-1}$.
\end{eg}

\begin{rem}\label{lb59}
Recall from Rem. \ref{lb11} that $\Ad_\phi$ fixes $U_0(g)$ for each $I\in\MJ$ and $g\in\GA(I)$. Therefore, by \eqref{eq4}, $\Ad_{\Gamma_\phi}$ sends $\pi_{i,I}(U_0(g))$ to $\pi_{\phi i,I}(U_0(g))$. Thus, by the uniqueness in Thm. \ref{lb12}, we have
\begin{align}
\Gamma_\phi\circ U_i(g)=U_{\phi i}(g)\circ\Gamma_\phi
\end{align}
for each $\MH_i\in\RepGA$ and $g\in\GA$.
\end{rem}

\begin{df}
For each $\phi\in G$, define a $*$-functor $\pmb{\fk T_\phi}:\RepGA\rightarrow\RepGA$ such that
\begin{align*}
\fk T_\phi(\MH_i)=\MH_{\phi i}\qquad \fk T_\phi(S)=(\Gamma_\phi|_{\MH_j})\circ S\circ (\Gamma_\phi|_{\MH_i})^{-1}
\end{align*}
for each $S\in\Hom_\MA(\MH_i,\MH_j)$. This functor is denoted by $T_\phi$ in \cite{MS26a}.
\end{df}

The following theorem is proved in \cite[Sec. 3.6]{MS26a}. Note that, by Rem. \ref{lb138}, our balancing differs from that in \cite{MS26a,MS26b}: the latter is defined using $\varrho_\MA(-2\pi)$, whereas ours is defined using $\varrho_\MA(2\pi)$.  

\begin{thm}\label{lb5}
Suppose that $G\rightarrow\Aut(\MA)$ is faithful. Then the $W^*$-category $\RepGA$ with grading $\bigoplus_{\phi\in G}\Rep^\phi(\MA)$ is canonically a $G$-crossed balanced $W^*$-tensor category defined by Connes fusion, where the action of $G$ on $\RepGA$ is defined by $\phi\in G\mapsto \fk T_\phi$, and the balancing $\vartheta$ is defined by
\begin{align*}
\pmb{\vartheta_i}=\Gamma_\phi\circ U_i(\varrho_\MA(2\pi)):\MH_i\rightarrow\MH_{\phi i}
\end{align*}
for each $\MH_i\in\Rep^\phi(\MA)$.
\end{thm}

When $G\rightarrow\Aut(\MA)$ is not faithful, then $\RepGA$ does not admit the grading $\bigoplus_{\phi\in G}\Rep^\phi(\MA)$, since $\Hom_\MA(\MH_i,\MH_j)$ might not be zero for $\MH_i\in\Rep^\omega(\MA)$ and $\MH_j\in\Rep^\phi(\MA)$ even when $\omega\neq\phi$ in $G$. In that case, to apply Thm. \ref{lb5}, it suffices to replace $G$ with its image in $\Aut(\MA)$.

\begin{cv}\label{lb9}
Let $\boxtimes$ denote the tensor product bifunctor of $\RepGA$, and call it the \textbf{fusion product} bifunctor. For each $\MH_i,\MH_j\in\RepGA$, we write $\pmb{\MH_{i\boxtimes j}}:=\MH_i\boxtimes\MH_j$. According to the definition of $G$-crossed tensor categories, we have
\begin{align}\label{eq10}
\MH_i\in\Rep^\phi(\MA),\MH_j\in\Rep^\omega(\MA)\qquad\Longrightarrow\qquad\MH_i\boxtimes\MH_j\in\Rep^{\phi\omega}(\MA)
\end{align}
The braiding operation is denoted by  $\pmb{\Bbb}$. More precisely, for each $\MH_i\in\Rep^\phi(\MA)$ and $\MH_j\in\RepGA$, the braiding operator (which is a unitary equivalence of $G$-twisted $\MA$-modules) is denoted by
\begin{align*}
\pmb{\Bbb_{i,j}}:\MH_i\boxtimes\MH_j\rightarrow\MH_{\phi j}\boxtimes\MH_i
\end{align*}

By the coherence conditions for the $W^*$-tensor category $\RepGA$, we make the identification
\begin{align*}
(\MH_i\boxtimes\MH_j)\boxtimes\MH_k=\MH_i\boxtimes(\MH_j\boxtimes\MH_k)
\end{align*}
via the associator, and denote them by $\MH_i\boxtimes\MH_j\boxtimes\MH_k\equiv\MH_{i\boxtimes j\boxtimes k}$; we identify
\begin{align*}
\MH_0\boxtimes\MH_i=\MH_i=\MH_i\boxtimes\MH_0
\end{align*}
using the unitors. In other words, we treat $\RepGA$ as if it is strict.  \hqed
\end{cv}

\subsubsection{The compatibility conditions for the $G$-action}

In the following remark, we describe the action of $G$ on $\RepGA$ more explicitly.  This is a concrete version of the more abstract and concise description given in \cite[Def. 2.8]{MS26a}. Although we will not use this explicit description in our construction of open/closed CFTs, we feel it is worthwhile to record it here.

\begin{rem}\label{lb7}
Recall that $e$ is the unit of $G$. In addition to the assignment $\phi\mapsto\fk T_\phi$ the action involves the data and the compatibility conditions listed below.
\begin{enumerate}[label=(\alph*)]
\item A unitary isomorphism $\mbf A_i:\MH_i\simeq\MH_{ei}$ where $e\in G$ is the group identity. This is given by the unique unitary morphism $(\MH_i,\id)\rightarrow(\MH_{ei},\Gamma_e|_{\MH_i})$. So
\begin{align*}
\mbf A_i=\Gamma_e|_{\MH_i}
\end{align*}
The equivalence $\mbf A:\id_{\RepGA}\rightarrow\fk T_e$ is natural, i.e., we have a commutative diagram
\begin{equation*}
\begin{tikzcd}
\MH_i \arrow[r,"T"] \arrow[d,"\mbf A_i"'] & \MH_j \arrow[d,"\mbf A_j"] \\
\MH_{ei} \arrow[r,"\fk T_e(T)"]           & \MH_{ej}       
\end{tikzcd}
\end{equation*}
for each $T\in\Hom_\MA(\MH_i,\MH_j)$, which is straightforward to check.
\item A unitary isomorphism $\mbf B_\phi:\MH_0\rightarrow\MH_{\phi0}$ for each $\phi\in G$, which is given by the unique unitary morphism $(\MH_0,\phi)\rightarrow(\MH_{\phi0},\Gamma_\phi|_{\MH_0})$ discussed in Exp. \ref{lb6}. So
\begin{align*}
\mbf B_\phi=\Gamma_\phi\circ\phi^{-1}|_{\MH_0}
\end{align*}
\item A unitary natural isomorphism $\fk T_\phi\circ\fk T_\psi\rightarrow\fk T_{\phi\circ\psi}$ for $\phi,\psi\in G$, where, for each $\MH_i\in\RepGA$, the unitary equivalence 
\begin{align*}
\mbf C_{\phi,\psi,i}:\MH_{\phi(\psi i)}\rightarrow \MH_{(\phi\psi)i}\qquad \mbf C_{\phi,\psi,i}=\Gamma_{\phi\psi}\Gamma_\psi^{-1}\Gamma_\phi^{-1}|_{\MH_{\phi(\psi i)}}
\end{align*}
is given by the unique unitary morphism $(\MH_{\phi(\psi i)},\Gamma_\phi\circ\Gamma_\psi|_{\MH_i})\rightarrow (\MH_{(\phi\psi)i},\Gamma_{\phi\psi}|_{\MH_i})$.

The naturality means that for each $T\in\Hom_\MA(\MH_i,\MH_j)$,
\begin{subequations}\label{eq71}
\begin{gather}\label{eq71a}
\begin{tikzcd}[column sep=huge,row sep=large,ampersand replacement=\&]
\MH_{\phi(\psi i)} \arrow[r,"\fk T_\phi\circ\fk T_\psi(T)"] \arrow[d,"\mbf C_{\phi,\psi,i}"'] \& \MH_{\phi(\psi j)} \arrow[d,"\mbf C_{\phi,\psi,j}"] \\
\MH_{(\phi\psi)i} \arrow[r,"\fk T_{\phi\psi}(T)"]           \& \MH_{(\phi\psi)j}       
\end{tikzcd}
\end{gather}
Moreover, the tuple $(\fk T,\mbf C,\mbf A)$ is a tensor functor, i.e., $\mbf C$ is a tensorator and $\mbf A$ is a unitor. This means the commutativity of
\begin{gather}
\begin{tikzcd}[column sep=huge,row sep=large,ampersand replacement=\&]
\MH_{\phi(\psi(\omega i))} \arrow[r,"\fk T_\phi(\mbf C_{\psi,\omega, i})"] \arrow[d,"\mbf C_{\phi,\psi,\omega i}"'] \& \MH_{\phi((\psi\omega)i)} \arrow[d,"\mbf C_{\phi,\psi\omega,i}"] \\
\MH_{(\phi\psi)(\omega i)} \arrow[r,"\mbf C_{\phi\psi,\omega,i}"]           \& \MH_{(\phi\psi\omega)i}       
\end{tikzcd}\\
\begin{tikzcd}[column sep=large,ampersand replacement=\&]
\MH_{e(\phi i)} \arrow[r, bend left,"\mbf C_{e,\phi,i}"] \& \MH_{\phi i} \arrow[l, bend left,"\mbf A_{\phi i}"]
\end{tikzcd}
\qquad
\begin{tikzcd}[column sep=large,ampersand replacement=\&]
\MH_{\phi(ei)} \arrow[r, bend left,"\mbf C_{\phi,e,i}"] \& \MH_{\phi i} \arrow[l, bend left,"\fk T_\phi(\mbf A_i)"]
\end{tikzcd}
\end{gather}
\end{subequations}
The commutativity of these diagrams is easy to check.\footnote{In particular, in the first two diagrams, going from the upper-left corner to the lower-right corner along either path yields, respectively, $\Gamma_{\phi\psi}T\Gamma_\psi^{-1}\Gamma_\phi^{-1}$ and $\Gamma_{\phi\psi\omega}\Gamma_\omega^{-1}\Gamma_\psi^{-1}\Gamma_\phi^{-1}$.} 
\item A unitary (natural) tensorator $\fk T_\phi(-)\boxtimes \fk T_\phi(-)\rightarrow \fk T_\phi(-\boxtimes-)$ for each $\phi\in G$, where, for each $\MH_i,\MH_j\in\RepGA$, the unitary equivalence
\begin{align*}
\mbf D_{\phi,i,j}:\MH_{\phi i}\boxtimes\MH_{\phi j}\rightarrow \MH_{\phi(i\boxtimes j)} 
\end{align*}
will be described in Def. \ref{lb16}. Moreover, $\mbf B_\phi$ is a unitor. (So $(\fk T_\phi,\mbf D_\phi,\mbf B_\phi)$ gives a $W^*$-tensor automorphism of $\RepGA$.)

Unraveling the phrases ``natural'', ``tensorator'', and ``unitor'' gives the following commutative diagrams, where $S\in\Hom_\MA(\MH_i,\MH_{\wtd i})$ and $T\in\Hom_\MA(\MH_j,\MH_{\wtd j})$, and $\fk u_{0,\phi i}:\MH_0\boxtimes\MH_{\phi i}\rightarrow\MH_{\phi i}$, $\fk u_{0,i}:\MH_0\boxtimes\MH_i\rightarrow\MH_i$, $\fk u_{\phi i,0}:\MH_{\phi i}\boxtimes \MH_0\rightarrow\MH_{\phi i}$, and $\fk u_{i,0}:\MH_i\boxtimes\MH_0\rightarrow\MH_i$ are the unitors.
\begin{subequations}\label{eq12}
\begin{gather}
\begin{tikzcd}[column sep=huge,row sep=large,ampersand replacement=\&]
\MH_{\phi i}\boxtimes\MH_{\phi j} \arrow[r,"\fk T_\phi(S)\boxtimes\fk T_\phi(T)"] \arrow[d,"\mbf D_{\phi,i,j}"'] \& \MH_{\phi \wtd i}\boxtimes\MH_{\phi \wtd j} \arrow[d,"\mbf D_{\phi,\wtd i,\wtd j}"] \\
\MH_{\phi(i\boxtimes j)}  \arrow[r,"\fk T_\phi(S\boxtimes T)"]           \& \MH_{\phi(\wtd i\boxtimes \wtd j)}        
\end{tikzcd}\label{eq12a}\\
\begin{tikzcd}[column sep=huge,row sep=large,ampersand replacement=\&]
\MH_{\phi i}\boxtimes\MH_{\phi j}\boxtimes\MH_{\phi k} \arrow[r,"\idt\boxtimes \mbf D_{\phi,j,k}"] \arrow[d,"\mbf D_{\phi,i,j}\boxtimes\idt"'] \& \MH_{\phi i}\boxtimes\MH_{\phi (j\boxtimes k)} \arrow[d,"\mbf D_{\phi,i,j\boxtimes k}"] \\
\MH_{\phi(i\boxtimes j)}\boxtimes\MH_{\phi k}  \arrow[r,"\mbf D_{\phi,i\boxtimes j,k}"]           \& \MH_{\phi(i\boxtimes j\boxtimes k)}        
\end{tikzcd}\label{eq12b}\\
\begin{tikzcd}[column sep=huge,row sep=large,ampersand replacement=\&]
\MH_0\boxtimes\MH_{\phi i} \arrow[r,"\fk u_{0,\phi i}"] \arrow[d,"\mbf B_\phi\boxtimes\idt"'] \& \MH_{\phi i}  \\
\MH_{\phi 0}\boxtimes\MH_{\phi i}  \arrow[r,"\mbf D_{\phi,0,i}"]           \& \MH_{\phi(0\boxtimes i)}        \arrow[u,"\fk T_\phi(\fk u_{0,i})"']
\end{tikzcd}\label{eq12c}\\
\begin{tikzcd}[column sep=huge,row sep=large,ampersand replacement=\&]
\MH_{\phi i} \boxtimes\MH_0\arrow[r,"\fk u_{\phi i,0}"] \arrow[d,"\idt\boxtimes\mbf B_\phi"'] \& \MH_{\phi i}  \\
\MH_{\phi i}\boxtimes\MH_{\phi 0}  \arrow[r,"\mbf D_{\phi,i,0}"]           \& \MH_{\phi(i\boxtimes 0)}        \arrow[u,"\fk T_\phi(\fk u_{i,0})"']
\end{tikzcd}\label{eq12d}
\end{gather}
\end{subequations}
\item The natural isomorphism $\mbf A:\id_{\RepGA}\rightarrow\fk T_e$ is tensor natural if the tensorator and unitor of $\id_{\RepGA}$ are chosen to be the identity, and those of $\fk T_e$ are chosen as in (d). This means that the following diagrams commute:
\begin{subequations}\label{eq73}
\begin{gather}
\begin{tikzcd}[column sep=huge,row sep=large,ampersand replacement=\&]
\MH_i\boxtimes\MH_j \arrow[r,"\mbf A_i\boxtimes\mbf A_j"] \arrow[d,"="'] \& \MH_{ei}\boxtimes\MH_{ej}\arrow[d,"\mbf D_{e,i,j}"']  \\
\MH_{i\boxtimes j}\  \arrow[r,"\mbf A_{i\boxtimes j}"]           \& \MH_{e(i\boxtimes j)}        
\end{tikzcd}\label{eq73a}\\
\begin{tikzcd}[column sep=large,ampersand replacement=\&]
\MH_0 \arrow[r, bend left,"\mbf B_e"]\arrow[r, bend right,"\mbf A_0"'] \& \MH_{e0} 
\end{tikzcd}\label{eq73b}
\end{gather}
\end{subequations}
\item The natural isomorphism $\mbf C_{\phi,\psi}:\fk T_\phi\circ\fk T_\psi\rightarrow\fk T_{\phi\circ\psi}$ is tensor natural if the tensorators and the unitors of $\fk T_\phi,\fk T_\psi,\fk T_{\phi\circ\psi}$ are chosen as in (d). (In particular, the tensorator and unitor of $\fk T_\phi\circ\fk T_\psi$ are $\fk T_\phi(\mbf D_\psi)\circ\mbf D_\phi$ and $\fk T_\phi(\mbf B_\psi)\circ\mbf B_\phi$.) This means that the following diagrams commute:
\begin{subequations}\label{eq72}
\begin{gather}
\begin{tikzcd}[column sep=3cm,row sep=large,ampersand replacement=\&]
\MH_{\phi(\psi i)} \boxtimes\MH_{\phi(\psi j)}\arrow[r,"\mbf C_{\phi,\psi,i}\boxtimes\mbf C_{\phi,\psi,j}"] \arrow[d,"\mbf D_{\phi,\psi i,\psi j}"'] \& \MH_{(\phi\psi) i} \boxtimes\MH_{(\phi\psi)j}\arrow[dd,"\mbf D_{\phi\psi,i,j}"] \\
\MH_{\phi(\psi i\boxtimes\psi j)}\arrow[d,"\fk T_\phi(\mbf D_{\psi,i,j})"']\&\\
\MH_{\phi(\psi(i\boxtimes j))} \arrow[r,"\mbf C_{\phi,\psi,i\boxtimes j}"]           \& \MH_{(\phi\psi)(i\boxtimes j)}        
\end{tikzcd}\label{eq72a}\\
\begin{tikzcd}[column sep=huge,row sep=large,ampersand replacement=\&]
\MH_0 \arrow[r,"\mbf B_{\phi\psi}"] \arrow[d,"\mbf B_\phi"'] \& \MH_{(\phi\psi)0}  \\
\MH_{\phi 0}\  \arrow[r,"\fk T_\phi(\mbf B_\psi)"]           \& \MH_{\phi(\psi 0)}        \arrow[u,"\mbf C_{\phi,\psi,0}"']
\end{tikzcd}\label{eq72b}
\end{gather}
\end{subequations}
\end{enumerate}
\end{rem}

The above compatibility conditions in Rem. \ref{lb7} allow the $G$-action to be treated as strict, in the following sense.

\begin{cv}\label{lb10}
Unless otherwise stated, for $\MH_i,\MH_j\in\RepGA$ and $\phi,\psi\in G$, we make the identifications
\begin{subequations}\label{eq5}
\begin{gather}
\MH_i=\MH_{ei}\qquad \MH_0=\MH_{\phi 0}\qquad \MH_{\phi(\psi i)}=\MH_{(\phi\psi) i}\label{eq5a}\\
\MH_{\phi i\boxtimes\phi j}=\MH_{\phi(i\boxtimes j)} \label{eq5b}
\end{gather}
\end{subequations}
using the unitary isomorphisms $\mbf A,\mbf B,\mbf C,\mbf D$ mentioned in Rem. \ref{lb7}. Thus, for $\phi_1,\dots,\phi_n\in G$, we interpret $\MH_{\phi_1\cdots\phi_n i}$ with parentheses inserted arbitrarily. Under these identifications, by (a,b,c) of Rem. \ref{lb7}, we have
\begin{align*}
\Gamma_e|_{\MH_i}=\id_{\MH_i}\qquad \Gamma_\phi|_{\MH_0}=\phi\qquad \Gamma_{\phi\psi}|_{\MH_i}=\Gamma_\phi\circ\Gamma_\psi|_{\MH_i}
\end{align*}
which we abbreviate as
\begin{align}\label{eq6}
\Gamma_e=\id\qquad\Gamma_\phi|_{\MH_0}=\phi\qquad  \Gamma_{\phi\psi}=\Gamma_\phi\circ\Gamma_\psi
\end{align}
In other words, the assignment $\phi\in G\mapsto \Gamma_\phi$ defines a ``categorical group representation'' $G\curvearrowright\RepGA$ extending $G\curvearrowright\MH_0$.
\end{cv}

As a consequence of \eqref{eq6}, we have
\begin{align}\label{eq7}
\Gamma_\phi^{-1}=\Gamma_{\phi^{-1}}
\end{align}
as unitary maps from each $\MH_i$ to $\MH_{\phi^{-1}i}$.

\subsection{$G$-crossed categorical extensions}

Let $\MA$ be a conformal net and $G\rightarrow\Aut(\MA)$ a group homomorphism. In this section, we recall the $G$-crossed categorical extension for $\MA$ established in \cite{MS26a,MS26b}. Recall Def. \ref{lb8} for the clockwise and anticlockwise complements $\wtd I'$ and $\bpr\wtd I$.

\begin{df}
For each $\MH_i,\MH_j\in\RepGA$ and $\wtd I\in\Jtd$, let
\begin{gather*}
\pmb{\Hom_{\MA(\wtd I)}(\MH_i,\MH_j)}=\big\{T\in\fk L(\MH_i,\MH_j):T\pi_{i,\wtd I}(x)=\pi_{j,\wtd I}(x)T\text{ for each }x\in\MA(I)\big\}\\
\pmb{\MH_j(I)}=\Hom_{\MA(\wtd I')}(\MH_0,\MH_j)\cdot\Omega
\end{gather*}
where $\MH_j(I)$ is independent of $\arg_I$. Note that the Reeh-Schlieder property implies that $\MH_j(I)$ is dense in $\MH_j$. We write
\begin{align*}
\pmb{\End_{\MA(\wtd I)}(\MH_i)}=\Hom_{\MA(\wtd I)}(\MH_i,\MH_i)
\end{align*}
We write $\Hom_{\MA(\wtd I)}(\MH_i,\MH_j)$ as $\pmb{\Hom_{\MA(I)}(\MH_i,\MH_j)}$ when $\MH_i,\MH_j\in\RepA$.
\end{df}

\begin{proof}[Explanation]
We need to explain why $\MH_j(I)$ is independent of $\arg_I$. It suffices to show that for each $\wtd I\in\Jtd$, the space $\MH_j(I)$ defined by $\wtd I$ coincides with the one defined by $\varrho(2\pi)\wtd I$, that is,
\begin{align*}
\Hom_{\MA(\wtd I')}(\MH_0,\MH_j)\Omega=\Hom_{\MA(\bpr\wtd I)}(\MH_0,\MH_j)\Omega
\end{align*}
Assume without loss of generality that $\MH_j\in\Rep^\phi(\MA)$. Then the above relation follows from the following Lem. \ref{lb17} (due to \cite[Eq. (6)]{MS26a}) and the fact that $\phi\Omega=\Omega$.
\end{proof}

\begin{lm}\label{lb17}
Assume that $\phi\in G$ and $\MH_j\in\Rep^\phi(\MA)$. Let $T\in\Hom_{\MA(\bpr\wtd I)}(\MH_0,\MH_j)$. Then $T\phi\in\Hom_{\MA(\wtd I')}(\MH_0,\MH_j)$.
\end{lm}

Consequently, if $\xi\in\MH_j(I)$, and if $S\in\Hom_{\MA(\wtd I')}(\MH_0,\MH_j)$ and $T\in\Hom_{\MA(\bpr\wtd I)}(\MH_0,\MH_j)$ are the unique maps such that $S\Omega=\xi=T\Omega$, then $S=T\phi$. (The uniqueness is due to the Reeh-Schlieder property for $\MA$.)

\begin{proof}
By \eqref{eq8}, for each $x\in\MA(I')$ we have
\begin{align*}
\pi_{j,\wtd I'}(x)T\phi=\pi_{j,\bpr\wtd I}(\phi x\phi^{-1})T\phi=T\phi x\phi^{-1}\phi=T\phi x
\end{align*}
Thus $T\phi\in\Hom_{\MA(\wtd I')}(\MH_0,\MH_j)$. 
\end{proof}

\begin{rem}\label{lb43}
Note that Haag duality implies $\MH_0(I)=\MA(I)\Omega$. Note also that if $T\in\Hom_{\MA(\wtd I')}(\MH_i,\MH_j)$, then
\begin{align*}
T\MH_i(I)\subset\MH_j(I)
\end{align*}
\end{rem}

\begin{rem}\label{lb13}
By \eqref{eq4}, we have
\begin{align}
\Gamma_\phi\Hom_{\MA(\wtd I)}(\MH_i,\MH_j)\Gamma_\phi^{-1}=\Hom_{\MA(\wtd I)}(\MH_{\phi i},\MH_{\phi j})
\end{align}
In particular, noting that $\Gamma_\phi^{-1}|_{\MH_0}$ equals $\phi^{-1}$ (due to \eqref{eq6}) and hence fixes $\Omega$, we obtain\footnote{An alternative proof of \eqref{eq74}, which does not assume Conv. \ref{lb74}, is to apply $\Gamma_\phi\Hom_{\MA(\wtd I)}(\MH_0,\MH_j)\phi^{-1}=\Hom_{\MA(\wtd I)}(\MH_0,\MH_{\phi j})$.}
\begin{align}\label{eq74}
\Gamma_\phi\big(\MH_j(I)\big)=\MH_{\phi j}(I)
\end{align}
By Cor. \ref{lb30}, for each $g\in\GA$ we have
\begin{align}\label{eq76}
g\Hom_{\MA(\wtd I)}(\MH_i,\MH_j) g^{-1}=\Hom_{\MA(g\wtd I)}(\MH_i,\MH_j)
\end{align}
\end{rem}

\begin{df}
Let $A:\mc P\rightarrow\mc R$, $B:\mc Q\rightarrow\mc S$, $C:\mc P\rightarrow \mc Q$, $D:\mc R\rightarrow\mc S$ be bounded linear operators between Hilbert spaces. We say that the diagram
\begin{equation*}
\begin{tikzcd}
\mc P \arrow[r,"C"] \arrow[d,"A"'] & \mc Q \arrow[d,"B"] \\
\mc R \arrow[r,"D"]           & \mc S        
\end{tikzcd}
\end{equation*}
\textbf{commutes adjointly} if $DA=BC$ and $D^*B=AC^*$.
\end{df}

The following theorem continues to assume Conv. \ref{lb9} but not Conv. \ref{lb10}.

\begin{thm}\label{lb15}
For each $\wtd I\in\Jtd$, $\MH_i\in\RepGA$,  $\xi\in\MH_i (I)$, we have bounded linear operators
\begin{gather}\label{eq68}
\begin{gathered}
L(\xi,\wtd I)\big|_{\MH_k}\in\Hom_{\MA(\wtd I')}(\MH_k,\MH_i\boxtimes\MH_k)\\
R(\xi,\wtd I)\big|_{\MH_k}\in\Hom_{\MA(\bpr\wtd I)}(\MH_k,\MH_k\boxtimes\MH_i)
\end{gathered}
\end{gather}
associated to each $\MH_k\in\RepGA$, called the \textbf{L operators} and the \textbf{R operators} and abbreviated to $L(\xi,\wtd I),R(\xi,\wtd I)$ when no confusion arises, satisfying the following conditions:
\begin{enumerate}[label=(\arabic*)]
\item (Naturality) If $\MH_k,\MH_{k'}\in\RepGA$,  $T\in\Hom_{\MA}(\MH_k,\MH_{k'})$,  $\xi\in\MH_i(I)$, and $\chi\in\MH_k$, then
\begin{align}
	(\idt_i\boxtimes T)L(\xi,\wtd I)\chi=L(\xi,\wtd I)T\chi\qquad (T\boxtimes \idt_i)R(\xi,\wtd I)\chi=R(\xi,\wtd I)T\chi
\end{align}
\item (State-field correspondence) If $\xi\in\MH_i(I)$, then (under the identifications $\MH_i=\MH_0\boxtimes\MH_i=\MH_i\boxtimes\MH_0$) we have
\begin{align}
	L(\xi,\wtd I)\Omega=R(\xi,\wtd I)\Omega=\xi
\end{align}
\item (Density of fusion products) The sets 
\begin{align*}
L(\MH_i(I),\wtd I)\MH_k\qquad  R(\MH_i(I),\wtd I)\MH_k
\end{align*}
span dense subspaces of $\MH_i\boxtimes\MH_k$ and $\MH_k\boxtimes\MH_i$, respectively.\footnote{Indeed, they are equal to the full space $\MH_i\boxtimes\MH_k$ and $\MH_k\boxtimes\MH_i$ respectively by the fact that $\MA(I)$ is a type III factor.}
\item (Locality) For any $\MH_k\in\RepGA$, any $\wtd I,\wtd J\in\Jtd$ with $\wtd J$ clockwise to $\wtd I$, and any $\xi\in\MH_i(I),\eta\in\MH_j(J)$, the following diagram commutes adjointly.
\begin{equation}
\begin{tikzcd}
\quad \MH_k\quad \arrow[rr,"{R(\eta,\wtd J)}"] \arrow[d, "{L(\xi,\wtd I)}"'] &&\quad \MH_k\boxtimes\MH_j\quad \arrow[d, "{L(\xi,\wtd I)}"]\\
\MH_i\boxtimes\MH_k\arrow[rr,"{R(\eta,\wtd J)}"] &&\MH_i\boxtimes\MH_k\boxtimes\MH_j
\end{tikzcd}
\end{equation}
\item (Braiding) If $\phi\in G$ and $\MH_i\in\Rep^\phi(\MA)$, the unitary isomorphism $\Bbb_{i,j}:\MH_i\boxtimes\MH_j\rightarrow\MH_{\phi j}\boxtimes \MH_i$ satisfies
\begin{align}
\Bbb_{i,j} L(\xi,\wtd I)\eta=R(\xi,\wtd I)\Gamma_\phi\eta
\end{align}
whenever $\xi\in\MH_i(I)$ and $\eta\in\MH_j$.
\item (Conformal covariance) If $g\in\GA$ and $\xi\in\MH_i(I)$, there exists a (necessarily unique) element of $\MH_i(gI)$, denoted by $g\xi g^{-1}$, satisfying
\begin{align}\label{eq14}
L(g\xi g^{-1},g\wtd I)=gL(\xi,\wtd I)g^{-1}\qquad R(g\xi g^{-1},g\wtd I)=gR(\xi,\wtd I)g^{-1}
\end{align}
when acting on any $\MH_j\in\RepGA$.
\end{enumerate}
Moreover, assuming the identification \eqref{eq5b}, we also have:
\begin{enumerate}
\item[(7)] ($G$-covariance) For each $\xi\in\MH_i(I)$ and $\phi\in G$, noting that $\Gamma_\phi\xi\in\MH_{\phi i}(I)$ (cf. Rem. \ref{lb13}), we have
\begin{align}\label{eq75}
\Gamma_\phi L(\xi,\wtd I)=L(\Gamma_\phi\xi,\wtd I)\Gamma_\phi\qquad \Gamma_\phi R(\xi,\wtd I)=R(\Gamma_\phi\xi,\wtd I)\Gamma_\phi
\end{align}
when acting on any $\MH_j\in\RepGA$ (with targets in $\MH_{\phi(i\boxtimes j)}=\MH_{\phi i\boxtimes\phi j}$ and $\MH_{\phi(j\boxtimes i)}=\MH_{\phi j\boxtimes\phi i}$, respectively).
\end{enumerate}
\end{thm}

When we wish to emphasize the dependence on $\MA$, we write the operations $L,R$ as $L^\MA,R^\MA$.

\begin{proof}
See Def. 3.1 and Thm. 3.4 of \cite{MS26b}.\footnote{Due to the difference in conventions mentioned in Rem. \ref{lb138}, in \cite{MS26a,MS26b}, $L(\xi,\wtd I)$ intertwines $\MA(\bpr\wtd I)$ rather than $\MA(\wtd I')$, while $R(\xi,\wtd I)$ intertwines $\MA(\wtd I')$ rather than $\MA(\bpr\wtd I)$.} See also Ch. \ref{lb140} for the proofs of (5)--(7). The proofs of (1)-(4) are similar to those in the case of trivial $G$, as treated in Ch. 2 and Sec. 3.2 of \cite{Gui21a}.
\end{proof}

\begin{rem}
Note that the fact that the braiding operation $\Bbb$ in (5) satisfies the axioms of a $G$-crossed braided $W^*$-tensor category is itself a consequence of properties (1)--(5) (together with the results of Subsec. \ref{lb141}). See the proof of Thm. 3.25 in \cite{MS26a}; see also \cite[Sec. 3.3]{Gui21a} for the special case in which $G$ is trivial.
\end{rem}

\begin{rem}\label{lb32}
For each $g\in\GA$, the uniqueness of $g\xi g^{-1}$ in the conformal covariance is due to the state-field correspondence, which implies
\begin{align}
g\xi g^{-1}=gL(\xi,\wtd I)g^{-1}\Omega=gR(\xi,\wtd I)g^{-1}\Omega
\end{align}
Applying the state-field correspondence again, we obtain
\begin{align*}
g\xi g^{-1}=g\xi\qquad\text{if }g^{-1}\Omega=\Omega
\end{align*}
For example, by the definition of conformal nets, for each $t\in\Rbb$ we have $U_0(\varrho_\MA(t))\Omega=\Omega$, and hence the \textbf{rotation covariance} property: We have
\begin{align}
\varrho_\MA(t)\MH_i(I)=\MH_i(\varrho(t)I)
\end{align}
and for each $\xi\in\MH_i(I)$ we have
\begin{subequations}
\begin{gather}
\varrho_\MA(t)L(\xi,\wtd I)=L(\varrho_\MA(t)\xi,\varrho(t)\wtd I)\varrho_\MA(t)\\
\varrho_\MA(t)R(\xi,\wtd I)=R(\varrho_\MA(t)\xi,\varrho(t)\wtd I)\varrho_\MA(t)
\end{gather}
\end{subequations}
\end{rem}

We are now ready to describe the unitary isomorphism  mentioned in Rem. \ref{lb7}-(d).

\begin{df}\label{lb16}
The linear map $\mbf D_{\phi,i,j}:\MH_{\phi i}\boxtimes\MH_{\phi j}\rightarrow \MH_{\phi(i\boxtimes j)}$ is the unique unitary map (indeed, the unique unitary isomorphism of $G$-twisted $\MA$-modules) such that
\begin{align}\label{eq11}
\Gamma_\phi L(\xi,\wtd I)\eta=\mbf D_{\phi,i,j}\circ L(\Gamma_\phi\xi,\wtd I)\Gamma_\phi\eta\qquad \Gamma_\phi R(\eta,\wtd J)\xi=\mbf D_{\phi,i,j}\circ R(\Gamma_\phi\eta,\wtd J)\Gamma_\phi\xi
\end{align} 
holds for each $\wtd I,\wtd J\in\Jtd$ with $\wtd J$ clockwise to $\wtd I$, and each $\xi\in\MH_i(I),\eta\in\MH_j(J)$. The uniqueness follows from the density of fusion products in Thm. \ref{lb15}.
\end{df}

The existence of such a unitary isomorphism, as well as the compatibility conditions in (d)--(f) Rem. \ref{lb7}, will be explained in Ch. \ref{lb140}. The $G$-covariance property in Thm. \ref{lb15} would then follow automatically.

\subsection{Calculus in $G$-crossed categorical extensions}

We continue to fix a group homomorphism $G\rightarrow\Aut(\MA)$ for a conformal net $\MA$. In this section, we do not assume Convention \ref{lb10}.

\subsubsection{Eq. \eqref{eq68} and conditions (2)--(4) of Thm. \ref{lb15} imply $\Rep^\phi(\MA)\boxtimes\Rep^\omega(\MA)\subset\Rep^{\phi\omega}(\MA)$}

In \cite{MS26a}, the fusion product of two objects in $\RepGA$ is first constructed as a solitonic $\MA$-module using operators $Z^+,Z^-$ and path continuations. For each $\MH_i,\MH_j\in\RepGA$, $\wtd I,\wtd J\in\Jtd$, and $\xi\in\MH_i(I),\eta\in\MH_j(J)$, the operators
\begin{gather*}
L(\xi,\wtd I)\big|_{\MH_0}\in\Hom_{\MA(\wtd I')}(\MH_0,\MH_i)\qquad L(\xi,\wtd I)|_{\MH_j}\in\Hom_{\MA(\wtd I')}(\MH_j,\MH_i\boxtimes\MH_j)\\
R(\eta,\wtd J)\big|_{\MH_0}\in\Hom_{\MA(\bpr\wtd J)}(\MH_0,\MH_j)\qquad R(\eta,\wtd J)|_{\MH_i}\in\Hom_{\MA(\bpr\wtd J)}(\MH_i,\MH_i\boxtimes\MH_j)
\end{gather*}
are likewise constructed using $Z^+(\xi,\wtd I),Z^-(\eta,\wtd J)$ and path continuations. (See \cite{MS26b} for a review of this construction.) The following properties are readily verified:
\begin{enumerate}
\item[(2)] $L(\xi,\wtd I)\Omega=\xi$ and $R(\eta,\wtd J)\Omega=\eta$. 
\item[(3)] Both $L(\MH_i(I),\wtd I)\MH_j$ and $R(\MH_j(J),\wtd J)\MH_i$ span dense subspaces of $\MH_i\boxtimes\MH_j$.
\item[(4)] If $\wtd J$ is clockwise to $\wtd I$, the following diagram commutes adjointly.
\begin{equation*}
\begin{tikzcd}
\quad \MH_0\quad \arrow[rr,"{R(\eta,\wtd J)|_{\MH_0}}"] \arrow[d, "{L(\xi,\wtd I)|_{\MH_0}}"'] &&\quad \MH_j\quad \arrow[d, "{L(\xi,\wtd I)|_{\MH_j}}"]\\
\MH_i\arrow[rr,"{R(\eta,\wtd J)|_{\MH_i}}"] &&\MH_i\boxtimes\MH_j
\end{tikzcd}
\end{equation*}
\end{enumerate}

These constructions and properties are, of course, part of those stated in Thm. \ref{lb15}. In what follows, we show that they already suffice to imply that the fusion of an $\phi$-twisted module and an $\omega$-twisted module is $\phi\omega$-twisted, a fact that was originally proved in \cite{MS26a} using path continuation.

\begin{lm}\label{lb19}
If $\MH_i\in\Rep^\phi(\MA)$ and $\MH_j\in\Rep^\omega(\MA)$, and if $\wtd J$ is clockwise to $\wtd I$, then for each $\xi\in\MH_i(I),\eta\in\MH_j(J)$ and each $x\in\MA(I),y\in\MA(J)$, we have
\begin{gather*}
R(\xi,\wtd I)y\big|_{\MH_0}=\pi_{\wtd J}(\phi^{-1}y\phi)R(\xi,\wtd I)\big|_{\MH_0}\\
L(\eta,\wtd J)x\big|_{\MH_0}=\pi_{\wtd I}(\omega x\omega^{-1})L(\eta,\wtd J)\big|_{\MH_0}
\end{gather*}
\end{lm}

\begin{proof}
Since $\wtd I$ is clockwise to $\varrho(2\pi)\wtd J$, we have
\begin{align*}
R(\xi,\wtd I)y\big|_{\MH_0}=\pi_{i,\varrho(2\pi)\wtd J}(y)R(\xi,\wtd I)\big|_{\MH_0}=\pi_{i,\wtd J}(\phi^{-1}y\phi)R(\xi,\wtd I)\big|_{\MH_0}
\end{align*}
where the last identity is due to the fact that $\MH_i$ is $\phi$-twisted. This proves the first identity. The second one follows from a similar argument, using the relation $\pi_{j,\varrho(-2\pi)\wtd I}(x)=\pi_{j,\wtd I}(\omega x\omega^{-1})$.
\end{proof}

\begin{thm}
If $\MH_i\in\Rep^\phi(\MA)$ and $\MH_j\in\Rep^\omega(\MA)$, then the solitonic $\MA$-module $\MH_i\boxtimes\MH_j$ is $\phi\omega$-twisted.
\end{thm}

\begin{proof}
Choose any $\wtd K^+\in\Jtd$, and let $\wtd K^-=\varrho(-2\pi)\wtd K^+$. Choose $\wtd I,\wtd J\in\Jtd$ such that $\wtd K^+,\wtd I,\wtd J,\wtd K^-$ are in clockwise order.\footnote{For example, if $\wtd K^+$ is a small neighborhood of $e^{\im\pi}$, one can choose $\wtd I$ to be a small neighborhood of $e^{\im\pi/2}$, and $\wtd J$ a small neighborhood of $e^{-\im\pi/2}$. Then $\wtd K^-$ is a small neighborhood of $e^{-\im\pi}$.}

Choose any $x\in\MA(K),\xi\in\MH_i(I),\eta\in\MH_j(J)$. By Lem. \ref{lb19}, we have
\begin{align*}
L(\xi,\wtd I)R(\eta,\wtd J)x\Omega=L(\xi,\wtd I)\pi_{\wtd K^-}(\omega^{-1}x\omega)R(\eta,\wtd J)\Omega=\pi_{\wtd K^-}(\omega^{-1}x\omega)L(\xi,\wtd I)R(\eta,\wtd J)\Omega
\end{align*}
and
\begin{align*}
R(\eta,\wtd J)L(\xi,\wtd I)x\Omega=R(\eta,\wtd J)\pi_{\wtd K^+}(\phi x\phi^{-1})L(\xi,\wtd I)\Omega=\pi_{\wtd K^+}(\phi x\phi^{-1})R(\eta,\wtd J)L(\xi,\wtd I)\Omega
\end{align*}
Therefore, by (4), the operators $\pi_{\wtd K^-}(\omega^{-1}x\omega)$ and $\pi_{\wtd K^+}(\phi x\phi^{-1})$ are equal when acting on vectors of the form $L(\xi,\wtd I)R(\eta,\wtd J)\Omega$ (which equals $L(\xi,\wtd I)\eta$ by (2)). Thus (3) implies $\pi_{\wtd K^-}(\omega^{-1}x\omega)=\pi_{\wtd K^+}(\phi x\phi^{-1})$ and hence
\begin{align*}
\pi_{\wtd K^-}(x)=\pi_{\wtd K^+}(\phi\omega x\omega^{-1}\phi^{-1})
\end{align*}
when acting on $\MH_i\boxtimes\MH_j$. Thus $\MH_i\boxtimes\MH_j$ is $\phi\omega$-twisted.
\end{proof}

\subsubsection{Consequences of Eq. \eqref{eq68} and (1)--(4) of Thm. \ref{lb15}}\label{lb141}

We now derive several consequences of properties (1)--(4) of Thm. \ref{lb15} that will be used throughout this paper.

\begin{pp}\label{lb20}
Let $\wtd J$ be clockwise to $\wtd I$, and let $\MH_i,\MH_j\in\RepGA$, $\xi\in\MH_i(I)$, $\eta\in\MH_j(J)$. Then
\begin{align*}
L(\xi,\wtd I)\eta=R(\eta,\wtd J)\xi
\end{align*}
\end{pp}

\begin{proof}
By the state-field correspondence and the locality in Thm. \ref{lb15}, we have
\begin{align*}
L(\xi,\wtd I)\eta=L(\xi,\wtd I)R(\eta,\wtd J)\Omega=R(\eta,\wtd J)L(\xi,\wtd I)\Omega=R(\eta,\wtd J)\xi
\end{align*}
\end{proof}

\begin{co}[\textbf{Isotony}]\label{lb31}
If $\wtd I_1\subset\wtd I_2\in\Jtd$, and $\xi\in\MH_i(I_1)$, then 
\begin{align*}
L(\xi,\wtd I_1)=L(\xi,\wtd I_2)\qquad R(\xi,\wtd I_1)=R(\xi,\wtd I_2)
\end{align*}
when acting on any  $\MH_j\in\RepGA$.
\end{co}

\begin{proof}
Choose $\wtd J\in\Jtd$ clockwise to $\wtd I_2$. By Prop. \ref{lb20}, for each $\eta\in\MH_j(J)$ we have
\begin{align*}
L(\xi,\wtd I_1)\eta=R(\eta,\wtd J)\xi=L(\xi,\wtd I_2)\eta
\end{align*}
The second identity follows from a similar argument.
\end{proof}

\begin{pp}\label{lb22}
Let $\wtd I\in\Jtd$ and $x\in\MA(I)$. Then, when acting on any $\MH_j\in\RepGA$, we have (recalling $\MH_0(I)=\MA(I)\Omega$)
\begin{align*}
L(x\Omega,\wtd I)=R(x\Omega,\wtd I)=\pi_{\wtd I}(x)
\end{align*}
\end{pp}

\begin{proof}
Choose $\wtd J$ clockwise to $\wtd I$. For each $\eta\in\MH_j(J)$, we compute that
\begin{align*}
&L(x\Omega,\wtd I)\eta=L(x\Omega,\wtd I)R(\eta,\wtd J)\Omega=R(\eta,\wtd J)L(x\Omega,\wtd I)\Omega=R(\eta,\wtd J)x\Omega\\
=&\pi_{\wtd I}(x)R(\eta,\wtd J)\Omega=\pi_{\wtd I}(x)\eta
\end{align*}
where the locality and the state-field correspondence in Thm. \ref{lb15} have been used. This proves $L(x\Omega,\wtd I)=\pi_{\wtd I}(x)$. The second relation follows from a similar argument.
\end{proof}

\begin{pp}\label{lb26}
If $S\in\Hom_\MA(\MH_i,\MH_{i'})$ and $T\in\Hom_\MA(\MH_j,\MH_{j'})$ are morphisms in $\RepGA$, then for each $\xi\in\MH_i(I)$ and $\eta\in\MH_j$, we have
\begin{align*}
(S\boxtimes T)L(\xi,\wtd I)\eta=L(S\xi,\wtd I)T\eta\qquad (T\boxtimes S)R(\xi,\wtd I)\eta=R(S\xi,\wtd I)T\eta
\end{align*}
\end{pp}

\begin{proof}
We prove the second relation, as the first one can be proved in a similar way. It suffices to assume that $\eta\in\MH_j(K)$ where $\wtd K$ is anticlockwise to $\wtd I$. By the naturality in Thm. \ref{lb15}, we have
\begin{align*}
(T\boxtimes S)R(\xi,\wtd I)\eta=(\idt\boxtimes S)(T\boxtimes\idt)R(\xi,\wtd I)\eta=(\idt\boxtimes S)R(\xi,\wtd I)T\eta
\end{align*}
By Prop. \ref{lb20} (and noting $T\eta\in\MH_j(K)$ by Rem. \ref{lb43}), the above expression equals
\begin{align*}
(\idt\boxtimes S)L(T\eta,\wtd K)\xi=L(T\eta,\wtd K)S\xi=R(S\xi,\wtd I)T\eta
\end{align*}
\end{proof}

We also note the useful fact that, as $\boxtimes$ is a $*$-bifunctor, we have
\begin{align}\label{eq67}
(S\boxtimes T)^*=S^*\boxtimes T^*
\end{align}

\begin{pp}\label{lb21}
Let $\mc H_i,\mc H_j,\mc H_k\in\RepGA$, $\wtd I\in\Jtd$, and $\xi\in\mc H_i(I)$. The following are true.
\begin{enumerate}[label=(\alph*)]
\item If $\eta\in\mc H_j(I)$, then $L(\xi,\wtd I)\eta\in(\mc H_i\boxtimes\mc H_j)(I)$, $R(\xi,\wtd I)\eta\in(\mc H_j\boxtimes\mc H_i)(I)$, and
\begin{gather*}
L(\xi,\wtd I)L(\eta,\wtd I)|_{\mc H_k}=L(L(\xi,\wtd I)\eta,\wtd I)|_{\mc H_k}\\
R(\xi,\wtd I)R(\eta,\wtd I)|_{\mc H_k}=R(R(\xi,\wtd I)\eta,\wtd I)|_{\mc H_k}
\end{gather*}
\item If $\psi\in(\mc H_i\boxtimes H_j)(I)$ and $\eta\in (\mc H_j\boxtimes H_i)(I)$, then $L(\xi,\wtd I)^*\psi\in\mc H_j(I)$, $R(\xi,\wtd I)^*\eta\in\mc H_j(I)$, and
\begin{gather*}
L(\xi,\wtd I)^*L(\psi,\wtd I)|_{\mc H_k}=L(L(\xi,\wtd I)^*\psi,\wtd I)|_{\mc H_k}\\
 R(\xi,\wtd I)^*R(\eta,\wtd I)|_{\mc H_k}=R(R(\xi,\wtd I)^*\eta,\wtd I)|_{\mc H_k}
\end{gather*}
\end{enumerate} 
\end{pp}

\begin{proof}
We prove the first of (b); part (a) and the second assertiong of (b) follow from a similar argument. Part (a) follows as in \cite[Prop. 3.6]{Gui21a}.

By the state-field correspondence, we have $L(\xi,\wtd I)^*\psi=L(\xi,\wtd I)^*L(\psi,\wtd I)\Omega$. This implies, noticing $L(\xi,\wtd I)^*L(\psi,\wtd I)\big|_{\MH_0}\in\Hom_{\MA(\wtd I')}(\MH_0,\MH_j)$, that $L(\xi,\wtd I)^*\psi\in\mc H_j(I)$. Choose $\wtd J$ clockwise to $\wtd I$ and $\chi\in\mc H_k(J)$. Then
\begin{align*}
&L(\xi,\wtd I)^*L(\psi,\wtd I)\chi=L(\xi,\wtd I)^*L(\psi,\wtd I)R(\chi,\wtd J)\Omega=R(\chi,\wtd J)L(\xi,\wtd I)^*L(\psi,\wtd I)\Omega\\
=&R(\chi,\wtd J)L(\xi,\wtd I)^*\psi=R(\chi,\wtd J)L(L(\xi,\wtd I)^*\psi,\wtd I)\Omega=L(L(\xi,\wtd I)^*\psi,\wtd I)R(\chi,\wtd J)\Omega\\
=&L(L(\xi,\wtd I)^*\psi,\wtd I)\chi
\end{align*}
where the locality and the state-field correspondence in Thm. \ref{lb15} are used.
\end{proof}

\subsection{Representations of commutant algebras on fusion products}\label{lb48}

Fix a conformal net $\MA$ and a group homomorphism $G\rightarrow\Aut(\MA)$. In this section, we only consider nonzero objects of $\RepGA$. 

The purpose of this section is to study how localized commutant algebras of $G$-twisted $\MA$-modules act on fusion products, with the goal of proving Thm.~\ref{lb47}, which computes certain commutants. This result is crucial for the proof of Haag duality for our open/closed CFT.

\begin{df}
For each $\MH_i\in\Rep^\phi(\MA)$ (where $\phi\in G$) and $I\in\MJ$, a vector $\xi\in\MH_i(I)$ is called $\pmb{I}$\textbf{-unitary} if there exists $\arg_I$ such that by setting $\wtd I=(I,\arg_I)$, the operator $L(\xi,\wtd I)|_{\MH_0}:\MH_0\rightarrow\MH_i$ is unitary. By Lem. \ref{lb18}, this is equivalent to that $R(\xi,\wtd I)|_{\MH_0}:\MH_0\rightarrow\MH_i$ is unitary.
\end{df}

By Isotony (Cor. \ref{lb31}), if $\xi$ is $I$-unitary, then $\xi$ is $J$-unitary for any $J\in\MJ$ containing $I$.

\begin{rem}
$I$-unitary vectors must exist. Indeed, since $\MA(I')$ is a type III factor, there exists a unitary $T\in\Hom_{\MA(\wtd I')}(\MH_0,\MH_i)$. Hence $\xi:=T\Omega$ is $I$-unitary. 
\end{rem}

\begin{pp}\label{lb33}
Suppose that $\MH_i\in\Rep^\phi(\MA)$ and $\xi\in\MH_i(I)$ is $I$-unitary. Then for each $\arg_I$ and each $\MH_j\in\RepGA$, by setting $\wtd I=(I,\arg_I)$, the following operators are unitary
\begin{align*}
L(\xi,\wtd I)|_{\MH_j}:\MH_j\rightarrow\MH_i\boxtimes\MH_j\qquad R(\xi,\wtd I)|_{\MH_j}:\MH_j\rightarrow\MH_j\boxtimes\MH_i
\end{align*}
\end{pp}

In particular, since we assume that $\MH_i,\MH_j$ are nonzero, the object $\MH_i\boxtimes\MH_j$ is also nonzero.

\begin{proof}
By Lem. \ref{lb17}, we have $L(\xi,\wtd I)\phi^{-1}|_{\MH_0}\in\Hom_{\MA((\varrho(2\pi)\wtd I)')}(\MH_0,\MH_i)$, and hence
\begin{align}
L(\xi,\varrho(2\pi)\wtd I)|_{\MH_0}=L(\xi,\wtd I)\phi^{-1}|_{\MH_0}
\end{align}
Thus, the unitarity of $L(\xi,\wtd I)|_{\MH_0}$ is independent of $\arg_I$. We now fix an arbitrary $\arg_I$. Choose $\wtd J$ clockwise to $\wtd I$. For each $\eta\in\MH_j(J)$, we have
\begin{align*}
&L(\xi,\wtd I)^*L(\xi,\wtd I)\eta=L(\xi,\wtd I)^*L(\xi,\wtd I)R(\eta,\wtd J)\Omega\\
=&R(\eta,\wtd J)L(\xi,\wtd I)^*L(\xi,\wtd I)\Omega=R(\eta,\wtd J)\Omega=\eta
\end{align*}
Hence $L(\xi,\wtd I)|_{\MH_j}$ is an isometry. Moreover, for $\chi\in\MH_0$ we have
\begin{align*}
L(\xi,\wtd I)R(\eta,\wtd J)\chi=R(\eta,\wtd J)L(\xi,\wtd I)\chi
\end{align*}
Thus, by the unitarity of $L(\xi,\wtd I)|_{\MH_0}$, the above vectors span a dense subspace of $\ovl{R(\MH_j,\wtd J)\MH_i}=\MH_i\boxtimes\MH_j$. This proves that $L(\xi,\wtd I)|_{\MH_j}$ is unitary. A similar argument shows that $R(\xi,\wtd I)|_{\MH_j}$ is unitary.
\end{proof}

\begin{eg}\label{lb147}
Let $\xi\in\MH(I)$ be $I$-unitary. Then for each $g\in\GA$, the vector $g\xi g^{-1}\in\MH(gI)$ defined by the conformal covariance in Thm. \ref{lb15} is $gI$-unitary, since $L(g\xi g^{-1},g\wtd I)|_{\MH_0}=gL(\xi,\wtd I)g^{-1}|_{\MH_0}$ is unitary.
\end{eg}

\begin{pp}\label{lb34}
Let $\MH_i,\MH_k\in\Rep^\phi(\MA)$ and $\wtd I\in\Jtd$. Fix $I$-unitary vectors $\xi_0\in\MH_i(I),\psi_0\in\MH_k(I)$. Then we have
\begin{align*}
&\Hom_{\MA(\wtd I')}(\MH_i,\MH_k)=\{L(\psi,\wtd I)L(\xi,\wtd I)^*|_{\MH_i}:\xi\in\MH_i(I),\psi\in\MH_k(I)\}\\
=&\{L(\psi_0,\wtd I)xL(\xi_0,\wtd I)^*|_{\MH_i}:x\in\MA(I)\}
\end{align*}
and
\begin{align*}
&\Hom_{\MA(\bpr\wtd I)}(\MH_i,\MH_k)=\{R(\psi,\wtd I)R(\xi,\wtd I)^*|_{\MH_i}:\xi\in\MH_i(I),\psi\in\MH_k(I)\}\\
=&\{R(\psi_0,\wtd I)xR(\xi_0,\wtd I)^*|_{\MH_i}:x\in\MA(I)\}
\end{align*}
\end{pp}

\begin{proof}
We prove the first set of equalities; the second follows by the same argument. The inclusions $\supset$ are clear. In particular, the middle set contains the last one, since
\begin{align*}
L(\psi_0,\wtd I)x|_{\MH_0}=L(\psi_0,\wtd I)L(x\Omega,\wtd I)|_{\MH_0}=L(L(\psi_0,\wtd I)x\Omega,\wtd I)|_{\MH_0}
\end{align*}
by Prop. \ref{lb22} and \ref{lb21}. It remains to show that the last set contains the first.

Choose any $T\in \Hom_{\MA(\wtd I')}(\MH_i,\MH_k)$. By the locality in Thm. \ref{lb15}, the element
\begin{align*}
x:=L(\psi_0,\wtd I)^*TL(\xi_0,\wtd I)|_{\MH_0}
\end{align*}
belongs to $\End_{\MA(I')}(\MH_0)=\MA(I)$. By the $I$-unitarity, $T$ equals $L(\psi_0,\wtd I)xL(\xi_0,\wtd I)^*|_{\MH_i}$.
\end{proof}

\begin{thm}\label{lb35}
Let $\MH_i\in\Rep^\phi(\MA)$ and $\wtd I\in\Jtd$. Then there exists an operation $\pmb{\pi^L_{\wtd I}}$ associating to each $\omega\in G$ and $\MH_j\in\Rep^\omega(\MA)$ a unique normal unitary representation
\begin{align}\label{eq41}
\pi^L_{\wtd I}|_{\MH_i\boxtimes\MH_j}:\End_{\MA(\wtd I')}(\MH_i)\rightarrow \fk L(\MH_i\boxtimes\MH_j)
\end{align}
commuting adjointly with $R(\eta,\wtd I')$ for each $\eta\in\MH_j(I')$, that is, for each $A\in \End_{\MA(\wtd I')}(\MH_i)$ the following diagram commutes adjointly.
\begin{equation}\label{eq15}
\begin{tikzcd}[column sep=huge]
\MH_i \arrow[r,"{R(\eta,\wtd I')}"] \arrow[d,"A"'] & \mc \MH_i\boxtimes\MH_j \arrow[d,"\pi^L_{\wtd I}(A)"] \\
\MH_i \arrow[r,"{R(\eta,\wtd I')}"]           & \MH_i\boxtimes\MH_j       
\end{tikzcd}
\end{equation}
Moreover, if $x\in\MA(I)$, then $\pi_{i,\wtd I}(x)\in\End_{\MA(\wtd I')}(\MH_i)$ and 
\begin{align}\label{eq28}
\pi^L_{\wtd I}(\pi_{i,\wtd I}(x))\big|_{\MH_i\boxtimes\MH_j}=\pi_{i\boxtimes j,\wtd I}(x)
\end{align}
\end{thm}

\begin{proof}
The uniqueness follows from \eqref{eq15} and the density of fusion products (or simply by choosing an $I'$-unitary vector $\eta$). Let us prove the existence.

Fix an $I'$-unitary vector $\eta_0$, and define the representation $\pi^L_{\wtd I}$ on $\MH_i\boxtimes\MH_j$ by
\begin{align*}
\pi^L_{\wtd I}(A)|_{\MH_i\boxtimes\MH_j}=R(\eta_0,\wtd I')AR(\eta_0,\wtd I')^*|_{\MH_i\boxtimes\MH_j}
\end{align*}
By Prop. \ref{lb34}, each $A$ can be written as $L(\xi_2,\wtd I)L(\xi_1,\wtd I)^*|_{\MH_i}$ where $\xi_1,\xi_2\in\MH_i(I)$. The locality in Thm. \ref{lb15} implies that
\begin{align}\label{eq17}
\pi^L_{\wtd I}(L(\xi_2,\wtd I)L(\xi_1,\wtd I)^*|_{\MH_i})\big|_{\MH_i\boxtimes\MH_j}=L(\xi_2,\wtd I)L(\xi_1,\wtd I)^*\big|_{\MH_i\boxtimes\MH_j}
\end{align}
Another application of the locality implies the adjoint commutativity of \eqref{eq15}.

If $x\in\MA(I)$, then $\pi_{i,\wtd I}(x)\in\End_{\MA(\wtd I')}(\MH_i)$ by Rem. \ref{lb133}, and 
\begin{align*}
R(\eta_0,\wtd I')\pi_{i,\wtd I}(x)R(\eta_0,\wtd I')^*|_{\MH_i\boxtimes\MH_j}=\pi_{i\boxtimes j,\wtd I}(x)R(\eta_0,\wtd I')R(\eta_0,\wtd I')^*|_{\MH_i\boxtimes\MH_j}=\pi_{i\boxtimes j,\wtd I}(x)
\end{align*}
This proves \eqref{eq28}.
\end{proof}

\begin{thm}\label{lb63}
Let $\MH_j\in\Rep^\omega(\MA)$ and $\wtd J\in\Jtd$. Then there exists an operation $\pmb{\pi^R_{\wtd J}}$ associating to each $\phi\in G$ and $\MH_i\in\Rep^\phi(\MA)$ a unique normal unitary representation
\begin{align}\label{eq42}
\pi^R_{\wtd J}|_{\MH_i\boxtimes\MH_j}:\End_{\MA(\bpr\wtd J)}(\MH_j)\rightarrow \fk L(\MH_i\boxtimes\MH_j)
\end{align}
commuting adjointly with $L(\xi,\bpr\wtd J)$ for each $\xi\in\MH_i(J')$, that is, for each $B\in \End_{\MA(\bpr\wtd J)}(\MH_j)$ the following diagram commutes adjointly.
\begin{equation}\label{eq16}
\begin{tikzcd}[column sep=large]
\MH_j \arrow[r,"B"] \arrow[d,"{L(\xi,\bpr\wtd J)}"'] & \mc \MH_j \arrow[d,"{L(\xi,\bpr\wtd J)}"] \\
\MH_i\boxtimes\MH_j \arrow[r,"\pi^R_J(B)"]           & \MH_i\boxtimes\MH_j       
\end{tikzcd}
\end{equation}
Moreover, if $y\in\MA(J)$, then $\pi_{j,\wtd J}(y)\in\End_{\MA(\bpr\wtd J)}(\MH_j)$ and
\begin{align}\label{eq30}
\pi^R_{\wtd J}(\pi_{j,\wtd J}(y))\big|_{\MH_i\boxtimes\MH_j}=\pi_{i\boxtimes j,\wtd J}(y)
\end{align}
\end{thm}

\begin{proof}
This is similar to the proof of Thm. \ref{lb35}. 
\end{proof}

\begin{rem}\label{lb36}
Similar to the proof of \eqref{eq17}, by the locality in Thm. \ref{lb15}, we can prove the more general fact that if $\xi_1,\xi_2\in\MH_i(I)$, $\eta_1,\eta_2\in\MH_j(J)$, and $x\in\MA(I),y\in\MA(J)$, then
\begin{subequations}\label{eq18}
\begin{gather}
\pi^L_{\wtd I}(L(\xi_2,\wtd I)xL(\xi_1,\wtd I)^*|_{\MH_i})\big|_{\MH_i\boxtimes\MH_j}=L(\xi_2,\wtd I)\pi_{\wtd I}(x)L(\xi_1,\wtd I)^*\big|_{\MH_i\boxtimes\MH_j}\\
\pi^R_{\wtd J}(R(\eta_2,\wtd J)yR(\eta_1,\wtd J)^*|_{\MH_j})\big|_{\MH_i\boxtimes\MH_j}=R(\eta_2,\wtd J)\pi_{\wtd J}(y)R(\eta_1,\wtd J)^*\big|_{\MH_i\boxtimes\MH_j}
\end{gather}
\end{subequations}
\end{rem}

\begin{rem}\label{lb61}
In Thm. \ref{lb35}, if $\wtd I\subset\wtd J$, then $\pi^L_{\wtd J}$ restricts to $\pi^L_{\wtd I}$ on $\End_{\MA(\wtd I')}(\MH_i)$. Similarly, in Thm. \ref{lb63}, if $\wtd I\subset\wtd J$, then $\pi^R_{\wtd J}$ restricts to $\pi^R_{\wtd I}$ on $\End_{\MA(\bpr\wtd I)}(\MH_j)$.
\end{rem}

\begin{proof}
Choose any $A\in\End_{\MA(\wtd I')}(\MH_i)$. Let $\eta_0\in\MH_j(J')$ be $J'$-unitary, and let $V=R(\eta_0,\wtd J')|_{\MH_i}$. Then \eqref{eq15} implies $\pi^L_{\wtd I}(A)=VAV^{-1}=\pi^L_{\wtd J}(A)$. 
\end{proof}

\begin{rem}\label{lb91}
The representation \eqref{eq41} is faithful since, by the adjoint commutativity of \eqref{eq15}, it is unitarily equivalent to the action of $\End_{\MA(\wtd I')}(\MH_i)$ on $\MH_i$ via the unitary operator $R(\eta_0,\wtd I')|_{\MH_i}$, where $\eta_0\in\MH_j(I')$ is $I'$-unitary. Similarly, the representation \eqref{eq42} is faithful. It follows that
\begin{align*}
\pi^L_{\wtd I}(\End_{\MA(\wtd I')}(\MH_i))\big|_{\MH_i\boxtimes\MH_j}\text{ and } \pi^R_{\wtd J}(\End_{\MA(\bpr\wtd J)}(\MH_j))\big|_{\MH_i\boxtimes\MH_j}\text{ are type III factors}
\end{align*}
since $\End_{\MA(\wtd I')}(\MH_i)$ and $\End_{\MA(\bpr\wtd J)}(\MH_j)$ are so.
\end{rem}

\begin{thm}\label{lb108}
Let $\MH_i\in\Rep^\phi(\MA),\MH_j\in\Rep^\omega(\MA)$, $g\in\GA$, and $\wtd I,\wtd J\in\Jtd$ with $\wtd J=g\wtd I$. Then
\begin{gather*}
U_i(g)\End_{\MA(\wtd I')}(\MH_i)U_i(g)^*=\End_{\MA(\wtd J')}(\MH_i)\\ 
U_j(g)\End_{\MA(\bpr\wtd I)}(\MH_j)U_j(g)^*=\End_{\MA(\bpr\wtd J)}(\MH_j)
\end{gather*}
Moreover, for each $A\in \End_{\MA(\wtd I')}(\MH_i)$ and $B\in \End_{\MA(\bpr\wtd I)}(\MH_j)$, we have
\begin{gather*}
U_{i\boxtimes j}(g)\cdot\pi^L_{\wtd I}(A)|_{\MH_i\boxtimes\MH_j}\cdot U_{i\boxtimes j}(g)^*=\pi^L_{\wtd J}(U_i(g)AU_i(g)^*)|_{\MH_i\boxtimes\MH_j}\\
U_{i\boxtimes j}(g)\cdot\pi^R_{\wtd I}(B)|_{\MH_i\boxtimes\MH_j}\cdot U_{i\boxtimes j}(g)^*=\pi^R_{\wtd J}(U_j(g)BU_j(g)^*)|_{\MH_i\boxtimes\MH_j}
\end{gather*}
\end{thm}

\begin{proof}
The first two identities follow from \eqref{eq76}. Let us prove the third identity; the fourth one will follow from a similar argument.

By Prop. \ref{lb34}, elements of $\End_{\MA(\wtd I')}(\MH_i)$ are precisely of the form
\begin{align*}
A=L(\xi_2,\wtd I)L(\xi_1,\wtd I)^*|_{\MH_i}
\end{align*}
where $\xi_1,\xi_2\in\MH_i(I)$. By the conformal covariance in Thm. \ref{lb15}, for the vectors $g\xi_1 g^{-1},g\xi_2 g^{-1}\in\MH_i(J)$ we have
\begin{align*}\label{eq47}
U_i(g)AU_i(g)^*=L(g\xi_2 g^{-1},\wtd J)L(g\xi_1 g^{-1},\wtd J)^*|_{\MH_i}\tag{$\star$}
\end{align*}
By Rem. \ref{lb36} and the conformal covariance in Thm. \ref{lb15},
\begin{align*}
&U_{i\boxtimes j}(g)\cdot \pi^L_{\wtd I}(A)|_{\MH_i\boxtimes\MH_j}\cdot U_{i\boxtimes j}(g)^*=U_{i\boxtimes j}(g)\cdot L(\xi_2,\wtd I)L(\xi_1,\wtd I)^*|_{\MH_i\boxtimes\MH_j}\cdot U_{i\boxtimes j}(g)^*\\
=&L(g\xi_2 g^{-1},\wtd J)L(g\xi_1 g^{-1},\wtd J)^*|_{\MH_i\boxtimes\MH_j}=\pi^L_{\wtd J}(\eqref{eq47})|_{\MH_i\boxtimes\MH_j}
\end{align*}
This proves the third identity.
\end{proof}

\begin{thm}\label{lb47}
Let $\MH_i\in\Rep^\phi(\MA),\MH_j\in\Rep^\omega(\MA)$ and $\wtd I\in\Jtd$. Then the following two von Neumann algebras on $\MH_i\boxtimes\MH_j$ are commutants of each other. 
\begin{align*}
\pi^L_{\wtd I}(\End_{\MA(\wtd I')}(\MH_i))\big|_{\MH_i\boxtimes\MH_j}\qquad\pi^R_{\wtd I'}(\End_{\MA(\wtd I)}(\MH_j))\big|_{\MH_i\boxtimes\MH_j}
\end{align*}
\end{thm}

\begin{proof}
Eq. \eqref{eq18}, together with Prop. \ref{lb34}, implies that the above two von Neumann algebras commute with each other. 

Suppose that $B\in\fk L(\MH_i\boxtimes\MH_j)$ commutes with $\pi^L_{\wtd I}(\End_{\MA(\wtd I')}(\MH_i))\big|_{\MH_i\boxtimes\MH_j}$. Choose an $I$-unitary $\xi\in\MH_i(I)$ and an $I'$-unitary $\eta\in\MH_j(I')$. By Rem. \ref{lb36}, $B$ commutes with $L(\xi,\wtd I)\pi_{\wtd I}(\MA(I))L(\xi,\wtd I)^*|_{\MH_i\boxtimes\MH_j}$. Thus
\begin{align*}
L(\xi,\wtd I)^*BL(\xi,\wtd I)|_{\MH_j}
\end{align*}
commutes with $\pi_{\wtd I,j}(\MA(I))$, and hence (by Prop. \ref{lb34}) can be written as $R(\eta,\wtd I')yR(\eta,\wtd I')^*|_{\MH_j}$ for some $y\in\MA(I')$. By the locality in Thm. \ref{lb15}, 
\begin{align*}
&B=L(\xi,\wtd I)R(\eta,\wtd I')yR(\eta,\wtd I')^*L(\xi,\wtd I)^*|_{\MH_i\boxtimes\MH_j}\\
=&R(\eta,\wtd I')\pi_{\wtd I'}(y)R(\eta,\wtd I')^*L(\xi,\wtd I)L(\xi,\wtd I)^*|_{\MH_i\boxtimes\MH_j}=R(\eta,\wtd I')\pi_{\wtd I'}(y)R(\eta,\wtd I')^*|_{\MH_i\boxtimes\MH_j}
\end{align*}
Therefore, by Rem. \ref{lb36}, $B$ belongs to $\pi^R_{\wtd I'}(\End_{\MA(\wtd I)}(\MH_j))\big|_{\MH_i\boxtimes\MH_j}$.
\end{proof}

\subsection{Conjugates and commutants in fusion products}

In this section, we consider only nonzero objects of $\RepGA$.

Define $\pmb{\rk}:\Sbb^1\rightarrow\Sbb^1$ to be the conjugation map $z\mapsto \ovl z$. It lifts to the map  $\rk:\Rbb\rightarrow\Rbb,x\mapsto -x$ on the universal cover. This defines $\pmb{\rk\wtd I}$ as an element of $\Jtd$ if $\wtd I\in\Jtd$. Let $\pmb{\wtd\Sbb^1_+}$ be the upper half circle $\Sbb^1_+=\{e^{\im t}:0<t<\pi\}$ with arg function ranging in $(0,\pi)$, and let $\pmb{\wtd\Sbb^1_-}=\rk\Sbb^1_+$.

The goal of this section is to show that when $\MH_j=\MH_{\ovl i}$, the pair of commutants appearing in Thm.~\ref{lb47} are related by an anti-unitary map (introduced in \cite[Sec. A]{MS26a}), thereby yielding a Tomita--Takesaki-type theorem for $G$-crossed categorical extensions; see Cor. \ref{lb87}. This result will play a crucial role in establishing the PCT theorem and Haag duality for open/closed CFTs. As preparation for the proof of Cor. \ref{lb87}, we recall and adapt in Subsec. \ref{lb142} and \ref{lb143} several constructions and basic properties about conjugate modules introduced in \cite[Sec. A]{MS26a}.

\subsubsection{The conjugation in $\GA$}\label{lb142}

\begin{rem}\label{lb37}
Let $\pmb{\fk J_\MA}$ (or simply $\fk J$) be the (involutive) modular conjugation of $\MA(\Sbb^1_+)$ with respect to the cyclic separating vector $\Omega$. By the geometric modular theory (cf. \cite[Sec. II.2]{GF93}), we have
\begin{gather*}
\fk J\MA(I)\fk J=\MA(\rk I)\qquad\text{for each }I\in\MJ\\
\fk JU_0(\varrho(t))\fk J=U_0(\varrho(-t))\qquad\text{for each }t\in\Rbb
\end{gather*}
\end{rem}

\begin{rem}\label{lb38}
Each $\varphi\in\Aut(\MA)$ clearly commutes with the modular $S$-operator $x\Omega\mapsto x^*\Omega$ where $x\in\MA(\Sbb^1_+)$. Therefore, by the polar decomposition, $\varphi$ commutes with $\fk J$.
\end{rem}

\begin{df}\label{lb51}
For each $\wtd g\in\GA$, recall that $\wtd g$ is of the form $(g,V)$ where $g\in\SG$ and $V\in\MU(\MH_0)$ represents $U_0(g)$. Define $\pmb{\rk\wtd g\rk}\in\GA$ by
\begin{align}
\rk\wtd g\rk=(\rk g\rk,\fk JV\fk J)
\end{align} 
where, in terms of the description \eqref{eq2} of $\SG$, the element $\rk g\rk\in\SG$ sends each $x\in\Rbb$ to $-g(-x)$. It follows that
\begin{align}\label{eq20}
U_0(\rk\wtd g\rk)=\fk JU_0(\wtd g)\fk J
\end{align}
\end{df}

\begin{proof}
We need to explain why $\fk JV\fk J$ represents $U_0(\rk g\rk)$. For that purpose, it suffices to show that
\begin{align*}
U_0':g\in\DiffS\mapsto \fk JU_0(\rk g\rk)\fk J
\end{align*}
agrees with $U_0$ as projective representations of $\DiffS$. By \cite[Thm. 2.19]{GF93}, $U_0$ agrees with $U_0'$ when restricted to $\PSU$. Therefore, $\MA$ is conformally covariant under $U_0'$ in the sense of Def. \ref{lb2}. By \cite[Thm. 6.1.9]{Wei05}, a conformal covariance extending a given M\"obius covariance of a conformal net is unique. Therefore $U_0'=U_0$ on $\DiffS$.
\end{proof}

\begin{rem}
Recall the identification $\GA=\Gc$ in Rem. \ref{lb39}. Then the definition of $\rk \wtd g\rk$ for $g\in\Gc$ is independent of the conformal net $\MA$ extending $\Vir_c$.
\end{rem}

\begin{proof}
Let $\MK_0=\ovl{\Vir_c\Omega}$. Then $(\Vir_c,\MK_0)$ is the Virasoro subnet of $(\MA,\MH_0)$. Since the unitary representation of $\PSU$ on $\MH_0$ restricts to that on $\MK_0$, by the Bisognano-Wichmann property \cite[Thm. 2.19]{GF93}, the modular operators $\Delta_\MA,\Delta_c$ of $\MA(\Sbb^1_+)$ and $\Vir_c(\Sbb^1_+)$ satisfy $\Delta_\MA^{\im t}|_{\MK_0}=\Delta_c^{\im t}$. Therefore, by (for instance) the proof of Thm. 4.2 in \cite[Sec. IX.4]{Tak03}, the modular conjugations $\fk J_{\MA},\fk J_c$ satisfy $\fk J_\MA|_{\MK_0}=\fk J_c$. Hence, for each $\wtd g=(g,V)\in\GA$, the operator $\fk J_\MA V\fk J_\MA$ restricts to $\fk J_c V\fk J_c$ on $\MK_0$.
\end{proof}

\begin{eg}\label{lb41}
Recall the rotation subgroup $\rho_\MA:\Rbb\rightarrow\GA$ defined in Rem. \ref{lb40}. Then by Rem. \ref{lb37},
\begin{align*}
\rk\varrho_\MA(t)\rk=\varrho_\MA(-t)
\end{align*}
\end{eg}

\subsubsection{Conjugate modules}\label{lb143}

\begin{df}\label{lb46}
For each $\MH_i\in\RepGA$, define the $G$-crossed \textbf{conjugate $\MA$-module} $\pmb{\MH_{\ovl i}}$ as follows. As a Hilbert space, $\MH_{\ovl i}$ is the complex conjugate of $\MH_i$. We denote the conjugation map by $\pmb{\Co_i}:\MH_i\rightarrow\MH_{\ovl i}$, and abbreviate it to $\pmb\Co$ when no confusion arises. Thus
\begin{align*}
\Co_i\xi\equiv\Co\xi\equiv\ovl \xi\qquad\text{for each }\xi\in\MH_i
\end{align*}
For each $\wtd I\in\Jtd$ and $x\in\MA(I)$, we define
\begin{align*}
\pi_{\ovl i,\wtd I}(x)=\Co_i\pi_{i,\rk\wtd I}(\fk Jx\fk J)\Co_i^{-1}
\end{align*}
noting that $\fk Jx\fk J\in\MA(\rk I)$ by Rem. \ref{lb37}. In short,
\begin{align}\label{eq19}
\pi_{\wtd I}(x)=\Co\pi_{\rk\wtd I}(\fk Jx\fk J)\Co^{-1}
\end{align}
We identify $\MH_{\ovl{\ovl i}}$ with $\MH_i$ canonically, that is, by identifying $\xi\in\MH_i$ with $\Co_{\ovl i}\Co_i\xi$ if $\xi\in\MH_i$. Then $\Co_i^{-1}=\Co_{\ovl i}$, in short
\begin{align*}
\Co^{-1}=\Co
\end{align*}
\end{df}

\begin{rem}
The fact that $\MH_i\in\RepGA$ whenever $\MH_i\in\RepGA$ follows from the fact that if $\MH_i\in\Rep^\phi(\MA)$ then $\MH_{\ovl i}\in\Rep^{\phi^{-1}}(\MA)$.
\end{rem}

\begin{proof}
Choose any $\wtd I\in\Jtd$ and $x\in\MA(I)$. Then, by Rem. \ref{lb38},
\begin{align*}
\pi_{\ovl i,\varrho(-2\pi)\wtd I}(\phi x\phi^{-1})=\Co\pi_{i,\rk\varrho(-2\pi)\wtd I}(\fk J\phi x\phi^{-1}\fk J)\Co^{-1}=\Co\pi_{i,\varrho(2\pi)\rk\wtd I}(\phi\fk J x\fk J\phi^{-1})\Co^{-1}
\end{align*}
which equals $\Co\pi_{i,\rk\wtd I}(\fk Jx\fk J)\Co^{-1}=\pi_{\ovl i,\wtd I}(x)$ because $\MH_i$ is $\phi$-twisted.
\end{proof}

\begin{df}\label{lb82}
Let $\MH_i,\MH_j\in\RepGA$. A bounded antilinear map $T:\MH_i\rightarrow\MH_j$ is called an \textbf{anti-morphism} of $G$-crossed $\MA$-modules if one of the following equivalent conditions hold:
\begin{enumerate}[label=(\arabic*)]
\item For each $\wtd I\in\Jtd$ and $x\in\MA(I)$, we have
\begin{align*}
T\pi_{i,\wtd I}(x)=\pi_{j,\rk\wtd I}(\fk Jx\fk J)T
\end{align*}
\item The map $T\circ\Co_{\ovl i}^{-1}:\MH_{\ovl i}\rightarrow\MH_j$ is a morphism of $G$-crossed $\MA$-modules.
\item The map $\Co_j\circ T:\MH_i\rightarrow\MH_{\ovl j}$ is a morphism of $G$-crossed $\MA$-modules.
\end{enumerate}
The equivalence is clear by \eqref{eq19}.
\end{df}

\begin{pp}\label{lb42}
Let $\MH_i,\MH_j\in\RepGA$, and let $T:\MH_i\rightarrow\MH_j$ be a morphism (resp. anti-morphism) of $G$-crossed $\MA$-modules. Then for each $g\in\GA$ we have
\begin{align}
TU_i(g)=U_j(g)T\qquad\text{resp.}\qquad TU_i(g)=U_j(\rk g\rk)T
\end{align}  
\end{pp}

\begin{proof}
We treat the case where $T$ is an anti-morphism; the other case is similar and easier. By Lem. \ref{lb14}, it suffices to assume that $g\in\GA(I)$ for some $I\in\MJ$. Then
\begin{align*}
TU(g)=T\pi_I(U_0(g))=\pi_{\rk I}(\fk JU_0(g)\fk J)T=\pi_{\rk I}(U_0(\rk g\rk))T=U(\rk g\rk)T
\end{align*}
where Thm. \ref{lb12} and \eqref{eq20} are used.
\end{proof}

\begin{co}\label{lb84}
If $T:\MH_i\rightarrow\MH_j$ is an anti-morphism of $G$-crossed $\MA$-modules, then
\begin{align*}
TU_i(\varrho_\MA(t))=U_j(\varrho_\MA(-t))T
\end{align*}
\end{co}

\begin{proof}
This follows immediately from Prop. \ref{lb42} and Exp. \ref{lb41}.
\end{proof}

Parallel to Rem. \ref{lb43}, we have:
\begin{pp}\label{lb44}
Let $T:\MH_i\rightarrow\MH_j$ be an anti-morphism of $G$-crossed $\MA$-modules. Then for each $I\in\MJ$ we have
\begin{align*}
T\MH_i(I)\subset\MH_j(\rk I)
\end{align*}
Moreover, if $T$ is anti-unitary, then $T\MH_i(I)=\MH_j(\rk I)$.
\end{pp}

\begin{proof}
Choose any $\xi\in\MH_i(I)$. Note that $(\rk\wtd I)'=\rk(\bpr\wtd I)$. Therefore, for each $x\in\MA((\rk I)')$, we have $\fk Jx\fk J\in\MA(\bpr\wtd I)$, and hence
\begin{align*}
TR(\xi,\wtd I)\fk Jx|_{\MH_0}=T\pi_{\bpr\wtd I}(\fk Jx\fk J)R(\xi,\wtd I)|_{\MH_0}=\pi_{\rk(\bpr\wtd I)}(x)TR(\xi,\wtd I)\fk J|_{\MH_0}
\end{align*}
So $TR(\xi,\wtd I)\fk J|_{\MH_0}$ belongs to $\Hom_{\MA((\rk\wtd I)')}(\MH_0,\MH_j)$, and hence $T\xi=TR(\xi,\wtd I)\fk J\Omega$ belongs to $\MH_j(\rk I)$. We have thus finished proving $T\MH_i(I)\subset\MH_j(\rk I)$. If $T$ is anti-unitary, then similarly we have $T^*\MH_j(\rk I)\subset\MH_i(I)$, and hence $T\MH_i(I)=\MH_j(\rk I)$.
\end{proof}

\begin{rem}\label{lb45}
From the proof of Prop. \ref{lb44}, we see that if $\xi\in\MH_i(I)$ and $T:\MH_i\rightarrow\MH_j$ is an anti-morphism of objects of $\RepGA$, then
\begin{align*}
L(T\xi,\rk\wtd I)|_{\MH_0}=TR(\xi,\wtd I)\fk J|_{\MH_0}\qquad R(T\xi,\rk\wtd I)|_{\MH_0}=TL(\xi,\wtd I)\fk J|_{\MH_0}
\end{align*}
\end{rem}

\subsubsection{Conjugates and commutants}

The following theorem and its proof are adapted from \cite[Prop. A.9]{MS26a} and the proof therein.

\begin{thm}\label{lb49}
Let $\MH_i,\MH_j\in\RepGA$. Then there exists a unique antiunitary map
\begin{gather*}
\pmb{\Theta_{i,j}}:\MH_i\boxtimes\MH_j\rightarrow\MH_{\ovl j}\boxtimes\MH_{\ovl i}
\end{gather*}
abbreviated to $\pmb{\Theta}$ when no confusion arises, such that for each $\wtd I,\wtd J\in\Jtd$ and $\xi\in\MH_i(I),\eta\in\MH_j(J)$,
\begin{subequations}\label{eq21}
\begin{gather}
\Theta_{i,j}L(\xi,\wtd I)\big|_{\MH_j}=R(\Co\xi,\rk \wtd I)\Co\big|_{\MH_j}\label{eq21a}\\
\Theta_{i,j}R(\eta,\wtd J)\big|_{\MH_i}=L(\Co\eta,\rk \wtd J)\Co\big|_{\MH_i}\label{eq21b}
\end{gather}
\end{subequations}
Moreover, $\Theta_{i,j}$ is an anti-morphism of $G$-crossed $\MA$-modules.
\end{thm}

Note that by Prop. \ref{lb44}, we have $\Co\xi\in\MH_{\ovl i}(\rk I)$ and $\Co\eta\in\MH_{\ovl j}(\rk J)$. Thus the RHS of \eqref{eq21} can be defined.

\begin{proof}
For each $\xi_1,\xi_2\in\MH_i(I)$ and $\eta_1,\eta_2\in\MH_j$, by \eqref{eq22} and Rem. \ref{lb45},
\begin{align*}
&\bk{R(\Co\xi_2,\rk\wtd I)\Co\eta_2|R(\Co\xi_1,\rk\wtd I)\Co\eta_1}=\bk{\Co\eta_2|\pi_{\rk\wtd I}(R(\Co\xi_2,\rk\wtd I)^*R(\Co\xi_1,\rk\wtd I)|_{\MH_0})\Co\eta_1}\\
=&\bk{\Co\pi_{\rk\wtd I}(\fk JL(\xi_2,\wtd I)^*L(\xi_1,\wtd I)\fk J|_{\MH_0})\Co\eta_1|\eta_2}=\bk{\pi_{\wtd I}(L(\xi_2,\wtd I)^*L(\xi_1,\wtd I)|_{\MH_0})\eta_1|\eta_2}\\
=&\bk{L(\xi_1,\wtd I)\eta_1|L(\xi_2,\wtd I)\eta_2}
\end{align*}
where the second last equality is due to Def. \ref{lb46}. Therefore, by the density of fusion products in Thm. \ref{lb15}, there exists a unique anti-unitary map $\Theta_{i,j}^{\wtd I}:\MH_i\boxtimes\MH_j\rightarrow\MH_{\ovl j}\boxtimes\MH_{\ovl i}$ sending each $L(\xi,\wtd I)\eta$ to $R(\Co\xi,\rk\wtd I)\Co\eta$. Similar to Step 2 of the proof of Thm. \ref{lb28}, this map is independent of $\wtd I$, and hence can be denoted by $\Theta_{i,j}$.

To show that $\Theta_{i,j}$ is an anti-morphism, we choose any $\wtd J\in\Jtd,y\in\MA(J)$, let $\wtd I=\bpr\wtd J$ and $\xi\in\MH_i(I)$, and compute that
\begin{align*}
\Theta_{i,j}\pi_{\wtd J}(y)L(\xi,\wtd I)\eta=\Theta_{i,j}L(\xi,\wtd I)\pi_{\wtd J}(y)\eta=R(\Co\xi,\rk\wtd I)\Co\pi_{\wtd J}(y)\eta=R(\Co\xi,\rk\wtd I)\pi_{\rk\wtd J}(\fk Jy\fk J)\Co\eta
\end{align*}
Since $\rk\wtd J$ is anticlockwise to $\rk\wtd I$, we conclude
\begin{align*}
\Theta_{i,j}\pi_{\wtd J}(y)L(\xi,\wtd I)\eta=\pi_{\rk\wtd J}(\fk Jy\fk J)R(\Co\xi,\rk\wtd I)\Co\eta=\pi_{\rk\wtd J}(\fk Jy\fk J)\Theta_{i,j}L(\xi,\wtd I)\eta
\end{align*}

Finally, for each $\xi\in\MH_i(I)$ where $\wtd I$ is anticlockwise to $\wtd J$, we have
\begin{align*}
\Theta_{i,j}R(\eta,\wtd J)\xi=\Theta_{i,j}L(\xi,\wtd I)\eta=R(\Co\xi,\rk\wtd I)\Co\eta=L(\Co\eta,\rk\wtd J)\Co\xi
\end{align*}
due to \eqref{eq21a} and Prop. \ref{lb20}. This proves \eqref{eq21b}.
\end{proof}

\begin{rem}\label{lb81}
By \eqref{eq21}, we clearly have
\begin{align*}
\Theta_{\ovl j,\ovl i}\Theta_{i,j}=\id_{\MH_i\boxtimes\MH_j}
\end{align*}
In particular, $\Theta_{i,\ovl i}$ is an involutive antiunitary operator on $\MH_i\boxtimes\MH_{\ovl i}$.
\end{rem}


\begin{thm}\label{lb107}
Let $\MH_i\in\Rep^\phi(\MA)$ and $\wtd I,\wtd J\in\Jtd$ with $\wtd J=\rk\wtd I$. Then
\begin{align*}
\Co_i \End_{\MA(\wtd I')}(\MH_i)\Co_i^{-1}=\End_{\MA(\bpr\wtd J)}(\MH_{\ovl i})
\end{align*}
Moreover, for each $\omega\in G,\MH_j\in\Rep^\omega(\MA)$ and $A\in\End_{\MA(\wtd I')}(\MH_i)$, 
\begin{align*}
\Theta_{i,j}\cdot\pi^L_{\wtd I}(A)\big|_{\MH_i\boxtimes\MH_j}\cdot\Theta_{i,j}^{-1}=\pi^R_{\wtd J}(\Co_i A\Co_i^{-1})\big|_{\MH_{\ovl j}\boxtimes\MH_{\ovl i}}
\end{align*}
\end{thm}

\begin{proof}
By Prop. \ref{lb34}, elements of $\End_{\MA(\wtd I')}(\MH_i)$ are precisely of the form
\begin{align*}
A=L(\xi_2,\wtd I)L(\xi_1,\wtd I)^*|_{\MH_i}
\end{align*}
where $\xi_1,\xi_2\in\MH_i(I)$. By Prop. \ref{lb44}, we have $\Co\xi_1,\Co\xi_2\in\MH_{\ovl i}(J)$. By Rem. \ref{lb45}, we have
\begin{align*}
\Co_iA\Co_i^{-1}=R(\Co\xi_2,\wtd J)R(\Co\xi_1,\wtd J)^*|_{\MH_{\ovl i}}=:B
\end{align*}
By Prop. \ref{lb34}, elements of $\End_{\MA(\bpr\wtd J)}(\MH_{\ovl i})$ are precisely of the form $B$. This proves the first identity. By Rem. \ref{lb36} and Thm. \ref{lb49},
\begin{align*}
&\Theta_{i,j}\cdot\pi^L_{\wtd I}(A)|_{\MH_i\boxtimes\MH_j}\cdot\Theta_{i,j}^{-1}=\Theta_{i,j}\cdot L(\xi_2,\wtd I)L(\xi_1,\wtd I)^*|_{\MH_i\boxtimes\MH_j}\cdot\Theta_{i,j}^*\\
=&R(\Co\xi_2,\wtd J)R(\Co\xi_1,\wtd J)^*|_{\MH_{\ovl j}\boxtimes\MH_{\ovl i}}=\pi^R_{\wtd J}(B)|_{\MH_{\ovl j}\boxtimes\MH_{\ovl i}}
\end{align*}
This proves the second identity.
\end{proof}

\begin{co}\label{lb87}
Let $\MH_i\in\Rep^\phi(\MA),\MH_j\in\Rep^\omega(\MA)$ and $\wtd I,\wtd J\in\Jtd$ with $\wtd J=\rk\wtd I$. Let
\begin{align*}
\MM=\pi^L_{\wtd I}(\End_{\MA(\wtd I')}(\MH_i))\big|_{\MH_i\boxtimes\MH_j}\qquad \MN=\pi^R_{\wtd J}(\End_{\MA(\bpr\wtd J)}(\MH_{\ovl i}))\big|_{\MH_{\ovl j}\boxtimes\MH_{\ovl i}}
\end{align*}
Then $\Theta_{i,j}\MM\Theta_{i,j}^{-1}=\MN$. Moreover, if $j=\ovl i$, and if $\wtd I=\Std_+$ (and hence $\wtd J=\Std_-$), then
\begin{align*}
\Theta_{i,\ovl i}\MM\Theta_{i,\ovl i}=\MN=\MM'
\end{align*}
\end{co}

\begin{proof}
Thm. \ref{lb107} implies $\Theta_{i,j}\MM\Theta_{i,j}^{-1}=\MN$. When $j=\ovl i$ and $\wtd I=\Std_+$, then $\MN=\MM'$ follows from Thm. \ref{lb47}.
\end{proof}

\section{The closed CFT}

Fix a conformal net $\MA$ with central charge $c$. The convention introduced in this chapter and assumed throughout the remainder of the paper is Convention \ref{lb122}. We assume the identification $\GA=\Gc$ described in Rem. \ref{lb39}.

\subsection{The closed-string state space $\MHcl=\MH_{\sqz}\boxtimes\sbsc{\Maa}\MH_{\ovl\sqz}$}

In this section, we define the state space for our closed-string CFT as the fusion product of a twisted $\Maa$-module $\MH_\sqz$ with its conjugate. This twisted $\Maa$-module was introduced in \cite{LX04} in the language of sectors, where its basic properties were also studied. Here, we adapt that construction and some of the basic properties to the setting of the present paper.

Recall that the tensor product conformal net
\begin{align*}
\Maa=\MA\otimes\MA
\end{align*}
acting on the Hilbert space $\MH_0\otimes_\Cbb\MH_0=\MH_0\otimes\MH_0$ and vacuum vector $\Omega\otimes\Omega$. By definition, for each $I\in\MJ$, we have
\begin{align*}
\Maa(I)=\MA(I)\otimes\MA(I):=\{x^+\otimes x^-\in\fk L(\MH_0\otimes\MH_0):x^\pm\in\MA(I)\}''
\end{align*}
The strongly-continuous projectively-unitary representation of $\DiffS$ on $\MH_0\otimes\MH_0$ is defined by 
\begin{align*}
g\in\DiffS\mapsto U_0(g)\otimes U_0(g)
\end{align*}

\subsubsection{The separate conformal structures on untwisted $\Maa$-modules}

\begin{df}\label{lb83}
Fix $\pm\in\{+,-\}$. Then each $\MH_\Jbb\in\Rep(\Maa)$ gives rise to an $\MA$-module
\begin{align*}
(\MH_\Jbb,\pmb{\pi^\pm_\Jbb})
\end{align*}
abbreviated to $\pmb{\pi^\pm}$ when no confusion arises, such that for each $I\in\MJ$ and $x\in\MA(I)$,
\begin{align*}
\pi_{\Jbb,I}^+(x)=\pi_{\Jbb,I}(x\otimes 1)\qquad \pi_{\Jbb,I}^-(x)=\pi_{\Jbb,I}(1\otimes x)
\end{align*}
\end{df}

\begin{df}\label{lb53}
By Thm. \ref{lb12}, $(\MH_\Jbb,\pi^\pm_\Jbb)$ admits a canonical strongly-continuous unitary representation of $\GA$, which we denote by $\pmb{U_\Jbb^\pm}$, or simply by $\pmb{U^\pm}$ when the context is clear. In other words, $U^\pm_\Jbb:\GA\rightarrow \MU(\MH_\Jbb)$ is the unique strongly-continuous unitary representation such that
\begin{gather}\label{eq23}
U^+_\Jbb(g)=\pi_{\Jbb,I}(U_0(g)\otimes 1)\qquad U^-_\Jbb(g)=\pi_{\Jbb,I}(1\otimes U_0(g))
\end{gather}
for each $I\in\MJ,g\in\GA(I)$.
\end{df}

\begin{rem}\label{lb104}
Let $\MH_\Jbb\in\Rep(\Maa)$. Then for each $I,J\in\MJ$, we have
\begin{align*}
[\pi^+_{\Jbb,I}(\MA(I)),\pi^-_{\Jbb,J}(\MA(J))]=0
\end{align*}
Indeed, in the special case where $I\cup J\subset K$ for some $K\in\MJ$, the commutativity holds because for each $x\in\MA(I),y\in\MA(J)$, we have
\begin{align*}
\pi^+_{\Jbb,I}(x)\pi^-_{\Jbb,J}(y)=\pi_{\Jbb,K}(x\otimes y)=\pi^-_{\Jbb,J}(y)\pi^+_{\Jbb,I}(x)
\end{align*}
The general case reduces to this special case by the additivity of $\MA$.

It thus follows from Thm. \ref{lb12} that
\begin{align*}
[U_\Jbb^+(g^+),U_\Jbb^-(g^-)]=0\qquad [U_\Jbb^\pm(g^\pm),\pi^\mp_{\Jbb,I}(x)]=0
\end{align*}
for each $\pm\in\{+,-\}$, $g^+,g^-\in\GA$, $I\in\MJ$, and $x\in\MA(I)$.   \hqed
\end{rem}

\subsubsection{The $\sigma$-twisted $\Maa$-module $\MH_{\sqrt i}$}\label{lb57}

\begin{df}
For each $\wtd I=(I,\arg_I)\in\Jtd$, define elements $\pmb{\sqrt{\wtd I}}=(\sqrt I,\arg_{\sqrt I})$ and $\pmb{-\sqrt{\wtd I}}=(-\sqrt I,\arg_{-\sqrt I})$ of $\Jtd$ by
\begin{gather*}
\sqrt I=\big\{e^{\frac{\im\theta}2}:\theta\in\arg_I(I)\big\}\qquad \arg_{\sqrt I}(e^{\frac{\im\theta}2})=\frac{\theta}2 \\
-\sqrt I=\big\{-e^{\frac{\im\theta}2}:\theta\in\arg_I(I)\big\}\qquad\arg_{-\sqrt I}(-e^{\frac{\im\theta}2})=\frac\theta 2-\pi
\end{gather*}
for each $\theta\in\arg_I(I)$. In other words, by viewing $\Sbb^1=\Rbb/2\pi\Zbb$ and writing $\wtd I=(a,b)$, we have $\sqrt{\wtd I}=(\frac a2,\frac b2)$ and $-\sqrt{\wtd I}=(\frac a2-\pi,\frac b2-\pi)$.
\end{df}

We warn the readers that the notation $\sqrt I,-\sqrt I$ depends not only on $I$, but also on $\arg_I$.

\begin{df}\label{lb50}
For each $\wtd I\in\Jtd$, choose $g_+,g_-\in\DiffS$ such that
\begin{align*}
g_+(e^{\im\theta})=e^{\frac{\im\theta}2}\qquad g_-(e^{\im\theta})=-e^{\frac{\im\theta}2}\qquad\text{for each }\theta\in\arg_I(I)
\end{align*}
Choose representing elements of $U_0(g_\pm)$ (which we still denote by $U_0(g_\pm)$), and define (spatial) isomorphisms of von Neumann algebras
\begin{gather*}
\mbf Q_{\sqrt{\wtd I}}:\MA(I)\rightarrow\MA(\sqrt I)\qquad x\mapsto U_0(g_+)xU_0(g_+)^{-1}\\
\mbf Q_{-\sqrt{\wtd I}}:\MA(I)\rightarrow\MA(-\sqrt I)\qquad x\mapsto U_0(g_-)xU_0(g_-)^{-1}
\end{gather*}
This definition is independent of the choice of $g_\pm$.
\end{df}

\begin{proof}[Explanation]
To prove the independence of $\mbf Q_{\sqrt{\wtd I}}$ on $g_+$ (and similarly the independence of $\mbf Q_{-\sqrt{\wtd I}}$ on $g_-$), assume that $h_+$ satisfies the same property as $g_+$. Then $h_+^{-1}g_+\in\Diff_{I'}(\Sbb^1)$, and hence $U_0(h_+)^{-1}U_0(g_+)$ is a unitary element in $\MA(I')$ and hence commutes with $\MA(I)$. 
\end{proof}

\begin{df}
Throughout this paper, we let $\pmb\sigma\in\Aut(\Maa)$ be defined by
\begin{align*}
\sigma:\MH_0\otimes\MH_0\rightarrow\MH_0\otimes\MH_0\qquad\xi\otimes\eta\mapsto\eta\otimes\xi
\end{align*} 
\end{df}

\begin{df}\label{lb52}
For each $\MH_i\in\RepA$, define $\pmb{\MH_{\sqrt i}}\in\Rep^\sigma(\Maa)$ as follows. As a Hilbert space, $\MH_{\sqrt i}$ is identical to $\MH_i$. For each $\wtd I\in\Jtd$, $\pi_{\sqrt i,\wtd I}:\Maa(I)\rightarrow\fk L(\MH_i)$ is the unique normal representation such that
\begin{align}\label{eq24}
\pi_{\sqrt i,\wtd I}(x^+\otimes 1)=\pi_{i,\sqrt I}\circ\mbf Q_{\sqrt{\wtd I}}(x^+)\qquad \pi_{\sqrt i,\wtd I}(1\otimes x^-)=\pi_{i,-\sqrt I}\circ\mbf Q_{-\sqrt{\wtd I}}(x^-)
\end{align}
for each $x^\pm\in\MA(I)$.
\end{df}

\begin{proof}
We need to prove the existence of such $\MH_{\sqrt i}$; the uniqueness is obvious. By the split property (cf. Rem. \ref{lb1}), there exists a unitary $\Phi:\MH_0\rightarrow\MH_0\otimes\MH_0$ such that
\begin{align*}
\Phi x^+\Phi^{-1}=x^+\otimes 1\qquad\Phi x^-\Phi^{-1}=1\otimes x^-\qquad \text{for each }x^\pm\in\MA(\pm\sqrt I)
\end{align*}
Let $g_\pm$ be as in Def. \ref{lb50}. Define $\pi_{\sqrt 0,\wtd I}:\Maa(I)\rightarrow\fk L(\MH_0)$ by
\begin{gather*}
\pi_{\sqrt 0,\wtd I}=\Ad_{\Phi^{-1}}\circ\Ad_{U_0(g_+)\otimes U_0(g_-)}
\end{gather*}
Then, for each $x^\pm\in\MA(I)$, we have
\begin{align*}
\Ad_{U_0(g_+)\otimes U_0(g_-)}(x^+\otimes 1)=\mbf Q_{\sqrt{\wtd I}}(x^+)\otimes 1\quad \Ad_{U_0(g_+)\otimes U_0(g_-)}(1\otimes x^-)=1\otimes\mbf Q_{-\sqrt{\wtd I}}(x^-)
\end{align*}
and hence
\begin{align*}
\pi_{\sqrt 0,\wtd I}(x^+\otimes 1)=\mbf Q_{\sqrt{\wtd I}}(x^+)\qquad \pi_{\sqrt 0,\wtd I}(1\otimes x^-)=\mbf Q_{-\sqrt{\wtd I}}(x^-)
\end{align*}
In particular, $\pi_{\sqrt 0,\wtd I}$ ranges in $\MA(\sqrt I)\vee\MA(-\sqrt I)$. Choose any $K\in\MJ$ containing $\sqrt I$ and $-\sqrt I$. Then $\pi_{\sqrt i,\wtd I}$ can be defined by
\begin{align*}
\pi_{\sqrt i,\wtd I}=\pi_{i,K}\circ\pi_{\sqrt 0,\wtd I}
\end{align*}
We have thus constructed $(\MH_{\sqrt i},\pi_{\sqrt i})$ as a solitonic $\Maa$-module. By using \eqref{eq24}, one easily checks that it is $\sigma$-twisted.
\end{proof}

\subsubsection{The separate conformal structures on $\MH_{\sqrt i}$}

\begin{df}\label{lb126}
For each $\wtd I\in\Jtd$ and $g\in\Diff_I(\Sbb^1)$, define $\pmb{\sqrt g}\in\Diff_{\sqrt I}(\Sbb^1)$ and $\pmb{-\sqrt g}\in\Diff_{-\sqrt I}(\Sbb^1)$ by
\begin{align*}
\sqrt g=g_+\cdot g\cdot g_+^{-1}\qquad -\sqrt g=g_-\cdot g\cdot g_-^{-1}
\end{align*}
where $g_\pm$ are as in Def. \ref{lb50}. This definition is independent of the choice $g_\pm$. Indeed, $\sqrt g$ is the unique element in $\Diff_{\sqrt I}(\Sbb^1)$ whose restriction to $\sqrt I$ is the pullback of $g|_I:I\rightarrow I$ along the squaring map $\sqrt I\rightarrow I$, and $-\sqrt g$ is the unique element in $\Diff_{-\sqrt I}(\Sbb^1)$ whose restriction to $-\sqrt I$ is the pullback of $g|_I:I\rightarrow I$ along the squaring map $-\sqrt I\rightarrow I$. 
\end{df}

One easily checks that $-\sqrt g$ is the pullback of $\sqrt g$ along the map $z\mapsto -z$, i.e.,
\begin{align}
(-\sqrt g)(z)=-(\sqrt g(-z))\qquad\text{for each }z\in\Sbb^1
\end{align}

\begin{df}
For each $\wtd I\in\Jtd$ and $g\in\SG(I)$, the elements $\pmb{\sqrt g}\in\SG(\sqrt I)$ and $\pmb{-\sqrt g}\in\SG(-\sqrt I)$ are defined by lifting the elements $\pm\sqrt{\kappa(g)}\in \Diff_{\pm\sqrt I}(\Sbb^1)$ to $\SG(\pm\sqrt I)$ via the covering map $\kappa:\SG\rightarrow\DiffS$.
\end{df}

\begin{df}\label{lb54}
For each $\wtd I\in\Jtd$ and $\wtd g=(g,V)$ in $\GA(I)$, we define $\pmb{\sqrt{\wtd g}}\in\GA(\sqrt I)$ and $\pmb{-\sqrt{\wtd g}}\in\GA(-\sqrt I)$ by
\begin{align*}
\sqrt{\wtd g}=\big(\sqrt g,\Qbf_{\sqrt{\wtd I}}(V)\big)\qquad -\sqrt{\wtd g}=\big(-\sqrt g,\Qbf_{-\sqrt{\wtd I}}(V)\big)
\end{align*}
In particular, we have
\begin{align*}
U_0(\sqrt{\wtd g})=\Qbf_{\sqrt{\wtd I}}(U_0(\wtd g))\qquad U_0(-\sqrt{\wtd g})=\Qbf_{-\sqrt{\wtd I}}(U_0(\wtd g))
\end{align*}
This definition is well-defined, since one checks easily that $\Qbf_{\pm\sqrt{\wtd I}}(V)$ represents $U_0(\pm\sqrt g)$.
\end{df}

Note that both $\pm\sqrt g$ and $\pm\sqrt{\wtd g}$ depend on $\arg_I$. On the other hand, it is clear that under the identification $\GA=\Gc$, the definition of $\pm\sqrt{\wtd g}$ is independent of the conformal net $\MA$ extending $\Vir_c$.

\begin{pp}\label{lb66}
Let $\wtd I\in\Jtd$ and $\wtd g\in\GA(I)$. Then for each $\MH_i\in\RepA$ we have
\begin{align*}
U_i(\sqrt {\wtd g})=\pi_{\sqrt i,\wtd I}(U_0(\wtd g)\otimes 1)\qquad U_i(-\sqrt{\wtd g})=\pi_{\sqrt i,\wtd I}(1\otimes U_0(\wtd g))
\end{align*}
\end{pp}

\begin{proof}
By Thm. \ref{lb12} and Def. \ref{lb54},
\begin{align*}
U_i(\sqrt {\wtd g})=\pi_{i,\sqrt I}(U_0(\sqrt{\wtd g}))=\pi_{i,\sqrt I}\big(\Qbf_{\sqrt{\wtd I}}(U_0(\wtd g))\big)\xlongequal{\text{Def.\ref{lb52}}}\pi_{\sqrt i,\wtd I}(U_0(\wtd g)\otimes 1)
\end{align*}
A similar argument proves the second identity.
\end{proof}

\subsubsection{The state space of the closed CFT}

\begin{df}\label{lb155}
Let $\Zbb_2\leq\Aut(\Maa)$ be generated by $\sigma$. Using the fusion product in $\Rep^{\Zbb_2}(\Maa)$, we define the untwisted $\Maa$-module
\begin{align*}
\pmb{\MHcl}:=\MH_{\sqz}\boxtimes\MH_{\ovl\sqz}
\end{align*}
where $\MH_{\ovl\sqz}$ is the $\Zbb_2$-crossed conjugate module of $\MH_\sqz$.
\end{df}

\begin{rem}
The $\sigma$-twisted module $\MH_{\ovl\sqz}$ is indeed unitarily isomorphic to $\MH_\sqz$, as one can check (by applying Def. \ref{lb51} to $g_\pm$) that $\fk J_\MA$ gives an anti-unitary anti-automorphism of $\MH_\sqz$. However, we will not need this fact in the paper.
\end{rem}

\subsection{Double cones in the Einstein cylinder $\Rcl$}

\subsubsection{The Minkowski spaces}\label{lb65}

Recall that $\pmb{\Rbb^{1,1}}$ is the two dimensional Minkowski space, i.e., $\Rbb^{1,1}=\Rbb\times\Rbb$ with Minkowski metric $(dt)^2-(dx)^2$, and $x,t$ denote the first and second components of $\Rbb^{1,1}$. Let
\begin{align*}
\pmb{\Rcl}=(\Rbb/2\pi\Zbb)\times\Rbb\qquad \pmb{\Rpp}=(-\pi,\pi)\times\Rbb
\end{align*}
where $\Rcl$ is viewed as a quotient space of $\Roo$, and $\Rpp$ is viewed as an open subset of both $\Roo$ and $\Rcl$. For each $x\in\Rbb$, let
\begin{align}\label{eq29}
\pmb{\SL_x}=(x+2\pi\Zbb)\times\Rbb
\end{align}
which is a vertical line in $\Rcl$. Then
\begin{align*}
\Rpp\simeq \Rcl\setminus\SL_\pi=\Rcl\setminus\SL_{-\pi}
\end{align*}

\subsubsection{The $u^+u^-$-coordinates}

\begin{df}
We introduce a new set of coordinates $\pmb{(u^+,u^-)}$ on $\Roo$ (and hence on $\Rpp$) by
\begin{align*}
\left\{
\begin{array}{l}
u^+=t+x\\[0.5ex]
u^-=t-x
\end{array}
\right.\qquad\text{equivalently}\qquad
\left\{
\begin{array}{l}
x=\frac{u^+-u^-}2\\[0.5ex]
t=\frac{u^++u^-}2
\end{array}
\right.
\end{align*}
In this coordinate, $(u^+,u^-)$ and $(u^++2n\pi,u^--2n\pi)$ correspond to the same point of $\Rcl$ when $n\in\Zbb$, that is,
\begin{align}\label{eq25}
\Rcl=\Roo/(2\pi,-2\pi)\Zbb\qquad\text{in the }u^+u^-\text{ coordinates}
\end{align}
\end{df}

\begin{df}
Two points $p_1,p_2$ of $\Rcl$ are called \textbf{spacelike separated} if, in the $u^+u^-$-coordinates, they can be represented by $(u^+_1,u^-_1)$ and $(u^+_2,u^-_2)$ such that
\begin{subequations}\label{eq32}
\begin{align}\label{eq32a}
0<u^+_2-u^+_1<2\pi\qquad 0<u^-_1-u^-_2<2\pi
\end{align}
equivalently (by \eqref{eq25} and a different choice of representation),
\begin{align}\label{eq32b}
0<u^+_1-u^+_2<2\pi\qquad 0<u^-_2-u^-_1<2\pi
\end{align}
\end{subequations}
For each $O\subset\Rcl$, the interior
\begin{align*}
\pmb{O'}:=\Int\{q\in\Rcl:q\text{ is spacelike separated with any }p\in O\}
\end{align*}
is called the \textbf{spacelike complement} of $O$.
\end{df}

\begin{rem}\label{lb76}
Suppose that $p_1,p_2\in\Roo_{(-\pi,\pi]}=(-\pi,\pi]\times\Rbb$ are spacelike separated. Then, by viewing $\Roo_{(-\pi,\pi]}$ as a subset of $\Roo$, precisely one of \eqref{eq32} holds true. In fact, \eqref{eq32a} holds iff $x_1<x_2$, and \eqref{eq32b} holds iff $x_2<x_1$.
\end{rem}

\begin{proof}
Clearly $x_1\neq x_2$, otherwise $u_2^+-u_1^+$ and $u_2^--u_1^-$ would have the same sign, contradicting \eqref{eq32}. 

Suppose $x_1<x_2$. Choose $(w_2^+,w_2^-)\in (u_2^+,u_2^-)+(2\pi,-2\pi)\Zbb$ such that 
\begin{align}\label{eq33}
0<w^+_2-u^+_1<2\pi\qquad 0<u^-_1-w^-_2<2\pi
\end{align}
Let $\wtd x_2$ be the $x$-coordinate of $(w_2^+,w_2^-)$. (Note that $(w_2^+,w_2^-)$ and $(u_2^+,u_2^-)$ have the same $t$-coordinate $t_2$.) Then $x_1<\wtd x_2<x_1+2\pi$, because
\begin{align*}
(w_2^+-w_2^-)-(u_1^+-u_1^-)=(w^+_2-u^+_1)+(u^-_1-w^-_2)
\end{align*}
is $>0$ and $<4\pi$.

Since $x_1>-\pi$, we have $\wtd x_2>-\pi$. We claim that $\wtd x_2\leq\pi$. If not, then $\pi<\wtd x_2<x_1+2\pi$, and hence
\begin{align*}
-\pi<\wtd x_2-2\pi<x_1
\end{align*}
In particular, $-\pi<\wtd x_2-2\pi\leq\pi$. Thus, the unique representation of $p_2$ in $\Roo_{(-\pi,\pi]}$ is given by $(\wtd x_2-2\pi,t_2)$. Therefore $x_2=\wtd x_2-2\pi$, and hence $x_2<x_1$, contradicting the assumption $x_1<x_2$.

We have thus proved that $-\pi<\wtd x_2\leq\pi$. Then $(\wtd x_2,t_2)$ is the unique representation of $p_2$ in $\Roo_{(-\pi,\pi]}$. Thus $w_2^+=u_2^+$ and $w_2^-=u_2^-$, and hence \eqref{eq32a} follows from \eqref{eq33}. A similar argument shows that if $x_2<x_1$ then \eqref{eq32b} follows.
\end{proof}

A pictorial proof of Rem. \ref{lb76} is given by Fig. \ref{fig1}, where the parts of the spacelike complement of $p=p_1\in \Roo_{(-\pi,\pi]}$ satisfying \eqref{eq32a} and \eqref{eq32b} are illustrated by the yellow and green regions, respectively.

\begin{figure}[h]
	\centering
\begin{gather*}
\includegraphics[height=4cm]{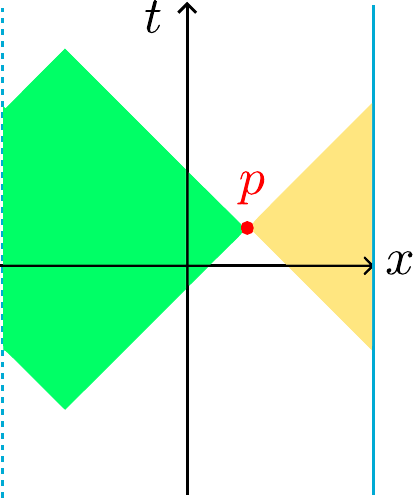}
   \end{gather*}
\caption{.~~The spacelike complement of $p$ in $\Rcl$, projected onto $\Roo_{(-\pi,\pi]}$}\label{fig1}
\end{figure}

\subsubsection{Double cones}

\begin{df}
For each element $p=(u^+,u^-)$ of $\Roo$ (in the $u^+u^-$ coordinate), the associated ordered pair of arg-valued points $(e^{\im u^+},e^{\im u^-})$ with arguments $u^+,u^-$ is called the corresponding pair of $p$ in the \textbf{double circle picture}.
\end{df}

\begin{df}
A subset $O\subset\Rcl$ is called a \textbf{double cone} if it is the image of an open rectangle $K^+\times K^-$ (in the $u^+u^-$ coordinate) under the quotient map $\Roo\rightarrow \Rcl$, where the lengths of the open intervals $K^\pm\subset\Rbb$ satisfy
\begin{align*}
0<|K^+|,|K^-|<2\pi
\end{align*}
Equivalently, a double cone is an open subset of $\Rcl$ corresponding to a pair $(\wtd I,\wtd J)$ in the double circle picture, where $\wtd I,\wtd J\in\Jtd$. 
\end{df}

\begin{df}
The \textbf{top}, \textbf{bottom}, \textbf{left}, and \textbf{right corners} of $O$ are understood in the usual sense. That is, if in the $u^+u^-$-coordinates we have $O=K^+\times K^-$, by setting $a^\pm=\inf K^\pm$ and $b^\pm=\sup K^\pm$, these four corners are respectively the images of the following points under the covering map $\Roo\rightarrow\Rcl$.
\begin{align*}
(b^+,b^-)\qquad (a^+,a^-)\qquad (a^+,b^-)\qquad (b^+,a^-)
\end{align*} 
\end{df}

\begin{rem}\label{lb55}
Note that infinitely many pairs $(\wtd I,\wtd J)$ correspond to a given double cone. However, if $O$ is a double cone in $\Rpp$, since the $u^+u^-$-coordinates on $\Rpp$ are genuine coordinates, $O$ corresponds to a unique pair $(\wtd I,\wtd J)$ in the double circle picture, i.e., the one satisfying
\begin{align}\label{eq26}
|\arg_I z-\arg_J\zeta|<2\pi\qquad\text{for each }z\in I,\zeta\in J
\end{align}
Conversely, any $(\wtd I,\wtd J)$ satisfying \eqref{eq26} corresponds to a double cone $O\subset\Rpp$. 
\end{rem}

\begin{cv}\label{lb122}
Unless otherwise specified, the pair $(\wtd I,\wtd J)$ corresponding to a double cone $O\subset\Rpp$ is chosen to satisfy \eqref{eq26}. 
\end{cv}

\begin{eg}
Let $O\subset\Rcl$ be a double cone. Then $O'$ is also a double cone. Indeed, if $(\wtd I,\wtd J)$ corresponds to $O$, then both $(\wtd I',\bpr\wtd J)$ and $(\bpr\wtd I,\wtd J')$ correspond to $O'$.
\end{eg}

\begin{df}
A double cone $O\subset\Rcl$ is called \textbf{equilateral} if it corresponds to some $(\wtd I,\wtd J)$ in the double circle picture with $|I|=|J|$. A double cone is called \textbf{standard} if it corresponds to $(\wtd I,\wtd I)$ for a (necessarily unique) $\wtd I\in\Jtd$. 
\end{df}

Standard double cones must be compactly contained in $\Rpp$. Indeed, one checks easily that, in the $xt$-coordinates, a standard double cone is precisely a double cone $O\subset\Rpp$ whose top and bottom corners lie in the $t$-axis.

\begin{rem}\label{lb56}
Let $O\subset\Rcl$ be a double cone corresponding to $(\wtd I,\wtd J)$. Then 
\begin{align*}
O\Subset\Rpp\qquad\Longleftrightarrow\qquad\text{there exists $\wtd K\in\Jtd$ containing $\wtd I,\wtd J$}
\end{align*}
In other words, $O\Subset\Rpp$ iff $O$ is contained in some standard double cone.
\end{rem}

\begin{proof}
This is due to \eqref{eq26}, which implies that when $\wtd I,\wtd J$ are viewed as intervals in $\Rbb$, the property $O\Subset\Rpp$ holds iff the diameter of $\wtd I\cup\wtd J$ is $<2\pi$.
\end{proof}

\begin{df}\label{lb109}
We let
\begin{gather*}
\pmb{\MO}=\{\text{double cones }O\subset\Rcl\}\qquad \pmb{\Opp}=\{\text{double cones }O\Subset\Rpp\}
\end{gather*}
\end{df}

\subsubsection{Conformal symmetry}

\begin{df}
We define a group action of $\SG\times\SG$ on $\Rcl$, which lifts to a group action of $\Gc\times\Gc$ on $\Rcl$. For each $g^+,g^-\in\SG$, by viewing them as diffeomorphisms of $\Rbb$ (cf. \eqref{eq2}), we let
\begin{align*}
(g^+,g^-)(u^+,u^-)=(g^+(u^+),g^-(u^-))
\end{align*}
Thanks to \eqref{eq25} and the relation $g^\pm(u+2\pi)=g^\pm(u)+2\pi$, this defines a well-defined action of $\SG\times\SG$ on $\Rcl$.
\end{df}

\begin{rem}
Although we will not use this fact, we note that $\SG\times\SG$ is the universal cover of the group $\mathrm{Cf}^+(\Rcl)$ of conformal diffeomorphisms of $\Rcl$ preserving both the orientation of the manifold $\Rcl$ and the direction of $t$. Moreover, the kernel of the covering map $\SG\times\SG\rightarrow \mathrm{Cf}^+(\Rcl)$ is isomorphic to $\Zbb$, generated by $(\varrho(2\pi),\varrho(-2\pi))$.
\end{rem}

\subsection{The bulk net $\Bcl$ acting on $\MHcl$}

\begin{df}\label{lb121}
For each $I,J\in\MJ$ whose closures are disjoint, let
\begin{align*}
\pmb{\fk A(I,J)}=\MA(I^\vee)'\cap\MA(J^\vee)'
\end{align*}
where $I^\vee,J^\vee$ are the two components of the interior of $\Sbb^1\setminus(I\cup J)$.
\end{df}

For example, if $\wtd K\in\Jtd$ and $I=\sqrt K,J=-\sqrt K$, then $I^\vee,J^\vee$ can be chosen to be $\sqrt{K'},-\sqrt{K'}$.

\subsubsection{The net $\Bcl$  defined on $\Opp$}

Recall Thm. \ref{lb35} for the normal representation $\pi^L_{\wtd I}$, which will be applied in this subsection to crossed $\Maa$-module.

For each $\wtd I\in\Jtd$, by the construction of $\MH_\sqz$ in Subsec. \ref{lb57}, we have
\begin{align}\label{eq27}
\End_{\Maa(\wtd I')}(\MH_\sqz)=\MA(\sqrt{I'})'\cap\MA(-\sqrt{I'})'=\fk A(\sqrt I,-\sqrt I)
\end{align}
Consequently, if $\wtd I,\wtd J\subset\wtd K$, then
\begin{align}\label{eq70}
\fk A(\sqrt I,-\sqrt J)\subset \End_{\Maa(\wtd K')}(\MH_\sqz)
\end{align}
where $\sqrt{I}$ and $-\sqrt{J}$ are defined using $\wtd I$ and $\wtd J$, respectively.

\begin{df}\label{lb64}
Let $O\in\Opp$. Let $(\wtd I,\wtd J)$ be the unique pair satisfying \eqref{eq26} that corresponds to $O$. By Rem. \ref{lb56}, there exists a standard double cone $D$ containing $O$. Choose $\wtd K\in\Jtd$ such that $(\wtd K,\wtd K)$ represents $D$, and so $\wtd I,\wtd J\subset\wtd K$. Define the von Neumann algebra $\Bcl(O)$ on $\Hcl=\MH_\sqz\boxtimes\MH_{\ovl\sqz}$ by
\begin{align*}
\Bcl(O)=\pi^L_{\wtd K}\big(\fk A(\sqrt I,-\sqrt J)\big)\big|_{\MH_\sqz\boxtimes\MH_{\ovl\sqz}}
\end{align*}
In particular, if $O$ is standard (and hence $\wtd I=\wtd J$), then
\begin{align}\label{eq50}
\Bcl(O)=\pi^L_{\wtd I}\big(\End_{\Maa(\wtd I')}(\MH_\sqz)\big)\big|_{\MH_\sqz\boxtimes\MH_{\ovl\sqz}}
\end{align}
\end{df}

\begin{proof}[Explanation]
This definition is independent of the choice of $D$ containing $O$. In fact, for two such double cones $D_1,D_2$ corresponding to $(\wtd K_1,\wtd K_1)$ and $(\wtd K_2,\wtd K_2)$ respectively, choose $\wtd K_0\in\Jtd$ containing $\wtd I,\wtd J$ and contained in $\wtd K_1,\wtd K_2$. By Rem. \ref{lb61},
\begin{align*}
\pi^L_{\wtd K_i}\big(\fk A(\sqrt I,-\sqrt J)\big)\big|_{\MH_\sqz\boxtimes\MH_{\ovl\sqz}}=\pi^L_{\wtd K_0}\big(\fk A(\sqrt I,-\sqrt J)\big)\big|_{\MH_\sqz\boxtimes\MH_{\ovl\sqz}}
\end{align*}
for $i=1,2$.
\end{proof}

\subsubsection{The net $\Bcl$ defined on $\MO$}

\begin{df}\label{lb74}
The net of von Neumann algebras
\begin{align*}
\pmb{\Bcl}:O\in\MO\mapsto \Bcl(O)\subset\fk L(\Hcl)
\end{align*}
is defined as follows. If $O\in\Opp$, then $\Bcl(O)$ is defined as in Def. \ref{lb64}. If $O\notin\Opp$, that is, if the vertical line $\SL_\pi$ defined in \eqref{eq29} intersects the closure of $O$, by viewing $\Rpp=\Rcl\setminus\SL_\pi$, we let
\begin{align}\label{eq31}
\Bcl(O)=\bigvee_{D\in\Opp,D\subset O}\Bcl(D)
\end{align}
\end{df}

\begin{df}\label{lb134}
For each $g^+,g^-\in\Gc$, define $\pmb{\Ucl(g^+,g^-)}\in\MU(\Hcl)$ by
\begin{align*}
\Ucl(g^+,g^-)=\Ucl^+(g^+)\Ucl^-(g^-)
\end{align*}
where $\Ucl^+,\Ucl^-:\Gc\rightarrow\MU(\MHcl)$ are defined in Def. \ref{lb53} (for $\Jbb=\cl$). Note that by Rem. \ref{lb104}, the ranges of $\Ucl^+$ and $\Ucl^-$ commute.  Clearly
\begin{align*}
\Ucl:\Gc\times\Gc\rightarrow\MU(\MHcl)
\end{align*}
is a strongly-continuous unitary representation.
\end{df}

The following observation will be crucial to the proof of the conformal covariance of $\Bcl$.

\begin{rem}\label{lb67}
By Def. \ref{lb53}, if $I\in\MJ$ and $g^+,g^-\in\Gc(I)$, then
\begin{align}\label{eq55}
\Ucl(g^+,g^-)=\pi_{\sqz\boxtimes\ovl\sqz,I}\big(U_0(g^+)\otimes U_0(g^-)\big)
\end{align}
By Prop. \ref{lb66}, if $\wtd I\in\Jtd$ and $g^+,g^-\in\Gc(I)$, noting that $\sqrt{g^+}\in\Gc(\sqrt I)$ and $-\sqrt{g^-}\in\Gc(-\sqrt I)$ are defined in Def. \ref{lb54} and rely on $\arg_I$, we have
\begin{align}\label{eq54}
U_0\big(\sqrt{g^+}\cdot(-\sqrt{g^-})\big)=\pi_{\sqz,\wtd I}\big(U_0(g^+)\otimes U_0(g^-)\big)
\end{align}
\end{rem}

Recall that $\rk$ acts on the universal cover $\Rbb$ of $\Sbb^1$ by $u\mapsto -u$. Thus $(\rk,\rk)$ acts on $\Roo$ and hence acts on $\Rcl$:
\begin{align}\label{eq44}
\pmb{(\rk,\rk)}:\Rcl\rightarrow\Rcl\qquad (u^+,u^-)\mapsto(-u^+,-u^-)
\end{align}
Recall Def. \ref{lb83} for the notation $\pi_\Jbb^\pm$, which we will apply below in the special case $\Jbb=\cl\equiv\sqz\boxtimes\ovl\sqz$.

\begin{thm}\label{lb72}
The net $\Bcl$ satisfies the following properties for each $O,O_1,O_2\in\MO$.
\begin{enumerate}[label=(\arabic*)]
\item (Isotony) If $O_1\subset O_2$, then $\Bcl(O_1)\subset\Bcl(O_2)$.
\item (Locality) If $O_1$ and $O_2$ are spacelike separated, then $\Bcl(O_1)$ commutes with $\Bcl(O_2)$.
\item (Extension property) If $O$ corresponds to $(\wtd I,\wtd J)$, then
\begin{align*}
\pi^+_{\sqz\boxtimes\ovl\sqz,I}(\MA(I))\vee \pi^-_{\sqz\boxtimes\ovl\sqz,J}(\MA(J))\subset \Bcl(O)
\end{align*}
\item (Conformal covariance) For each $g^+,g^-\in\Gc$, we have
\begin{align*}
\Ucl(g^+,g^-)\Bcl(O)\Ucl(g^+,g^-)^*=\Bcl((g^+,g^-)O)
\end{align*}
Moreover, if $O$ corresponds (non-uniquely) to some $(\wtd I,\wtd J)$ in the double circle picture, and if $g^+\in\Gc(I),g^-\in\Gc(J)$, then
\begin{align*}
\Ucl(g^+,g^-)\in\pi^+_{\sqz\boxtimes\ovl\sqz,I}(\MA(I))\vee \pi^-_{\sqz\boxtimes\ovl\sqz,J}(\MA(J))
\end{align*}
and hence $\Ucl(g^+,g^-)\in\Bcl(O)$ by the extension property.
\item (PCT symmetry) There exists an anti-unitary operator $\Jcl$ on $\Hcl$ satisfying
\begin{gather*}
\Jcl^2=1\qquad \Jcl\Ucl(g^+,g^-)\Jcl=U_\cl(\rk g^+\rk,\rk g^-\rk)\\
\Jcl\Bcl(O)\Jcl=\Bcl((\rk,\rk)O)
\end{gather*}
for each $g^+,g^-\in\Gc$ and $O\in\MO$.
\item (Positive energy) The self-adjoint generators of the one-parameter groups $u\mapsto \Ucl(\varrho_c(u),1)$ and $u\mapsto \Ucl(1,\varrho_c(u))$ are positive. Here, $1$ denotes the identity element of $\Gc$. 
\item (Irreducibility) $\Bcl(O)$ is a type III factor. Moreover, $\bigcup_{O\in\MO}\Bcl(O)$ generates the von Neumann algebra $\fk L(\Hcl)$.
\end{enumerate}
\end{thm}

\begin{proof}
Isotony: If $O_1,O_2\in\Opp$, the inclusion is obvious. If $O_1\in\Opp,O_2\notin\Opp$, this follows by choosing $O=O_2,D=O_1$ in \eqref{eq31}. The case $O_1,O_2\notin\Opp$ follows from \eqref{eq31} and the case $O_1,O_2\in\Opp$.

Locality: See Thm. \ref{lb78} in Sec. \ref{lb80}.

Extension property: See Cor. \ref{lb135} in Sec. \ref{lb100}.

Conformal covariance: See Thm. \ref{lb98} and Cor. \ref{lb99} in Sec. \ref{lb100}. 

PCT symmetry: See Thm. \ref{lb102} in Sec. \ref{lb103}.

Positive energy: By \cite{Wei06}, for each (normal) representation of $\MA$ the generator of rotation is positive. Apply this result to $\pi^\pm_\cl$.

Irreducibility: In the special case where $O$ is standard and hence corresponds to $(\wtd I,\wtd I)$ for some $\wtd I\in\Jtd$, $\Bcl(O)$ is described by \eqref{eq50}, and hence is a type III factor by Rem. \ref{lb91}. The general case reduces to this special case by the conformal covariance mentioned above. In particular, since $\Bcl(O)$ is a factor, $\Bcl(O)\cup\Bcl(O)'$ generates $\fk L(\Hcl)$. By the Haag duality (Thm. \ref{lb115}), $\Bcl(O)\cup\Bcl(O')$ generates $\fk L(\MH_\cl)$. Thus $\bigcup_{O\in\MO}\Bcl(O)$ generates $\fk L(\Hcl)$. 
\end{proof}

The following Sections \ref{lb144}--\ref{lb103} are devoted to the proof of Thm. \ref{lb72}, together with several important byproducts of that proof (such as additivity, cf. Thm. \ref{lb146}). Finally, in Sec. \ref{lb145}, we formulate and prove Haag duality for multi-double-cones, a result that is not included in Thm. \ref{lb72}.

\subsection{Local conformal covariance for $\Bcl|_{\Opp}$}\label{lb144}

\begin{lm}\label{lb68}
Let $O\in\Opp$ be contained in some standard double cone $\wht O$. Assume that $\wht O$ corresponds to $(\wtd K,\wtd K)$ in the double circle picture. Then for each $g^+,g^-\in\GA(K)$ we have
\begin{align*}
\Ucl(g^+,g^-)\Bcl(O)\Ucl(g^+,g^-)^*=\Bcl((g^+,g^-)O)
\end{align*}
\end{lm}

\begin{proof}
Assume that $O$ corresponds to $(\wtd I,\wtd J)$. Then $\wtd I,\wtd J\subset\wtd K$. Let $\sqrt{g^+},-\sqrt{g^-}$ be defined by $\wtd K$. By Rem. \ref{lb67},
\begin{align}
&\Ad_{\Ucl(g^+,g^-)}\big(\Bcl(O)\big)=\Ad_{\Ucl(g^+,g^-)}\Big(\pi^L_{\wtd K}\big(\fk A(\sqrt I,-\sqrt J)\big)\Big)\nonumber\\
=&\Ad_{\pi_{\sqz\boxtimes\ovl\sqz,K}(U_0(g^+)\otimes U_0(g^-))}\Big(\pi^L_{\wtd K}\big(\fk A(\sqrt I,-\sqrt J)\big)\Big)\nonumber
\end{align}
By \eqref{eq28}, the above equals
\begin{align*}
&\Ad_{\pi^L_{\wtd K}(\pi_{\sqz,\wtd K}(U_0(g^+)\otimes U_0(g^-)))}\Big(\pi^L_{\wtd K}\big(\fk A(\sqrt I,-\sqrt J)\big)\Big)\\
=&\pi^L_{\wtd K}\Big(\Ad_{\pi_{\sqz,\wtd K}(U_0(g^+)\otimes U_0(g^-))}\big(\fk A(\sqrt I,-\sqrt J)\big)\Big)\\
=&\pi^L_{\wtd K}\Big(\Ad_{U_0\big(\sqrt{g^+}\cdot\big(-\sqrt{g^-}\big)\big)}\big(\fk A(\sqrt I,-\sqrt J)\big)\Big)
\end{align*}
where Rem. \ref{lb67} is used in the last equality. Note that $\sqrt I\subset\sqrt K$, that $-\sqrt J\subset-\sqrt K$, and that $\sqrt{g^+}\cdot(-\sqrt{g^-})$ fixes points outside $\sqrt K\cup(-\sqrt K)$. 

Moreover, if $O_1:=(g^+,g^-)O$ corresponds to $(\wtd I_1,\wtd J_1)$, then we have $\wtd I_1,\wtd J_1\subset\wtd K$ and $g^+\wtd I=\wtd I_1$ and $g^-\wtd J=\wtd J_1$. Thus
\begin{align*}
\sqrt{g^+}\sqrt I=\sqrt{I_1}\qquad (-\sqrt{g^-})(-\sqrt J)=-\sqrt{J_1}
\end{align*}
Therefore, by the conformal covariance of $\MA$, we have
\begin{align*}
\Ad_{U_0\big(\sqrt{g^+}\cdot\big(-\sqrt{g^-}\big)\big)}\big(\fk A(\sqrt I,-\sqrt J)\big)=\fk A(\sqrt{I_1},-\sqrt{J_1})
\end{align*}
Inserting the right hand side above into $\pi^L_{\wtd K}$ yields $\Bcl(O_1)$. This finishes the proof.
\end{proof}

\begin{lm}\label{lb69}
Let $O\in\Opp$. Let $a,b>0$ such that the interval $(-a,b)$ is the connected component containing $0$ of the set
\begin{align*}
\big\{s\in\Rbb:(\varrho(s),\varrho(-s))O\Subset\Rpp\big\}
\end{align*}
Then for each $-a<s<b$ we have
\begin{align*}
\Ucl(\varrho_c(s),\varrho_c(-s))\cdot\Bcl(O)\cdot\Ucl(\varrho_c(s),\varrho_c(-s))^*=\Bcl((\varrho(s),\varrho(-s))O)
\end{align*}
\end{lm}

Note that $(\varrho(s),\varrho(-s))$ is the translation in the $x$-axis by $s$ units.

\begin{proof}
It suffices to prove the lemma in the special case where $s$ satisfies the extra condition $|s|\leq\eps$ for some $\eps>0$ (possibly depending on $O$). Indeed, once the lemma is established for such $s$, the general case follows by repeated application of this special case. The reduction of an arbitrary $s$ to a \textit{finite} composition of such sufficiently small translations follows from the Lebesgue number lemma.

Choose standard double cones $\wht O_1,\wht O_2,\wht O_3$ such that $O\Subset\wht O_1\Subset\wht O_2\Subset\wht O_3$. Let $\wht O_i,O$ correspond to $(\wtd K_i,\wtd K_i)$ and $(\wtd I,\wtd J)$ respectively. Then, by viewing $\Sbb^1=\Rbb/2\pi\Zbb$, there exists $\eps>0$ such that
\begin{align*}
\wtd K_1\Subset\wtd K_2\Subset\wtd K_3\qquad \ovl{\bigcup_{|s|\leq\eps}\varrho(s)\wtd I}\subset\wtd K_1\qquad \ovl{\bigcup_{|s|\leq\eps}\varrho(s)\wtd J}\subset\wtd K_1
\end{align*} 

We now consider two vector fields on $\Rbb$ as the universal cover of $\Sbb^1$. The first one is $\partial_u$, where $u$ denotes the standard coordinate of $\Rbb$. The second one is $f\partial_u$ where $f:\Rbb\rightarrow\Rbb$ is a smooth $2\pi$-periodic function which equals $1$ on $\wtd K_2$, and which vanishes on $E\setminus \wtd K_3$ where $E\subset\Rbb$ is any interval of length $2\pi$ containing $\wtd K_3$. Then, in view of the description \eqref{eq2} of $\SG$, the flow on $\Rbb$ generated by $\partial_u$ defines the rotation subgroup $\varrho$ of $\SG$, and the flow on $\Rbb$ generated by $f\partial_u$ defines a one-parameter subgroup $\gamma:\Rbb\rightarrow\SG$. We choose an arbitrary lift of $\gamma$ to $\gamma_c:\Rbb\rightarrow\Gc$.

By shrinking $\eps$, we have 
\begin{align*}
\rho(s)\wtd K_1\subset\wtd K_2\qquad\text{whenever } |s|\leq\eps
\end{align*}
We now fix $s$ satisfying $|s|\leq\eps$. Then
\begin{align*}
g^+:=\gamma_c(s)^{-1}\varrho_c(s)\in\Gc(K_1')\qquad g^-:=\gamma_c(-s)^{-1}\varrho_c(-s)\in\Gc(K_1')
\end{align*}
since $\gamma(\pm s)^{-1}\varrho(\pm s):\Rbb\rightarrow\Rbb$ fixes points of $\wtd K_1$, cf. Rem. \ref{lb129}. Thus, by Rem. \ref{lb67},
\begin{align*}
\Ucl(g^+,g^-)=\pi_{\sqz\boxtimes\ovl\sqz,K_1'}(U_0(g^+)\otimes U_0(g^-))\xlongequal{\eqref{eq30}}\pi^R_{\wtd K_1'}(\pi_{\ovl\sqz,\wtd K_1'}(U_0(g^+)\otimes U_0(g^-)))
\end{align*}
Since $\pi^L_{\wtd K_1}$ commutes with $\pi^R_{\wtd K_1'}$ in the sense of Thm. \ref{lb47}, the expression
\begin{align*}
\Ad_{\Ucl(g^+,g^-)}\big(\Bcl(O)\big)=\Ad_{\Ucl(g^+,g^-)}\Big(\pi^L_{\wtd K_1}\big(\fk A(\sqrt I,-\sqrt J)\big)\Big)
\end{align*}
equals $\pi^L_{\wtd K_1}\big(\fk A(\sqrt I,-\sqrt J)\big)$, and hence
\begin{align*}
\Ad_{\Ucl(g^+,g^-)}\big(\Bcl(O)\big)=\Bcl(O)
\end{align*}

On the other hand, since $\gamma(\pm s)\in\SG(K_3)$, we have $\gamma_c(\pm s)\in\Gc(K_3)$, and hence
\begin{align*}
\Ad_{\Ucl(\gamma_c(s),\gamma_c(-s))}\big(\Bcl(O)\big)=\Bcl((\gamma(s),\gamma(-s))O)=\Bcl((\varrho(s),\varrho(-s))O)
\end{align*}
by Lem. \ref{lb68}. Combining the above two relations together yields
\begin{align*}
\Ad_{\Ucl(\varrho_c(s),\varrho_c(-s))}\big(\Bcl(O)\big)=\Bcl((\varrho(s),\varrho(-s))O)
\end{align*}
for each $0\leq |s|\leq\eps$. 
\end{proof}

\begin{lm}\label{lb85}
Let $O\in\MO$. Then for each $s\in\Rbb$ we have
\begin{align*}
\Ucl(\varrho_c(s),\varrho_c(s))\cdot\Bcl(O)\cdot\Ucl(\varrho_c(s),\varrho_c(s))^*=\Bcl((\varrho(s),\varrho(s))O)
\end{align*}
\end{lm}

Note that $(\varrho(s),\varrho(s))$ is the translation in the $t$-axis by $s$ units.

\begin{proof}
In the case $O\in\Opp$, this is similar to the proof of Lem. \ref{lb69}, and hence is omitted. The case $O\notin\Opp$ reduces easily to the case $O\in\Opp$ by Def. \ref{lb74}. 
\end{proof}

\subsection{Some additivity properties}

\begin{pp}\label{lb62}
Let $G\rightarrow\Aut(\MA)$ be a group homomorphism. Let $\MH_i\in\RepGA$ and $\wtd I\in\Jtd$. Let $(\wtd I_\alpha)_{\alpha\in\scr A}$ be a collection in $\Jtd$ such that $\wtd I_\alpha\subset\wtd I$ and $\bigcup_\alpha I_\alpha=I$. Then
\begin{align*}
\End_{\MA(\wtd I')}(\MH_i)=\bigvee_\alpha \End_{\MA(\wtd I_\alpha')}(\MH_i) 
\end{align*}
\end{pp}

Clearly, the same relation holds with $\wtd I',\wtd I_\alpha'$ replaced by $\bpr\wtd I,\bpr\wtd I_\alpha$.

\begin{proof}
The relation $\supset$ is obvious. Let us prove $\subset$. Moreover, since the case $\MH_i=0$ is obvious, we assume $\MH_i\neq0$.

Step 1. We first consider the case in which there exists $\wtd I_0\in\Jtd$ such that $\wtd I_0\subset\wtd I_\alpha$ for all $\alpha$. Fix an $I_0$-unitary $\xi_0\in\MH_i(I_0)$, and let $V=L(\xi_0,\wtd I)|_{\MH_0}$. By Prop. \ref{lb34}, we have
\begin{align*}
\End_{\MA(\wtd I')}(\MH_i)=V\MA(I)V^*\qquad \End_{\MA(\wtd I'_\alpha)}(\MH_i)=V\MA(I_\alpha)V^*
\end{align*}
Therefore, the proposition follows from the additivity of $\MA$.

Step 2. We consider the case where $\scr A$ is a finite set, and prove the proposition by induction on $|\scr A|$. The case $|\scr A|=1$ is obvious, and the case $|\scr A|=n$ can be reduced to the case $|\scr A|=n-1$ by using Step 1 and the fact that at least two distinct members of $(I_\alpha)$ overlap (otherwise $I$ would be disconnected).

Step 3. We consider the case where $(\wtd I_\alpha)$ is an increasing sequence. But this case follows directly from Step 1.

Step 4. We consider the general case. Choose an increasing sequence $(\wtd K_n)$ in $\Jtd$ such that $\wtd K_n\subset \wtd I$ and $K_n\Subset I$ and $\bigcup_n K_n=I$. By Step 3, $\End_{\MA(\wtd I')}(\MH_i)$ is generated by all $\End_{\MA(\wtd K_n')}(\MH_i)$. Since $K_n$ has compact closure in $I$, there exists a finite set $\scr A_n\subset\scr A$ such that $K_n\subset\bigcup_{\alpha\in\scr A_n}I_\alpha$, and that each $I_\alpha$ intersects $K_n$ (so that $\bigcup_{\alpha\in\scr A_n}I_\alpha$ is an interval). By Step 2,  $\End_{\MA(\wtd K_n')}(\MH_i)$ is contained in the von Neumann algebra generated by $\End_{\MA(\wtd I_\alpha')}(\MH_i)$ for all $\alpha\in\scr A_n$. This proves the proposition.
\end{proof}

\begin{df}\label{lb70}
Let $O\in\Opp$. We consider the $xt$-coordinates. A subset $\Lambda\subset\Rpp$ is called an \textbf{arc between the top and bottom corners} of $O$ if the following property holds: By a translation and a (Euclidean) rotation so that the top, bottom, left, right corners of $O$ become the points $(0,b)$, $(a,0)$, $(0,0)$, $(a,b)$ with $a,b>0$, we have
\begin{align*}
\Lambda=\{(x,f(x)):0<x<a)\}
\end{align*}
where $f$ is a smooth function from the open interval $\{x:-\eps<x<a+\eps\}$ (where $\eps>0$) to $\Rbb$ satisfying
\begin{align*}
f(0)=b\qquad f(a)=0\qquad f'<0
\end{align*}
\end{df}

\begin{pp}\label{lb71}
Let $O\in\Opp$, and let $\Lambda$ be an arc between the top and bottom corners of $O$. Let $(O_\alpha)_{\alpha\in\scr A}$ be a collection in $\Opp$ such that $O_\alpha\subset O$ for each $\alpha$, and $\Lambda\subset\bigcup_\alpha O_\alpha$. Then
\begin{align}
\Bcl(O)=\bigvee_\alpha\Bcl(O_\alpha)
\end{align}
\end{pp}

In Thm. \ref{lb146}, this result will be generalized by dropping the assumption $O,O_\alpha\in\Opp$.

\begin{proof}
The relation $\supset$ is obvious. Let us prove $\subset$.

Choose a standard double cone corresponding to some $(\wtd K,\wtd K)$ and containing the closure of $O$. By using the function $f$ in Def. \ref{lb70}, one finds some $g\in\SG(K)$ such that $g(\wtd J)=\wtd I$, and that $(1,g)\Lambda$ is the vertical open interval between the top and bottom corners of the standard double cone $(1,g)O$. (This standard cone corresponds to $(\wtd I,\wtd I)$.) Therefore, by Lem. \ref{lb68}, one may assume at the beginning that $O$ is standard, and $\Lambda$ is the open interval between the top and bottom corners of $O$. In particular, $\Lambda$ lies in the $t$-axis.

For each $p\in\Lambda$, we choose a standard double cone $O_p\in\Opp$ containing $p$ and contained in $O_\alpha$ for some $\alpha$. Then it suffices to prove
\begin{align*}
\Bcl(O)=\bigvee_{p\in\Lambda}\Bcl(O_p)
\end{align*}
Let $(\wtd I,\wtd I)$ correspond to $O$ and $(\wtd I_p,\wtd I_p)$ correspond to $O_p$. Then $\wtd I_p\subset \wtd I$, and $\bigcup_p I_p=I$ by the fact that $\Lambda\subset\bigcup_p O_p$. By \eqref{eq27} and Prop. \ref{lb62}, we have
\begin{align*}
\fk A(\sqrt I,-\sqrt I)=\End_{\Maa(\wtd I')}(\MH_\sqz)=\bigvee_p \End_{\Maa(\wtd I_p')}(\MH_\sqz) =\bigvee_p\fk A(\sqrt{I_p},-\sqrt{I_p})
\end{align*}
Applying the normal representation $\pi^L_{\wtd I}(\cdots)|_{\MH_\sqz\boxtimes\MH_{\ovl\sqz}}$ to both sides yields the desired relation.
\end{proof}

Recall \eqref{eq29} for the meaning of $\SL_x$.

\begin{co}\label{lb75}
Let $O\in\MO$, let $E$ be a (possibly empty) finite subset of $[-\pi,\pi]$, and let $\SL_E=\bigcup_{x\in E}\SL_x$. Then for each collection $(O_\alpha)_{\alpha\in\fk A}$ in $\MO$ satisfying $O\setminus\SL_E=\bigcup_\alpha O_\alpha$, we have
\begin{align*}
\Bcl(O)=\bigvee_{\alpha\in\fk A}\Bcl(O_\alpha)
\end{align*}
\end{co}

\begin{proof}
The only nontrivial part is to prove $\subset$. We first consider the case $O\in\Opp$. By Prop. \ref{lb71}, $\Bcl(O)$ is generated by $\Bcl(D)$ for all equilateral double cones $D\subset O$. Let $D_\alpha=D\cap O_\alpha$, and ignore $\alpha$ if $D_\alpha=\emptyset$. Since $D$ is equilateral, there exists an arc $\Lambda$ between the top and bottom corners of $D$ such that $\Lambda\cap\SL_E=\emptyset$. Then $\Lambda\subset\bigcup_\alpha D_\alpha$, and hence $\Bcl(D)$ is generated by all $\Bcl(D_\alpha)$ due to Prop. \ref{lb71}. Thus $\Bcl(O)$ is generated by all $\Bcl(O_\alpha)$.

Next, assume $O\notin\Opp$. By Def. \ref{lb74}, $\Bcl(O)$ is generated by $\Bcl(D)$ for all $D\in\Opp$ satisfying $D\subset O$. By the above paragraph, $\Bcl(D)$ is generated by all $\Bcl(D_\alpha)$, where $D_\alpha:=D\cap O_\alpha$ can be viewed as in $\Opp$ because $D\in\Opp$. Thus, by the isotony in Thm. \ref{lb72}, $\Bcl(O)$ is again generated by all $\Bcl(O_\alpha)$.
\end{proof}

As an immediate application, we obtain the following improvement of Lem. \ref{lb69}.

\begin{lm}\label{lb86}
Let $O\subset\Rpp$ be a double cone. Let $a,b\geq0$ such that the interval $[-a,b]$ is the connected component containing $0$ of the set
\begin{align*}
\big\{s\in\Rbb:(\varrho(s),\varrho(-s))O\subset\Rpp\big\}
\end{align*}
Then for each $-a\leq s\leq b$ we have
\begin{align*}
\Ucl(\varrho_c(s),\varrho_c(-s))\cdot\Bcl(O)\cdot\Ucl(\varrho_c(s),\varrho_c(-s))^*=\Bcl((\varrho(s),\varrho(-s))O)
\end{align*}
\end{lm}

\begin{proof}
This is an immediate consequence of Lem. \ref{lb69} and Cor. \ref{lb75}.
\end{proof}

\subsection{Locality}\label{lb80}

\begin{lm}\label{lb77}
Suppose that $O_1,O_2\in\Opp$ have spacelike separated closures. Then there exist finite collections $(O_{1,\alpha})_{\alpha\in\scr A}$ and $(O_{2,\beta})_{\beta\in\scr B}$ of double cones in $\Opp$ such that
\begin{align*}
O_1=\bigcup_{\alpha\in\scr A} O_{1,\alpha}\qquad O_2=\bigcup_{\beta\in\scr B} O_{2,\beta}
\end{align*}
and that for each $\alpha\in\scr A,\beta\in\scr B$, the sets $O_{1,\alpha}$ and $O_{2,\beta}$ are compactly contained in a common standard double cone $\wht O_{\alpha,\beta}$.
\end{lm}

\begin{proof}
Step 1. We show that if $p_1,p_2\in\Rpp$ are spacelike separated, then $p_1,p_2$ are contained in a common standard double cone. In the $u^+u^-$-coordinates, write $p_1=(a_1,b_1)$ and $p_2=(a_2,b_2)$. Let $R$ be the diameter of $\{a_1,b_1,a_2,b_2\}$. If we can show that $R<2\pi$, then there exists an open interval in $\Rbb$ of length $<2\pi$ containing $a_1,b_1,a_2,b_2$. This interval corresponds to some $\wtd I\in\Jtd$, and $p_1,p_2$ belong to the double cone corresponding to $(\wtd I,\wtd I)$. 

Since the $x$-coordinates of $p_1,p_2$ are $>-\pi$ and $<\pi$, we have
\begin{align*}
|a_1-b_1|<2\pi\qquad |a_2-b_2|<2\pi
\end{align*}
By Rem. \ref{lb76}, we have
\begin{align*}
|a_1-a_2|<2\pi\qquad |b_1-b_2|<2\pi
\end{align*}
We claim that
\begin{align*}
|a_1-b_2|<R\qquad |a_2-b_1|<R
\end{align*}
which immediately implies $R<2\pi$. 

Suppose that $|a_1-b_2|=R$. Then either $a_1\leq \min\{a_2,b_1\}\leq\max\{a_2,b_1\}\leq b_2$, or $b_2\leq \min\{a_2,b_1\}\leq\max\{a_2,b_1\}\leq a_1$. In both cases, $a_2-a_1$ and $b_2-b_1$ have the same sign, contradicting Rem. \ref{lb76}. So we must have $|a_1-b_2|<R$. A similar argument shows $|a_2-b_1|<R$.\\[-1ex]

Step 2. Choose spacelike separated $D_1,D_2\in\Opp$ such that $O_1\Subset D_1$ and $O_2\Subset D_2$. By Step 1 and the compactness of $\ovl{O_2}$ (together with the Lebesgue number lemma), for each $p\in\ovl{O_1}$, there exists a double cone $D_{1,p}\subset D_1$ containing $p$, and some $\delta_p>0$, such that for each double cone $\Delta_2\subset D_2$ intersecting $\ovl{O_2}$ with diameter $<\delta_p$, then $D_{1,p}$ and $\Delta_2$ are compactly contained in a common standard double cone. 

Therefore, by the compactness of $\ovl{O_1}$ and $\ovl{O_2}$, there exist finite collections $(D_{1,\alpha})_{\alpha\in\scr A}$ and $(D_{2,\beta})_{\beta\in\SB}$ of equilateral double cones in $\Opp$ such that
\begin{align*}
\ovl{O_1}\subset\bigcup_{\alpha\in\scr A} D_{1,\alpha}\qquad \ovl{O_2}\subset\bigcup_{\beta\in\scr B} D_{2,\beta}
\end{align*}
and that for each $\alpha\in\scr A,\beta\in\scr B$, we have that $D_{1,\alpha}\subset D_1$, that $D_{2,\beta}\subset D_2$, and that $D_{1,\alpha}$ and $D_{2,\beta}$ are compactly contained in a common standard double cone $\wht D_{\alpha,\beta}$. The proof is finished by taking
\begin{align*}
O_{1,\alpha}=D_{1,\alpha}\cap O_1\qquad O_{2,\beta}=D_{2,\beta}\cap O_2
\end{align*}
and dropping those $O_{1,\alpha}$ and $O_{2,\beta}$ that are empty sets.
\end{proof} 


\begin{thm}\label{lb78}
Suppose that $O_1,O_2\in\MO$ are spacelike separated. Then $[\Bcl(O_1),\Bcl(O_2)]=0$.
\end{thm}

\begin{proof}
By Def. \ref{lb74}, it suffices to consider the case where $O_1,O_2\in\Opp$. By Cor. \ref{lb75} and Lem. \ref{lb77}, it suffices to assume that $O_1,O_2$ are contained in some standard double cone $\wht O$ corresponding to some $(\wtd K,\wtd K)$, and that $\ovl{O_1}$ and $\ovl{O_2}$ are spacelike separated. In fact, we will prove the commutativity without assuming spacelike separation of the closures; only the spacelike separation of $O_1,O_2$ will be used.

Let $O_1,O_2$ correspond to $(\wtd I_1,\wtd J_1)$ and $(\wtd I_2,\wtd J_2)$. Then $\wtd I_1,\wtd J_1,\wtd I_2,\wtd J_2$ are contained in $\wtd K$. By Rem. \ref{lb76} and a possible change of subscripts, we may assume that $\wtd I_2$ is clockwise to $\wtd I_1$ and $\wtd J_2$ is anticlockwise to $\wtd J_1$. Therefore, the following arg-valued intervals are in clockwise order:
\begin{align*}
\sqrt{\wtd I_1}\qquad\sqrt{\wtd I_2}\qquad -\sqrt{\wtd J_2}\qquad -\sqrt{\wtd J_1}
\end{align*}
Moreover, by enlarging $\wtd I_1$ and $\wtd J_1$ (and hence enlarging $O_1$), we assume in addition that by viewing $\Sbb^1=\Rbb/2\pi\Zbb$, we have $\inf\wtd I_1=\sup\wtd I_2$ and $\inf\wtd J_2=\sup\wtd J_1$. 

Applying the following Lem. \ref{lb73} to the case where $E_3\setminus \ovl{E_2}$ is the disjoint union $\sqrt{I_1}\sqcup -\sqrt{J_1}$, and $E_2\setminus \ovl{E_1}$ is the disjoint union $\sqrt{I_2}\sqcup -\sqrt{J_2}$, we see that $\fk A(\sqrt{I_1},-\sqrt{J_1})$ commutes with $\fk A(\sqrt{I_2},-\sqrt{J_2})$. Applying $\pi^L_{\wtd K}(\cdots)|_{\MH_\sqz\boxtimes\MH_{\ovl\sqz}}$ (and noting the Haag duality for $\MA$), we see that $\Bcl(O_1)$ commutes with $\Bcl(O_2)$.
\end{proof}

\begin{lm}\label{lb73}
Let $E_1,E_2,E_3\in\MJ$ such that $E_1\Subset E_2\Subset E_3$. Then $\MA(E_3)\cap\MA(E_2)'$ commutes with $\MA(E_2)\cap\MA(E_1)'$
\end{lm}

\begin{proof}
This is obvious.
\end{proof}

\subsection{PCT symmetry and Haag duality for $O_\leftarrow$ and $O_\rightarrow$}\label{lb161}

\begin{df}\label{lb101}
Recall from Rem. \ref{lb81} that $\Theta_{\sqz,\ovl\sqz}$ is an involutive antiunitary anti-automorphism of the $\Maa$-module $\MH_\sqz\boxtimes\MH_{\ovl\sqz}$. Define
\begin{align*}
\pmb{\Jcl}=\Ucl^-(\varrho_c(-\pi))\cdot\Theta_{\sqz,\ovl\sqz}\cdot \Ucl^-(\varrho_c(\pi))
\end{align*}
called the \textbf{PCT operator} of the closed CFT $(\Bcl,\MHcl)$. Due to the following lemma, we have
\begin{align*}
\Jcl=\Ucl^-(\varrho_c(-2\pi))\cdot\Theta_{\sqz,\ovl\sqz}=\Theta_{\sqz,\ovl\sqz}\cdot \Ucl^-(\varrho_c(2\pi))
\end{align*}
\end{df}

\begin{lm}\label{lb88}
For each $\pm\in\{+,-\}$ and $s\in\Rbb$, we have
\begin{align*}
\Ucl^\pm(\varrho_c(s))\Theta_{\sqz,\ovl\sqz}=\Theta_{\sqz,\ovl\sqz}\Ucl^\pm(\varrho_c(-s))
\end{align*}
\end{lm}

\begin{proof}
Since $\Theta_{\sqz,\ovl\sqz}$ is an anti-morphism of $\Maa$-modules, and since the modular conjugation for
\begin{align*}
\Maa(\Sbb^1_+)=\MA(\Sbb^1_+)\otimes\MA(\Sbb^1_+)
\end{align*}
and the cyclic separating vector $\Omega\otimes\Omega$ is $\fk J_\MA\otimes\fk J_\MA$, we see from Def. \ref{lb82} that $\Theta_{\sqz,\ovl\sqz}$ is an anti-automorphism of the $\MA$-module $(\MHcl,\pi_\cl^\pm)$ (where $\pi_\cl^\pm$ is defined in Def. \ref{lb83}). The desired relation then follows from Cor. \ref{lb84}.
\end{proof}

\begin{figure}[h]
	\centering
\begin{gather*}
\includegraphics[height=4cm]{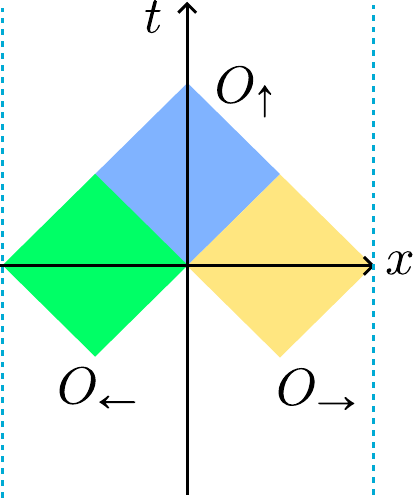}
   \end{gather*}
\caption{.~~The double cones $O_\uparrow,O_\rightarrow,O_\leftarrow$ in $\Rpp$}\label{fig2}
\end{figure}

\begin{thm}\label{lb89}
Let $\pmb{O_\uparrow}$, $\pmb{O_\rightarrow}$, and $\pmb{O_\leftarrow}$ be the equilateral double cones corresponding to
\begin{align*}
(\Std_+,\Std_+)\qquad (\Std_+,\Std_-)\qquad (\Std_-,\Std_+)
\end{align*}
respectively; see Fig. \ref{fig2}. Then
\begin{gather}\label{eq39}
\Ad_{\Jcl}\big(\Bcl(O_\rightarrow))=\Bcl(O_\rightarrow)'=\Bcl(O_\leftarrow)
\end{gather}
\end{thm}

Note that $O_\uparrow$ is standard, and hence compactly contained in $\Rpp$. Although $O_\rightarrow,O_\leftarrow$ are contained in $\Rpp$, they are not compactly contained in $\Rpp$.

\begin{proof}
Step 1. We write $\Theta_{\sqz,\ovl{\sqz}}=\Theta$ and $\Ucl^\pm(\varrho_c(s))=V^\pm(s)$ for simplicity. Clearly
\begin{align*}
(1,\varrho(\pi))O_\rightarrow=O_\uparrow=(\varrho(\pi),1)O_\leftarrow
\end{align*}
Therefore, by Lem. \ref{lb85} and \ref{lb86}, we have
\begin{align}\label{eq34}
\Bcl(O_\rightarrow)=\Ad_{V^-(-\pi)}\big(\Bcl(O_\uparrow)\big)\qquad \Bcl(O_\leftarrow)=\Ad_{V^+(-\pi)}\big(\Bcl(O_\uparrow)\big)
\end{align}
By Def. \ref{lb64}, we have
\begin{align*}
\Bcl(O_\uparrow)=\pi^L_{\Std_+}\big(\End_{\Maa(\Std_-)}(\MH_\sqz)\big)\big|_{\MHcl}
\end{align*}
It follows from Cor. \ref{lb87} that
\begin{align}\label{eq35}
\Ad_\Theta\big(\Bcl(O_\uparrow)\big)=\Bcl(O_\uparrow)'
\end{align}

Combining \eqref{eq35} with \eqref{eq34}, we obtain
\begin{subequations}\label{eq40}
\begin{gather}
\Ad_{V^-(-\pi)\Theta V^-(\pi)}\big(\Bcl(O_\rightarrow)\big)=\Bcl(O_\rightarrow)'\\
\Ad_{V^+(-\pi)\Theta V^+(\pi)}\big(\Bcl(O_\leftarrow)\big)=\Bcl(O_\leftarrow)'
\end{gather}
\end{subequations}
This proves the first half of \eqref{eq39}. \\[-1ex]

Step 2. By \eqref{eq40} and the locality of $\Bcl$ (cf. Thm. \ref{lb78}), we obtain that
\begin{subequations}\label{eq36}
\begin{gather}
\Bcl(O_\leftarrow)\subset \Ad_{V^-(-\pi)\Theta V^-(\pi)}\big(\Bcl(O_\rightarrow)\big)\label{eq36a}\\
\Bcl(O_\rightarrow)\subset \Ad_{V^+(-\pi)\Theta V^+(\pi)}\big(\Bcl(O_\leftarrow)\big)\label{eq36b}
\end{gather}
\end{subequations}
and that the inclusion in \eqref{eq36a} is an equality iff $\Bcl(O_\leftarrow)=\Bcl(O_\rightarrow)'$.

By Lem. \ref{lb88} and $\Theta^2=1$, and noting Rem. \ref{lb104} for the commutativity of $V^+(s)$ and $V^-(t)$, we have
\begin{gather*}
V^-(-\pi)\Theta V^-(\pi)V^+(-\pi)\Theta V^+(\pi)=V^+(2\pi)V^-(-2\pi)\\
V^+(-\pi)\Theta V^+(\pi)V^-(-\pi)\Theta V^-(\pi)=V^+(-2\pi)V^-(2\pi)
\end{gather*}
Combining this with \eqref{eq36}, we get
\begin{subequations}\label{eq37}
\begin{gather}
\Bcl(O_\leftarrow)\subset \Ad_{V^+(2\pi)V^-(-2\pi)}\Bcl(O_\leftarrow)\label{eq37a}\\
\Bcl(O_\rightarrow)\subset \Ad_{V^+(-2\pi)V^-(2\pi)}\Bcl(O_\rightarrow)
\end{gather}
\end{subequations}
and that if the inclusion in \eqref{eq37a} is an equality, then $\Bcl(O_\leftarrow)=\Bcl(O_\rightarrow)'$.

By \eqref{eq34}, we have
\begin{align*}
\Bcl(O_\rightarrow)=\Ad_{V^+(\pi)V^-(-\pi)}\Bcl(O_\leftarrow)
\end{align*}
Therefore, \eqref{eq37} is equivalent to
\begin{subequations}\label{eq38}
\begin{gather}
\Bcl(O_\leftarrow)\subset\Ad_{V^+(\pi)V^-(-\pi)}\Bcl(O_\rightarrow)\label{eq38a}\\
\Bcl(O_\rightarrow)\subset\Ad_{V^+(-\pi)V^-(\pi)}\Bcl(O_\leftarrow)
\end{gather}
\end{subequations}
Combining these two inclusion relations together, we get
\begin{align*}
\Bcl(O_\leftarrow)\subset \Bcl(O_\leftarrow)
\end{align*}
where the inclusion is an equality iff the two inclusions in \eqref{eq38} are both equalities. But the above inclusion is clearly an equality. Therefore, the inclusion in \eqref{eq38a} (hence the one in \eqref{eq37a}) is an equality, and thus $\Bcl(O_\leftarrow)=\Bcl(O_\rightarrow)'$. 
\end{proof}

\subsection{Conformal covariance}\label{lb100}

Recall Thm. \ref{lb89} for the meanings of $O_\rightarrow$ and $O_\leftarrow$.

\begin{lm}\label{lb92}
We have
\begin{align*}
\Ucl(\varrho_c(\pi),\varrho_c(-\pi))\cdot\Bcl(O_\rightarrow)\cdot\Ucl(\varrho_c(\pi),\varrho_c(-\pi))^*=\Bcl(O_\leftarrow)
\end{align*}
\end{lm}

\begin{proof}
At the end of the proof of Thm. \ref{lb89}, we have seen that the inclusion in \eqref{eq38a} is indeed an equality.
\end{proof}

\subsubsection{A computation of relative commutants}

In this subsection, we prove a result on relative commutants for $\Bcl$ that will be used in the proof of the conformal covariance of $\Bcl$. Recall that if $\mc M_1,\mc N_1$ are von Neumann algebras on a Hilbert space $\MH_1$, and if $\mc M_2,\mc N_2$ are von Neumann algebras on $\MH_2$, then $(\mc M_1\otimes\mc M_2)\cap(\mc N_1\otimes\mc N_2)=(\mc M_1\cap\mc N_1)\otimes(\mc M_2\cap\mc N_2)$, cf. \cite[III.4.5.9]{Bla06}.

\begin{lm}\label{lb90}
Let $E_1,E_2,E_3\in\MJ$ such that $E_1\Subset E_2\Subset E_3$. Then the relative commutant of $\MA(E_2)\cap\MA(E_1)'$ in $\MA(E_3)\cap\MA(E_1)'$ is $\MA(E_3)\cap\MA(E_2)'$.
\end{lm}

\begin{proof}
By the split property (cf. Rem. \ref{lb1}), the inclusion of $\MA(E_2)\subset\MA(E_3)$ contains a type I intermediate factor, and hence is unitarily equivalent to
\begin{align*}
\Cbb\otimes\scr \MA(E_2)\subset\scr A_3\otimes\fk L(\MH_0)\qquad\curvearrowright\MK_3\otimes\MH_0
\end{align*}
where $\scr A_3$ is a von Neumann algebra acting on a Hilbert space $\MK_3$. Similarly, the inclusion $\MA(E_1)\subset\MA(E_2)$ is unitarily equivalent to
\begin{align*}
\Cbb\otimes\scr \MA(E_1)\subset\scr A_2\otimes\fk L(\MH_0)\qquad\curvearrowright\MK_2\otimes\MH_0
\end{align*} 
where $\scr A_2$ is a von Neumann algebra acting on a Hilbert space $\MK_2$. Set $\MK_1=\MH_0$ and $\scr A_1=\MA(E_1)$. Then the chain $\MA(E_1)\subset\MA(E_2)\subset\MA(E_3)$ is unitarily equivalent to
\begin{align*}
\Cbb\otimes\Cbb\otimes\SA_1\subset\Cbb\otimes\SA_2\otimes\fk L(\MK_1)\subset\SA_3\otimes\fk L(\MK_2)\otimes\fk L(\MK_1) \qquad\curvearrowright\MK_3\otimes\MK_2\otimes\MK_1
\end{align*}

It follows that
\begin{gather*}
\MA(E_2)\cap\MA(E_1)'=\Cbb\otimes\SA_2\otimes\SA_1'\\
\MA(E_3)\cap\MA(E_1)'=\SA_3\otimes\fk L(\MK_2)\otimes\SA_1'\\
\MA(E_3)\cap\MA(E_2)'=\SA_3\otimes\SA_2'\otimes\Cbb
\end{gather*}
Taking the commutant of the first equality, we get
\begin{align*}
\big(\MA(E_2)\cap\MA(E_1)'\big)'=\fk L(\MK_3)\otimes\SA_2'\otimes\SA_1
\end{align*}
Using the fact that $\SA_1=\MA(E_1)$ is a (type III) factor and hence $\SA_1\cap\SA_1'=\Cbb$, we obtain
\begin{align*}
\big(\MA(E_3)\cap\MA(E_1)'\big)\cap\big(\MA(E_2)\cap\MA(E_1)'\big)'=\SA_3\otimes\SA_2'\otimes \Cbb
\end{align*} 
which equals $\MA(E_3)\cap\MA(E_2)'$.
\end{proof}

\begin{pp}\label{lb93}
Let $O,O_1\subset\Rpp$ be double cones with $O_1\subset O$. Assume that $O$ and $O_1$ have the same right corner,
but have different top, bottom, and left corners. Note that $O_2:=O\cap O_1'$ is also a double cone. Then
\begin{align*}
\Bcl(O_1)=\Bcl(O)\cap\Bcl(O_2)'\qquad \Bcl(O_2)=\Bcl(O)\cap\Bcl(O_1)'
\end{align*} 
\end{pp}

\begin{proof}
By Lem. \ref{lb86}, we may assume in addition that $O\Subset\Rpp$, and hence $O$ is contained in a standard double cone corresponding to some $(\wtd K,\wtd K)$. Let $O,O_i$ correspond to $(\wtd I,\wtd J)$ and $(\wtd I_i,\wtd J_i)$. Then $\wtd I_2$ is clockwise to $\wtd I_1$, $\wtd J_1$ is clockwise to $\wtd J_2$, $I_1\sqcup I_2$ is dense in $I$, and $J_1\sqcup J_2$ is dense in $J$. Applying Lem. \ref{lb90} to the case where $E_2\setminus \ovl E_1=\sqrt{I_2}\sqcup -\sqrt{J_2}$ and $E_3\setminus \ovl E_2=\sqrt{I_1}\sqcup -\sqrt{J_1}$, we obtain
\begin{align*}
\fk A(\sqrt{I_1},-\sqrt{J_1})=\fk A(\sqrt I,-\sqrt J)\cap\fk A(\sqrt{I_2},-\sqrt{J_2})'
\end{align*}
Since the representation $\pi^L_{\wtd K}(\cdots)|_{\Hcl}$ on $\fk A(\sqrt K,-\sqrt K)=\End_{\Maa(\wtd K')}(\MH_\sqz)$ is faithful (cf. Rem. \ref{lb91}), applying this representation to the above identity yields the first relation. The second relation can be proved in the same manner (by choosing different $E_1,E_2,E_3$).
\end{proof}

\subsubsection{Proof of the conformal covariance}

\begin{lm}\label{lb94}
Let $O\subset\Rpp$ be an equilateral double cone such that in the $xt$-coordinates, the right corner of $O$ is $(\pi,0)$, and the diameter $r$ of $O$ satisfies $r<\pi$. Then
\begin{align*}
\Ucl(\varrho_c(r),\varrho_c(-r))\cdot\Bcl(O)\cdot\Ucl(\varrho_c(r),\varrho_c(-r))^*=\Bcl((\varrho(r),\varrho(-r))O)
\end{align*}
\end{lm}

\begin{proof}
We write $\Ucl(\varrho_c(s),\varrho_c(-s))=V(s)$ and $\delta(s)=(\varrho(s),\varrho(-s))$ for simplicity. 
Since $\delta(r-\pi)O$ is contained in $O_\rightarrow$, Lem. \ref{lb86} implies
\begin{align*}
\Ad_{V(r-\pi)}\big(\Bcl(O)\big)\subset\Bcl(O_\rightarrow)
\end{align*}
and hence
\begin{align*}
\Ad_{V(r)}\big(\Bcl(O)\big)= \Ad_{V(\pi)}\Ad_{V(r-\pi)}\big(\Bcl(O)\big)\subset\Ad_{V(\pi)}\big(\Bcl(O_\rightarrow)\big)=\Bcl(O_\leftarrow)
\end{align*}
where the last equality is due to Lem. \ref{lb92}. 

Let
\begin{align*}
O_2:=\delta(r)O\qquad O_1=O_\leftarrow\cap O_2'
\end{align*}
Since $O$ and $\delta(-r)O_1$ are spacelike separated (because they are contained in $O_\rightarrow$ and $O_\leftarrow$, respectively), by the locality (Thm. \ref{lb78}), $\Bcl(O)$ commutes with $\Bcl(\delta(-r)O_1)$, and hence commutes with $\Ad_{V(-r)}\big(\Bcl(O_1)\big)$ by Lem. \ref{lb86}. Thus
\begin{align*}
\Ad_{V(r)}\big(\Bcl(O)\big)\subset\Bcl(O_1)'
\end{align*}
By Prop. \ref{lb93} and $O_2=O_\leftarrow\cap O_1'$, we obtain
\begin{align*}
\Ad_{V(r)}\big(\Bcl(O)\big)\subset\Bcl(O_2)
\end{align*}
A similar argument shows $\Ad_{V(-r)}\big(\Bcl(O_2)\big)\subset\Bcl(O)$. Thus $\Ad_{V(r)}\big(\Bcl(O)\big)=\Bcl(O_2)$.
\end{proof}

Recall \eqref{eq29} for the notation $\scr L_x$.

\begin{lm}\label{lb95}
Let $O\subset\Rpp$ be an equilateral double cone whose diameter in the $xt$-coordinates is $r<\pi$. Choose $s\in\Rbb$ such that $(\rho(s),\rho(-s))O$ does not intersect $\SL_{\pm\pi}$ (and hence can be viewed as contained in $\Rpp$). Then
\begin{align}\label{eq43}
\Ucl(\varrho_c(s),\varrho_c(-s))\cdot\Bcl(O)\cdot\Ucl(\varrho_c(s),\varrho_c(-s))^*=\Bcl((\varrho(s),\varrho(-s))O)
\end{align}
\end{lm}

\begin{proof}
The case $s<0$ follows easily from the case $s\geq0$. So we assume $s\geq0$. By Lem. \ref{lb85}, it suffices to assume that the left and right corners of $O$ lie on the $x$-axis. Since the translation of $O$ appearing in this lemma can be decomposed into translations of the types considered in Lem. \ref{lb86} and Lem. \ref{lb94}, the result follows from those two lemmas.
\end{proof}

\begin{lm}\label{lb96}
For each $O\in\MO$ and $s\in\Rbb$, the relation \eqref{eq43} holds.
\end{lm}

\begin{proof}
Write $\delta(s)=(\varrho(s),\varrho(-s))$. Choose finite sets $E,F\subset(-\pi,\pi]$ containing $\pi$ such that $\delta(s)\SL_E=\SL_F$, where $\SL_E=\bigcup_{x\in E}\SL_x$ and $\SL_F=\bigcup_{x\in F}\SL_x$. Let $(O_\alpha)$ be a collection of equilateral double cones whose union is $O\setminus\SL_E$, and whose diameters in the $xt$-coordinates are $<\pi$. Then each $O_\alpha$ and $\delta(s)O_\alpha$ can be viewed as subsets of $\Rpp$. By Lem. \ref{lb95}, \eqref{eq43} holds when $O$ is replaced by $O_\alpha$. 

By Cor. \ref{lb75}, $\Bcl(O)$ is generated by all $\Bcl(O_\alpha)$. Since $(\delta(s)O)\setminus\SL_F=\bigcup_\alpha \delta(s)O_\alpha$, Cor. \ref{lb75} also implies that $\Bcl(\delta(s)O)$ is generated by all $\Bcl(\delta(s)O_\alpha)$. This finishes the proof.
\end{proof}

\begin{lm}\label{lb97}
Let $O\in\MO$ be corresponding to some $(\wtd I,\wtd J)$. Let $E,F\in\MJ$. Suppose that there exist $\wtd K,\wtd L\in\MJ$ satisfying
\begin{align*}
\wtd I\subset\wtd K\qquad E\subset K\qquad \wtd J\subset \wtd L\qquad F\subset L
\end{align*}
Then for each $g^+\in\Gc(E)$ and $g^-\in\Gc(F)$ we have
\begin{align*}
\Ucl(g^+,g^-)\Bcl(O)\Ucl(g^+,g^-)^*=\Bcl((g^+,g^-)O)
\end{align*}
\end{lm}

\begin{proof}
Note that $g^+\in\Gc(K)$ and $g^-\in\Gc(L)$. By enlarging either $\wtd K$ or $\wtd L$, we assume they have the same length. Choose $s\in\Rbb$ such that $\varrho(s)\wtd L=\wtd K$. Let
\begin{align*}
h^-=\varrho_c(s)g^-\varrho_c(-s)\qquad D=(1,\varrho(s))O
\end{align*}
Then $h^-\in\Gc(K)$, and $D$ is contained in the standard double cone corresponding to $(\wtd K,\wtd K)$. Therefore, by Lem. \ref{lb68},
\begin{align*}
\Ucl(g^+,h^-)\Bcl(D)\Ucl(g^+,h^-)^*=\Bcl((g^+,h^-)D)
\end{align*}
Thus, writing $\Ucl^-(\varrho_c(s))$ as $V$, and applying $\Ad_{V^*}$ to the above identity, we have
\begin{align*}
\Ucl(g^+,g^-)V^*\Bcl(D)V\Ucl(g^+,g^-)^*=V^*\Bcl((g^+,h^-)D)V
\end{align*}
By Lem. \ref{lb85} and \ref{lb96}, we have $V^*\Bcl(D)V=\Bcl(O)$ and
\begin{align*}
&V^*\Bcl((g^+,h^-)D)V=\Bcl((g^+,\varrho_c(-s)h^-)D)\\
=&\Bcl((g^+,g^-\varrho_c(-s))D)=\Bcl((g^+,g^-)O)
\end{align*}
This finishes the proof.
\end{proof}

\begin{thm}\label{lb98}
For each $O\in\MO$ and $g^+,g^-\in\Gc$, we have
\begin{align*}
\Ucl(g^+,g^-)\Bcl(O)U_\cl(g^+,g^-)^*=\Bcl((g^+,g^-)O)
\end{align*}
\end{thm}

\begin{proof}
By Lem. \ref{lb14}, $\Gc$ is algebraically generated by $\Gc(E)$ for all $E\in\MJ$ whose length is $<\pi$. Therefore, it suffices to prove the desired relation for $g^+\in\Gc(E),g^-\in\Gc(F)$ where $E,F\in\MJ$ have lengths $<\pi$. Moreover, by Cor. \ref{lb75}, it suffices to consider the case where the arg-valued intervals $(\wtd I,\wtd J)$ corresponding to $O$ have lengths $<\pi$. Then one can find $\wtd K,\wtd L\in\Jtd$ satisfying the requirement of Lem. \ref{lb97}. The conclusion now follows from that lemma.
\end{proof}

\subsection{Some consequences of the conformal covariance}

Recall Def. \ref{lb83} for the meanings of $\pi^+$ and $\pi^-$.

\subsubsection{Some extension properties}

\begin{co}\label{lb135}
If $O\in\MO$ corresponds to $(\wtd I,\wtd J)$, then
\begin{align*}
\pi^+_{\sqz\boxtimes\ovl\sqz,I}(\MA(I))\subset \Bcl(O)\qquad \pi^-_{\sqz\boxtimes\ovl\sqz,J}(\MA(J))\subset \Bcl(O)
\end{align*}
\end{co}

\begin{proof}
We first consider the case $\wtd I=\wtd J$. Let $x\in\MA(I)$. By \eqref{eq28}, we have
\begin{align*}
\pi^+_{\sqz\boxtimes\ovl\sqz,I}(x)=\pi_{\sqz\boxtimes\ovl\sqz,I}(x\otimes 1)=\pi^L_{\wtd I}(\pi_{\sqz,\wtd I}(x\otimes 1))|_{\MH_\sqz\boxtimes\MH_{\ovl\sqz}}
\end{align*}
where $\pi_{\sqz,\wtd I}(x\otimes 1)$ belongs to $\End_{\Maa(\wtd I')}(\MH_\sqz)$. Therefore, $\pi^+_{\sqz\boxtimes\ovl\sqz,I}(x)$ belongs to $\Bcl(O)=\pi^L_{\wtd I}(\End_{\Maa(\wtd I')}(\MH_\sqz))|_{\MH_\sqz\boxtimes\MH_{\ovl\sqz}}$ (cf. Def. \ref{lb64}). This proves the first inclusion. Similarly, the second inclusion follows from
\begin{align*}
\pi^-_{\sqz\boxtimes\ovl\sqz,I}(x)=\pi^R_{\wtd I}(\pi_{\sqz,\wtd I}(1\otimes x))|_{\MH_\sqz\boxtimes\MH_{\ovl\sqz}}
\end{align*}

Now we consider the general case. Choose $\gamma\in\Gc$ such that $\gamma\wtd J=\wtd I$. Let $D\in\Opp$ correspond to $(\wtd I,\wtd I)$. By Def. \ref{lb134}, we have $\Ucl(1,\gamma)=\Ucl^-(\gamma)$. Thus, Thm. \ref{lb98} implies
\begin{align*}
\Ad_{\Ucl^-(\gamma)}\big(\Bcl(O)\big)=\Bcl(D)
\end{align*}
By Rem. \ref{lb104}, $\Ad_{\Ucl^-(\gamma)}$ sends $\pi^+_{\sqz\boxtimes\ovl\sqz,I}(\MA(I))$ to itself. By Cor. \ref{lb30}, $\Ad_{\Ucl^-(\gamma)}$ sends $\pi^-_{\sqz\boxtimes\ovl\sqz,J}(\MA(J))$ to $\pi^-_{\sqz\boxtimes\ovl\sqz,I}(\MA(I))$. Thus, the general case reduces to the special case $\wtd I=\wtd J$.
\end{proof}

\begin{co}\label{lb99}
If $O\in\MO$ corresponds to $(\wtd I,\wtd J)$ and $g^+\in\Gc(I),g^-\in\Gc(J)$, then 
\begin{align*}
\Ucl(g^+,g^-)\in \pi^+_{\sqz\boxtimes\ovl\sqz,I}(\MA(I))\vee \pi^-_{\sqz\boxtimes\ovl\sqz,J}(\MA(J))
\end{align*}
and hence $\Ucl(g^+,g^-)\in\Bcl(O)$ by Cor. \ref{lb135}. 
\end{co}

It follows that for each $m,n\in\Zbb$, we have
\begin{align}
\Ucl(g^+,g^-)\in\Bcl((\varrho(2m\pi),\varrho(2n\pi))O)
\end{align}
because $(\varrho(2m\pi),\varrho(2n\pi))O$ also corresponds to arg-valued intervals whose underlying intervals are $I$ and $J$.

\begin{proof}
We have
\begin{align*}
\Ucl^+(g^+)=U^+_{\sqz\boxtimes\ovl\sqz,I}(g^+)\xlongequal{\text{Thm.\ref{lb12}}}\pi^+_{\sqz\boxtimes\ovl\sqz,I}(U_0(g^+))\in\pi^+_{\sqz\boxtimes\ovl\sqz,I}(\MA(I))
\end{align*}
and, similarly, $\Ucl^-(g^-)$ belongs to $\pi^-_{\sqz\boxtimes\ovl\sqz,J}(\MA(J))$. Thus $U_\cl(g^+,g^-)=\Ucl^+(g^+)\Ucl^-(g^-)$ belongs to $\pi^+_{\sqz\boxtimes\ovl\sqz,I}(\MA(I))\vee\pi^-_{\sqz\boxtimes\ovl\sqz,J}(\MA(J))$.
\end{proof}

\subsubsection{An enhanced version of the additivity property}

The result of this subsection is of independent interest, although it will not be used elsewhere in this paper.

\begin{df}
Let $O\in\MO$. A subset $\Lambda\subset O$ is called an \textbf{arc between the top and bottom corners} of $O$ if, after choosing $s\in\Rbb$ such that $D:=(\varrho(s),\varrho(-s))O$ belongs to $\Opp$, the set $(\varrho(s),\varrho(-s))\Lambda$ is an arc between the top and bottom corners of $D$ in the sense of Def. \ref{lb70}.
\end{df}

\begin{thm}\label{lb146}
Let $O\in\MO$, and let $\Lambda$ be an arc between the top and bottom corners of $O$. Let $(O_\alpha)_{\alpha\in\scr A}$ be a collection in $\MO$ such that $O_\alpha\subset O$ for each $\alpha$, and $\Lambda\subset\bigcup_\alpha O_\alpha$. Then
\begin{align}
\Bcl(O)=\bigvee_\alpha\Bcl(O_\alpha)
\end{align}
\end{thm}

\begin{proof}
This is an easy consequence of Prop. \ref{lb71} and Thm. \ref{lb98}.
\end{proof}

\subsection{PCT symmetry}\label{lb103}

Recall the PCT operator $\Jcl$ defined in Def. \ref{lb101}. Recall the reflection $(\rk,\rk)$ of $\Rcl$ described in \eqref{eq44}. Recall Def. \ref{lb51} for the definition of $\rk g\rk$ where $g\in\Gc$.

\begin{thm}\label{lb102}
The anti-unitary operator $\Jcl$ is involutive (i.e. $\Jcl^2=1$), and satisfies
\begin{subequations}\label{eq45}
\begin{gather}
\Jcl\Ucl(g^+,g^-)\Jcl=U_\cl(\rk g^+\rk,\rk g^-\rk)\label{eq45a}\\
\Jcl\Bcl(O)\Jcl=\Bcl((\rk,\rk)O)\label{eq45b}
\end{gather}
\end{subequations}
for each $g^+,g^-\in\Gc$ and $O\in\MO$.
\end{thm}

\begin{proof}
We write $\Theta_{\sqz,\ovl{\sqz}}=\Theta$ and $\Ucl^\pm(\varrho_c(s))=V^\pm(s)$ for simplicity. Then $\Jcl=V^-(-2\pi)\Theta=\Theta V^-(2\pi)$. By Lem. \ref{lb88}, we have
\begin{align*}
\Jcl^2=\Theta V^-(2\pi)\Theta V^-(2\pi)=\Theta^2 V^-(-2\pi)V^-(2\pi)=1
\end{align*}

From the proof of Lem. \ref{lb88}, we see that $\Theta$ is anti-automorphism of the $\MA$-module $(\Hcl,\pi^\pm_\cl)$. Therefore, by Prop. \ref{lb42}, we have
\begin{align*}
\Theta \Ucl^+(g^+)\Theta=\Ucl^+(\rk g^+\rk)\qquad \Theta \Ucl^-(g^-)\Theta=\Ucl^-(\rk g^-\rk)
\end{align*}
Multiplying these two identities yields
\begin{align*}
\Theta\Ucl(g^+,g^-)\Theta=U_\cl(\rk g^+\rk,\rk g^-\rk)
\end{align*}
For each $\pm\in\{+,-\}$, by Cor. \ref{lb30}, $V^\pm(2\pi)$ is an automorphism of the $\MA$-module $(\Hcl,\pi^\pm_\cl)$. So it commutes with $\Ucl^\pm(g^\pm)$ by Prop. \ref{lb42}. By Rem. \ref{lb104}, $V^\pm(2\pi)$ commutes with $\Ucl^\mp(g^\mp)$. Thus \eqref{eq45a} is proved.

Choose $g^+,g^-\in\Gc$ such that $O_\rightarrow=(g^+,g^-)O$. By Thm. \ref{lb98}, we have
\begin{align*}
\Bcl(O_\rightarrow)=\Ad_{\Ucl(g^+,g^-)}\big(\Bcl(O)\big)
\end{align*}
and hence
\begin{align*}
\Ad_{\Jcl}\big(\Bcl(O_\rightarrow)\big)=\Ad_{\Jcl}\Ad_{\Ucl(g^+,g^-)}\big(\Bcl(O)\big)\xlongequal{\eqref{eq45a}}\Ad_{\Ucl(\rk g^+\rk,\rk g^-\rk)}\Ad_{\Jcl}\big(\Bcl(O)\big)  \tag{$\star$}\label{eq56}
\end{align*}
By Thm. \ref{lb89}, the left hand side above equals $\Bcl(O_\leftarrow)$. Since $O_\leftarrow=(\rk,\rk)O_\rightarrow$, we have
\begin{align*}
O_\leftarrow=(\rk,\rk)(g^+,g^-)O=(\rk g^+\rk,\rk g^-\rk)(\rk,\rk)O
\end{align*}
and hence, by Thm. \ref{lb98},
\begin{align*}
\Bcl(O_\leftarrow)=\Ad_{\Ucl(\rk g^+\rk,\rk g^-\rk)}\Bcl((\rk,\rk)O)  \tag{$\star\star$}\label{eq57}
\end{align*}
So the right hand sides of \eqref{eq56} and \eqref{eq57} are equal, finishing the proof of \eqref{eq45b}.
\end{proof}

\subsection{Haag duality for multi-double-cones}\label{lb145}

\begin{df}\label{lb105}
For $O_1,\dots,O_n\in\MO$ whose closures are mutually spacelike separated, let
\begin{align*}
\pmb{\Bcl(O_1\cup\cdots\cup O_n)}:=\big(\Bcl(O_1)\cup\cdots\cup\Bcl(O_n)\big)''
\end{align*}
Note that $(O_1\cup\cdots\cup O_n)'$ is also the disjoint union of some $D_1,\dots,D_n\in\MO$ whose closures are mutually spacelike separated.
\end{df}

\begin{figure}[h]
	\centering
\begin{gather*}
\includegraphics[height=3.5cm]{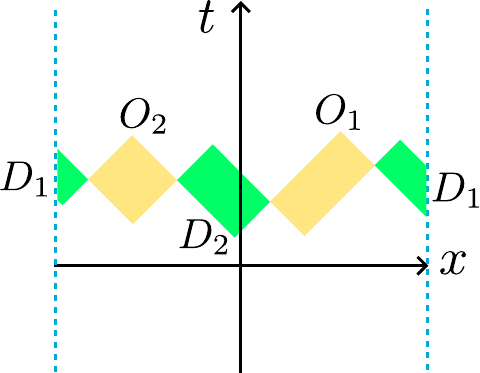}
   \end{gather*}
\caption{.~~The causal complement $D_1\cup D_2$ of $O_1\cup O_2$, projected onto $\Rpp$}\label{fig3}
\end{figure}

\begin{thm}[\textbf{Bulk-bulk Haag duality}]\label{lb115}
Let $n\geq 1$, and let $O_1,\dots,O_n$ and $D_1,\dots,D_n$ be as in Def. \ref{lb105}. Then
\begin{align*}
\big(\Bcl(O_1\cup\cdots\cup O_n)\big)'=\Bcl(D_1\cup\cdots\cup D_n)
\end{align*}
\end{thm}

\begin{proof}
When $n=1$, this follows from Thm. \ref{lb89} and \ref{lb98}. Now assume $n\geq2$, and let $\Delta=O_n'$. Then
\begin{align*}
O_1\cup\cdots\cup O_{n-1}\subset \Delta\qquad \Delta\cap O_1'\cap\cdots\cap O_{n-1}'=D_1\sqcup\cdots\sqcup D_n
\end{align*}
The proved special case implies $\Bcl(O_n)'=\Bcl(\Delta)$. Thus, it suffices to show
\begin{align}\label{eq46}
\Bcl(\Delta)\cap\Bcl(O_1\cup\cdots\cup O_{n-1})'=\Bcl(D_1\cup\cdots\cup D_n)
\end{align}

By Thm. \ref{lb98}, it suffices to assume that $\Delta$ is standard, and hence corresponding to $(\wtd\Gamma,\wtd\Gamma)$ for some $\wtd\Gamma\in\Jtd$. By changing the subscripts, we assume that the following double subcones of $\Delta$ are listed from right to left:
\begin{align*}
D_1,O_1,D_2,O_2,\dots,O_{n-1},D_n
\end{align*}
Let $O_i$ correspond to $(\wtd I_i,\wtd J_i)$ and $D_j$ to $(\wtd K_j,\wtd L_j)$, where all arg-valued intervals are contained in $\wtd\Gamma$. Then, the following arg-valued intervals are listed in the clockwise order:
\begin{gather*}
\sqrt{\wtd K_1},\sqrt{\wtd I_1},\sqrt{\wtd K_2},\sqrt{\wtd I_2},\dots,\sqrt{\wtd I_{n-1}},\sqrt{\wtd K_n}\\
-\sqrt{\wtd L_n},-\sqrt{\wtd J_{n-1}},-\sqrt{\wtd L_{n-1}},-\sqrt{\wtd J_{n-2}},\dots,-\sqrt{\wtd J_1},-\sqrt{\wtd L_1}
\end{gather*}
Applying the following Lem. \ref{lb106} to the case
\begin{gather*}
F_j\setminus\ovl{E_j}=\sqrt{K_j}\sqcup -\sqrt{L_j}\qquad E_i\setminus\ovl{F_{i+1}}=\sqrt{I_i}\sqcup -\sqrt{J_i} 
\end{gather*}
(where $i\leq n-1$ and $j\leq n$), and using the faithfulness of the normal representation $\pi^L_\Gamma(\cdots)|_{\Hcl}$ as in the proof of Prop. \ref{lb93}, we obtain \eqref{eq46}.
\end{proof}

\begin{lm}\label{lb106}
Let $F_1,E_1,\dots,F_n,E_n\in\MJ$ such that
\begin{align*}
F_1\Supset E_1\Supset F_2\Supset E_2\Supset\cdots\Supset F_n\Supset E_n
\end{align*}
Then the relative commutant of
\begin{align*}
\bigvee_{i=1}^{n-1}\big(\MA(E_i)\cap\MA(F_{i+1})'\big)
\end{align*}
in
\begin{align*}
\MA(F_1)\cap\MA(E_n)'
\end{align*}
is equal to
\begin{align*}
\bigvee_{j=1}^n\big(\MA(F_j)\cap\MA(E_j)'\big)
\end{align*}
\end{lm}

\begin{proof}
Similar to the proof of Lem. \ref{lb90}, the chain $\MA(F_1)\supset\MA(E_1)\supset\cdots\supset\MA(E_n)$ is unitarily equivalent to
\begin{align*}
&\SB_1\otimes\fk L(\MH_1)\otimes\fk L(\MK_2)\otimes\fk L(\MH_2)\otimes\cdots\otimes\fk L(\MK_n)\otimes\fk L(\MH_n)\\
\supset&\Cbb\otimes\SA_1\otimes\fk L(\MK_2)\otimes\fk L(\MH_2)\otimes\cdots\otimes\fk L(\MK_n)\otimes\fk L(\MH_n)\\
\supset& \Cbb\otimes\Cbb\otimes\SB_2\otimes\fk L(\MH_2)\otimes\cdots\otimes\fk L(\MK_n)\otimes\fk L(\MH_n)\\
\supset& \Cbb\otimes\Cbb\otimes\Cbb\otimes\SA_2\otimes\cdots\otimes\fk L(\MK_n)\otimes\fk L(\MH_n)\\
\supset&\cdots\\
\supset&\Cbb\otimes\Cbb\otimes\Cbb\otimes\Cbb\otimes\cdots\otimes\SB_n\otimes\fk L(\MH_n)\\
\supset&\Cbb\otimes\Cbb\otimes\Cbb\otimes\Cbb\otimes\cdots\otimes\Cbb\otimes\SA_n
\end{align*}
where each $\SA_i$ is a von Neumann algebra on $\MH_i$, and $\SB_j$ is a von Neumann algebra on $\MK_j$. Then our goal reduces to showing that the relative commutant of
\begin{align*}
\Cbb\otimes\SA_1\otimes\SB_2'\otimes\SA_2\otimes\SB_3'\otimes\cdots\otimes \SA_{n-1}\otimes\SB_n'\otimes\Cbb
\end{align*}
in
\begin{align*}
\SB_1\otimes\fk L(\MH_1)\otimes\fk L(\MK_2)\otimes\fk L(\MH_2)\otimes\fk L(\MK_3)\otimes\cdots\otimes \fk L(\MH_{n-1})\otimes\fk L(\MK_n)\otimes\SA_n'
\end{align*}
is equal to
\begin{align*}
\SB_1\otimes\SA_1'\otimes\SB_2\otimes\SA_2'\otimes\SB_3\otimes\cdots\otimes \SA_{n-1}'\otimes\SB_n\otimes\SA_n'
\end{align*}
But this is obvious.
\end{proof}

\section{The open CFT}\label{lb150}

Fix a conformal net $\MA$ with central charge $c$. In this chapter, all objects of $\RepA$ are assumed to be nonzero.

\subsection{Boundary intervals and double cones in $\Rop$}

\subsubsection{Boundary intervals in $\Rop$}

\begin{df}
In the $xt$-coordinates, define subsets of $\Roo$:
\begin{gather*}
\pmb{\Rop}=\{(x,t)\in\Roo:0\leq x\leq\pi\}\qquad \pmb{\Rzp}=\{(x,t)\in\Roo:0< x<\pi\}\\ 
\pmb{\partial_+\Rop}=\{0\}\times\Rbb\qquad\pmb{\partial_-\Rop}=\{\pi\}\times\Rbb
\end{gather*}
also viewed as subsets of $\Rcl$. Two points $p_1,p_2\in\Rop$ are called \textbf{spacelike separated} if they are so as elements of $\Rcl$, and hence can be described by Rem. \ref{lb76}.
\end{df}

\begin{df}\label{lb110}
The group action of $\SG$ on $\Rop$, which lifts to a group action of $\Gc$ on $\Rop$, is defined as follows. Consider the $u^+u^-$-coordinates. By viewing each $g\in\SG$ as a diffeomorphism of $\Rbb$ (cf. \eqref{eq2}), for each $(u^+,u^-)$ in $\Rop$, let
\begin{align*}
g(u^+,u^-)=(g(u^+),g(u^-))
\end{align*}
\end{df}

\begin{rem}\label{lb111}
The above action of $g$ on $\Rop$ is well-defined. In particular, note that in the $u^+u^-$-coordinates we have
\begin{align}\label{eq48}
\partial_+\Rop=\{(u,u):u\in\Rbb\}\qquad \partial_-\Rop=\{(u+2\pi,u):u\in\Rbb\}
\end{align}
Then clearly $g(\partial_+\Rop)=\partial_+\Rop$. Using the fact that $g(u+2\pi)=g(u)+2\pi$, cf. \eqref{eq2}, one immediately obtains $g(\partial_-\Rop)=\partial_-\Rop$.
\end{rem}

\begin{df}\label{lb113}
For each $\pm\in\{+,-\}$, let
\begin{gather*}
\pmb{\Jtd_\pm}=\{\text{non-empty open intervals in $\partial_\pm\Rop$ of lengths $<2\pi$}\}
\end{gather*}
Unless otherwise stated, we assume the equivalences
\begin{align}\label{eq49}
\Jtd_+\simeq\Jtd\simeq\Jtd_-
\end{align}
defined in the $xt$-coordinates by
\begin{subequations}\label{eq63}
\begin{align}
\{0\}\times\Rbb\simeq\Rbb\simeq\{\pi\}\times\Rbb\qquad (0,t)\simeq t\simeq (\pi,\pi+t)
\end{align}
equivalently, defined in the $u^+u^-$-coordinates (cf. \eqref{eq48}) by
\begin{align}\label{eq63b}
(u,u)\simeq u\simeq (u+2\pi,u)
\end{align}
\end{subequations}
\end{df}

For example, $\Std_+$ corresponds in the $xt$-coordinates to $\{(0,t):0<t<\pi\}$ in $\Jtd_+$, and $\Std_-$ corresponds in the $xt$-coordinates to $\{(\pi,t):0<t<\pi\}$ in $\Jtd_-$.

\begin{df}
By Def. \ref{lb110} and Rem. \ref{lb111}, the group actions of $\SG,\Gc$ on subsets of $\Rop$ restrict to group actions of $\SG,\Gc$ on $\Jtd_+$ and $\Jtd_-$, both equivalent to the actions of $\SG,\Gc$ on $\Jtd$ via the equivalence \eqref{eq49}.
\end{df}

\subsubsection{Double cones in $\Rzp$}\label{lb151}

\begin{df}
For each double cone $O\subset\Rpp$, by choosing unique $(\wtd I,\wtd J)$ satisfying \eqref{eq26} and corresponding to $O$, we define
\begin{align*}
\pmb{\sqrt O}=\text{the double cone corresponding to }\big(\sqrt{\wtd I},-\sqrt{\wtd J}\big)
\end{align*}
In other words, in the $u^+u^-$-coordinates,
\begin{align*}
\sqrt O=\Big\{\Big(\frac{u^+}2,\frac{u^-}2-\pi\Big):(u^+,u^-)\in O\Big\}
\end{align*}
Clearly $\sqrt O$ is a double cone contained in $\Rzp$.
\end{df}

\begin{df}
Recall Def. \ref{lb109} for the meaning of $\Opp$. Let
\begin{gather*}
\pmb{\Ozp}=\{\text{double cones }O\Subset\Rzp\}\\
\pmb{\ZOpp}=\{\sqrt O:O\in\Opp\}
\end{gather*}
Then clearly $\ZOpp\subset\Ozp$.
\end{df}

\begin{rem}\label{lb79}
One easily checks that double cones in $\Rzp$ are precisely those corresponding to some $(\wtd I,\wtd J)$ satisfying
\begin{align*}
0<\arg_I z-\arg_J\zeta<2\pi\qquad\text{for each }z\in I,\zeta\in J
\end{align*}
Therefore, $\Ozp$ consists precisely of double cones corresponding to some $(\wtd I,\wtd J)$ such that $\wtd J$ is clockwise to $\wtd I$, and that the closures of $I$ and $J$ are disjoint.

Moreover, $\ZOpp$ consists of $O\in\Ozp$ whose side lengths in the $u^+u^-$-coordinates are $<\pi$, that is, those corresponding to some $(\wtd I,\wtd J)$ such that $I,J$ have lengths $<\pi$. \hqed
\end{rem}

\subsection{The left boundary net $\Bopi$ and the right boundary net $\Bopi'$ for the boundary condition $i$}\label{lb131}

\begin{df}\label{lb116}
For each $\MH_i,\MH_j\in\RepA$, the \textbf{state space with boundary conditions $i,j$} is defined to be
\begin{align*}
\pmb{\Hopij}=\MH_i\boxtimes\MH_{\ovl j}
\end{align*}
where the fusion product is in $\RepA$.
\end{df}

Recall Thm. \ref{lb35} and \ref{lb63} for the definitions of $\pi^L,\pi^R$.

\begin{df}\label{lb112}
For each $\MH_i\in\RepA$, the \textbf{left resp. right boundary nets with boundary condition $i$} are collections of von Neumann algebras
\begin{gather}
\pmb{\Bopi}=(\Bopi(\wtd I))_{\wtd I\in\Jtd}\qquad\text{resp.}\qquad \pmb{\Bopi'}=(\Bopi'(\wtd I))_{\wtd I\in\Jtd}
\end{gather}
on $\Hopii=\MH_i\boxtimes\MH_{\ovl i}$ defined by
\begin{subequations}\label{eq52}
\begin{gather}
\Bopi(\wtd I)=\pi^L_{\wtd I}(\End_{\MA(I')}(\MH_i))\big|_{\MH_i\boxtimes\MH_{\ovl i}}\\
\Bopi'(\wtd I)=\pi^R_{\wtd I}(\End_{\MA(I')}(\MH_{\ovl i}))\big|_{\MH_i\boxtimes\MH_{\ovl i}}
\end{gather}
\end{subequations}
\end{df}

Recall the group action $U=U_{\op(i,j)}=U_{i\boxtimes\ovl j}$ of $\Gc$ on $\Hopij$ defined by Thm. \ref{lb12}. Recall the antiunitary map
\begin{align*}
\Theta_{i,\ovl j}:\Hopij\rightarrow \Hopji
\end{align*}
defined by Thm. \ref{lb49}, which satisfies
\begin{align*}
\Theta_{j,\ovl i}\Theta_{i,\ovl j}=1
\end{align*}
and hence is involutive when $i=j$.

\begin{thm}\label{lb114}
Let $\MH_i\in\RepA$. Then the nets $\Bopi$ and $\Bopi'$ satisfy the following properties for each $\wtd I,\wtd J\in\Jtd$.
\begin{enumerate}[label=(\arabic*)]
\item (Isotony) If $\wtd I\subset\wtd J$, then 
\begin{align*}
\Bopi(\wtd I)\subset\Bopi(\wtd J)\qquad \Bopi'(\wtd I)\subset\Bopi'(\wtd J)
\end{align*}
\item (Extension property) We have
\begin{align*}
\pi_{i\boxtimes\ovl i,I}(\MA(I))\subset \Bopi(\wtd I)\cap\Bopi'(\wtd I)
\end{align*}
\item (Conformal covariance) For each $g\in\Gc$, the operator $U(g)=U_{\opii}(g)$ satisfies
\begin{align*}
U(g)\Bopi(\wtd I)U(g)^*=\Bopi(g\wtd I)\qquad U(g)\Bopi'(\wtd I)U(g)^*=\Bopi'(g\wtd I)
\end{align*}
Moreover, if $g\in\Gc(I)$, then 
\begin{align*}
U_{\opii}(g)\in\pi_{i\boxtimes\ovl i,I}(\MA(I))
\end{align*}
and hence $U_{\opii}(g)\in\Bopi(\wtd I)$ by the extension property.
\item (PCT symmetry) The involutive anti-unitary operator $\Theta_{i,\ovl i}$ on $\Hopii$ satisfies
\begin{gather*}
\Theta_{i,\ovl i}U_{\op(i,i)}(g)\Theta_{i,\ovl i}=U_{\op(i,i)}(\rk g\rk) \\
\Theta_{i,\ovl i}\Bopi(\wtd I)\Theta_{i,\ovl i}=\Bopi'(\rk\wtd I)
\end{gather*}
for each $g\in\Gc$.
\item (Positive energy) The self-adjoint generator of the one-parameter group $u\mapsto U_{\op(i,i)}(\varrho_c(u))$ is positive.
\item (\textbf{Boundary-boundary Haag duality}) $\Bopi'$ is the \textbf{dual net} of $\Bopi$ in the sense of \cite[Sec. 4]{Gui21b}, that is, for each $\wtd I\in\Jtd$, 
\begin{align*}
\Bopi(\wtd I)'=\Bopi'(\wtd I')
\end{align*}
\item (Irreducibility) $\Bopi(\wtd I)$ and $\Bopi'(\wtd I)$ are type III factors. Moreover, $\bigcup_{\wtd I\in\Jtd}\Bopi(\wtd I)$ and $\bigcup_{\wtd I\in\Jtd}\Bopi'(\wtd I)$ together generate $\fk L(\Hopii)$.
\end{enumerate}
\end{thm}

As a consequence of the boundary-boundary Haag duality and the isotony, we obtain the \textbf{boundary-boundary locality}: if $\wtd J$ is clockwise to $\wtd I$, then
\begin{align*}
[\Bopi(\wtd I),\Bopi'(\wtd J)]=0
\end{align*}

\begin{proof}
Isotony: This is due to Rem. \ref{lb61} and the obvious fact that $\End_{\MA(I_1')}(\MH_i)\subset \End_{\MA(I_2')}(\MH_i)$ and $\End_{\MA(I_1')}(\MH_{\ovl i})\subset \End_{\MA(I_2')}(\MH_{\ovl i})$.

Extension property: Let $x\in\MA(I)$. By \eqref{eq28}, we have
\begin{align*}
\pi_{i\boxtimes\ovl i,I}(x)=\pi^L_{\wtd I}(\pi_{i,I}(x))\big|_{\MH_i\boxtimes\MH_{\ovl i}}
\end{align*}
where $\pi_{i,I}(x)$ belongs to $\End_{\MA(I')}(\MH_i)$. So $\pi_{i\boxtimes\ovl i,I}(x)\in\Bopi(\wtd I)$ by Def. \ref{lb112}. Similarly, \eqref{eq30} implies
\begin{align*}
\pi_{i\boxtimes\ovl i,I}(x)=\pi^R_{\wtd I}(\pi_{\ovl i,I}(x))\big|_{\MH_i\boxtimes\MH_{\ovl i}}
\end{align*}
where $\pi_{\ovl i,I}(x)\in\End_{\MA(I')}(\MH_{\ovl i})$, and hence $\pi_{i\boxtimes\ovl i,I}(x)\in\Bopi'(\wtd I)$.

Conformal covariance: The first part is due to Thm. \ref{lb108}. Suppose that $g\in\Gc(I)$. Then $U_{\opii}(g)$ equals $\pi_{i\boxtimes\ovl i,I}(U_0(g))$ by Thm. \ref{lb12}, and hence belongs to $\pi_{i\boxtimes\ovl i,I}(\MA(I))$.

PCT symmetry: The first relation is due to Prop. \ref{lb42}. The second is due to Cor. \ref{lb87}.

Positive energy: Similar to the proof of Thm. \ref{lb72}, this is due to \cite{Wei06}.

Haag duality: By Def. \ref{lb112},
\begin{align*}
\Bopi'(\Std_-)=\pi^R_{\Std_-}(\End_{\MA(\Sbb^1_+)}(\MH_{\ovl i}))\big|_{\MH_i\boxtimes\MH_{\ovl i}}=\Big(\pi^L_{\Std_+}(\End_{\MA(\Sbb^1_-)}(\MH_i))\big|_{\MH_i\boxtimes\MH_{\ovl i}}\Big)'
\end{align*}
where the last identity is due to Thm. \ref{lb47}. Since the last term is $\Bopi(\Std_+)'$, we obtain $\Bopi(\wtd I)'=\Bopi'(\wtd I')$ for $\wtd I=\Std_+$. The general case follows easily from this one using the conformal covariance established above.

Irreducibility: Similar to the proof of Thm. \ref{lb72}, the fact that $\Bopi(\wtd I)$ and $\Bopi'(\wtd I)$ are type III factors follows from Rem. \ref{lb91}. Thus $\Bopi(\wtd I)\cup\Bopi(\wtd I)'$ generates $\fk L(\Hopii)$. By the boundary-boundary Haag duality, $\Bopi(\wtd I)\cup\Bopi'(\wtd I')$ generates $\fk L(\Hopii)$.
\end{proof}

\begin{rem}\label{lb117}
One should view $\Bopi$ and $\Bopi'$ as nets with domains $\Jtd_+$ and $\Jtd_-$, respectively, by using the equivalence $\Jtd_+\simeq\Jtd\simeq\Jtd_-$ in Def. \ref{lb113}. Now, for each $\wtd I\in\Jtd$, if we view $\wtd I$ as in $\Jtd_+$, and if we view its clockwise complement $\wtd I'$ as in $\Jtd_-$, then one easily checks that
\begin{align*}
\wtd I'=\Int\{p\in\partial_-\Rop:p\text{ is spacelike separated from }\wtd I\}
\end{align*}
where the interior $\Int$ is with respect to $\partial_-\Rop$. In other words, $\wtd I'$ is the intersection of $\partial_-\Rop$ and the spacelike complement of $\wtd I$; see Fig. \ref{fig4}. This justifies the name ``boundary-boundary Haag duality'' in Thm. \ref{lb114}.
\end{rem}

\begin{figure}[h]
	\centering
\begin{gather*}
\includegraphics[height=3cm]{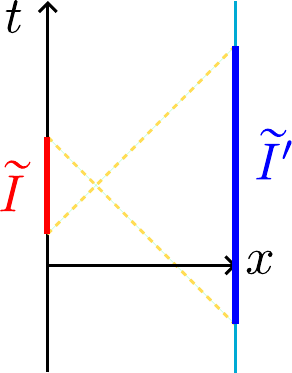}
   \end{gather*}
\caption{.~~Boundary-boundary Haag duality in $\Rop$}\label{fig4}
\end{figure}

\subsection{Boundary nets $\Bopi$ and $\Bopj'$ acting on the open-string state space $\Hop{i,j}=\MH_i\boxtimes\sbsc{\MA}\MH_{\ovl j}$}\label{lb130}

Recall from Def. \ref{lb116} that $\Hopij=\MH_i\boxtimes\MH_{\ovl j}$.

\begin{df}\label{lb118}
Let $\MH_i,\MH_j\in\RepA$ and $\wtd I\in\Jtd$. Define normal representations
\begin{gather*}
\pmb{\varpi^L_{\opij,\wtd I}}:\Bopi(\wtd I)\rightarrow\fk L(\Hopij)\qquad \pmb{\varpi^R_{\opij,\wtd I}}:\Bopj'(\wtd I)\rightarrow\fk L(\Hopij)
\end{gather*}
abbreviated to $\pmb{\varpi^L_{\wtd I}}$ and $\pmb{\varpi^R_{\wtd I}}$ when no confusion arises, such that the following diagrams commute:
\begin{gather*}
\begin{tikzcd}[ampersand replacement=\&]
\End_{\MA(I')}(\MH_i) \arrow[rr,"{\pi^L_{\wtd I}\big|_{\MH_i\boxtimes\MH_{\ovl j}}}"] \arrow[rd,"{\pi^L_{\wtd I}\big|_{\MH_i\boxtimes\MH_{\ovl i}}}"',"\simeq"] \&              \& \fk L(\MH_i\boxtimes\MH_{\ovl j}) \\
                        \& \Bopi(\wtd I) \arrow[ru,"\varpi^L_{\wtd I}"'] \&  
\end{tikzcd}\\
\begin{tikzcd}[ampersand replacement=\&]
\End_{\MA(I')}(\MH_{\ovl j}) \arrow[rr,"{\pi^R_{\wtd I}\big|_{\MH_i\boxtimes\MH_{\ovl j}}}"] \arrow[rd,"{\pi^R_{\wtd I}\big|_{\MH_j\boxtimes\MH_{\ovl j}}}"',"\simeq"] \&              \& \fk L(\MH_i\boxtimes\MH_{\ovl j}) \\
                        \& \Bopj'(\wtd I) \arrow[ru,"\varpi^R_{\wtd I}"'] \&  
\end{tikzcd}
\end{gather*}
\end{df}

\begin{rem}
The above definition makes sense because $\pi^L_{\wtd I}\big|_{\MH_i\boxtimes\MH_{\ovl i}}$ and $\pi^R_{\wtd I}\big|_{\MH_j\boxtimes\MH_{\ovl j}}$ are faithful, cf. Rem. \ref{lb91}. Note also that $\varpi^L_{\opii,\wtd I}$ and $\varpi^R_{\opjj,\wtd I}$ are inclusion maps.
\end{rem}

\begin{thm}\label{lb119}
Let $\MH_i,\MH_j\in\Rep(A)$. Then the following properties hold for each $\wtd I,\wtd J\in\Jtd$.
\begin{enumerate}[label=(\arabic*)]
\item (Isotony) If $\wtd I\subset\wtd J$, then
\begin{gather*}
\varpi^L_{\wtd J}\big|_{\Bopi(\wtd I)}=\varpi^L_{\wtd I}\qquad \varpi^R_{\wtd J}\big|_{\Bopj'(\wtd I)}=\varpi^R_{\wtd I}
\end{gather*}
\item (Extension property) For each $x\in\MA(I)$ and $y\in\MA(J)$, recalling that
\begin{gather*}
\pi_{i\boxtimes\ovl i}(x)\in\Bopi(\wtd I)\qquad \pi_{j\boxtimes\ovl j}(y)\in\Bopj'(\wtd J)
\end{gather*}
by the extension property in Thm. \ref{lb114}, we have
\begin{gather*}
\varpi^L_{\wtd I}(\pi_{i\boxtimes\ovl i,I}(x))=\pi_{i\boxtimes\ovl j,I}(x)\qquad \varpi^R_{\wtd J}(\pi_{j\boxtimes\ovl j,J}(y))=\pi_{i\boxtimes\ovl j,J}(y)
\end{gather*}
when acting on $\Hopij$.
\item (Conformal covariance) For each $g\in\Gc$ and each
\begin{align*}
A\in\Bopi(\wtd I)\qquad B\in\Bopj'(\wtd I)
\end{align*}
noting that $U(g)AU(g)^*$ belongs to $\Bopi(g\wtd I)$ and $U(g)BU(g)^*$ belongs to $\Bopj'(g\wtd I)$ by Thm. \ref{lb114}, we have
\begin{gather*}
U(g)\varpi^L_{\wtd I}(A)U(g)^*=\varpi^L_{g\wtd I}(U(g)AU(g)^*)\qquad\text{on }\Hopij\\ 
U(g)\varpi^R_{\wtd I}(B)U(g)^*=\varpi^R_{g\wtd I}(U(g)BU(g)^*)\qquad\text{on }\Hopij
\end{gather*}
\item (PCT symmetry) The anti-unitary map $\Theta_{i,\ovl j}:\Hopij\rightarrow\Hopji$ (with inverse $\Theta_{j,\ovl i}:\Hopji\rightarrow\Hopij$, cf. Rem. \ref{lb81}) satisfies
\begin{gather*}
\Theta_{i,\ovl j}\cdot U_{\op(i,j)}(g)\cdot\Theta_{j,\ovl i}=U_{\op(j,i)}(\rk g\rk)\qquad\text{on }\Hopji \\
\Theta_{i,\ovl j}\cdot\varpi^L_{\wtd I}(A)\cdot\Theta_{j,\ovl i}=\varpi^R_{\rk\wtd I}(\Theta_{i,\ovl i}A\Theta_{i,\ovl i})\qquad\text{on }\Hopji
\end{gather*}
for each $g\in\Gc$ and $A\in\Bopi(\wtd I)$, noting that $\Theta_{i,\ovl i}A\Theta_{i,\ovl i}\in\Bopi'(\rk\wtd I)$ by Thm. \ref{lb114}.
\item (Positive energy) The self-adjoint generator of the one-parameter group $u\mapsto U_{\op(i,j)}(\varrho_c(u))$ is positive.
\item(\textbf{Boundary-boundary Haag duality}) The following von Neumann algebras on $\Hopij$ are commutants of each other:
\begin{align}\label{eq51}
\varpi^L_{\wtd I}\big(\Bopi(\wtd I)\big)\qquad\varpi^R_{\wtd I'}\big(\Bopj'(\wtd I')\big)
\end{align}
\end{enumerate}
\end{thm}

The Haag duality in this theorem admits an interpretation analogous to that in Rem. \ref{lb117} and Fig. \ref{fig4}. Moreover, as a consequence of the isotony and the Haag duality, we obtain the \textbf{boundary-boundary locality}: if $\wtd J$ is clockwise to $\wtd I$, then
\begin{align*}
\big[\varpi^L_{\wtd I}\big(\Bopi(\wtd I)\big),\varpi^R_{\wtd J}\big(\Bopj'(\wtd J)\big)\big]=0
\end{align*}

\begin{proof}
Isotony: This is due to Rem. \ref{lb61}.

Extension property: By \eqref{eq28}, the element $\pi_{i,I}(x)\in\End_{\MA(I')}(\MH_i)$ is sent by $\pi^L_{\wtd I}|_{\MH_i\boxtimes\MH_{\ovl i}}$ and $\pi^L_{\wtd I}|_{\MH_i\boxtimes\MH_{\ovl j}}$ to $\pi_{i\boxtimes\ovl i}(x)$ and $\pi_{i\boxtimes\ovl j}(x)$, respectively. Thus $\varpi^L_{\wtd I}(\pi_{i\boxtimes\ovl i,I}(x))=\pi_{i\boxtimes\ovl j,I}(x)$ by Def. \ref{lb118}. A similar argument proves $\varpi^R_{\wtd J}(\pi_{j\boxtimes\ovl j,J}(y))=\pi_{i\boxtimes\ovl j,J}(y)$.

Conformal covariance: This is due to Thm. \ref{lb108}.

PCT symmetry: The first relation is due to Prop. \ref{lb42}. The second is due to Thm. \ref{lb107}.

Positive energy: Similar to the proof of Thm. \ref{lb72}, this is due to \cite{Wei06}.

Haag duality: By Def. \ref{lb118}, the two von Neumann algebras in \eqref{eq51} are equal to $\pi^L_{\wtd I}(\End_{\MA(I')}(\MH_i))$ and $\pi^R_{\wtd I'}(\End_{\MA(I)}(\MH_{\ovl j}))$ respectively. The latter two von Neumann algebras (acting on $\MH_i\boxtimes\MH_{\ovl j}$) are commutants of each other by Thm. \ref{lb47}.
\end{proof}

\subsection{Bulk net $\Bcl\big\vert_{\ZOpp}$ acting on $\Hopij$}

Fix $\MH_i,\MH_j\in\RepA$. The goal of this section and the next is to define the action of the restricted bulk net $\Bcl|_{\Ozp}$ on the open-string state space $\Hopij$, and to establish its key properties. In the present section, we construct this action for the restriction of $\Bcl|_{\Ozp}$ to $\ZOpp$, recalling that $\ZOpp\subset\Ozp$ (cf. Subsec. \ref{lb151}). In the next section, we extend the construction to the full net $\Bcl|_{\Ozp}$.

\subsubsection{Construction of the action $\Bcl\big|_{\ZOpp}\curvearrowright\Hopij$}

\begin{df}\label{lb127}
For each $O\in\Opp$, corresponding to $(\wtd I,\wtd J)$ satisfying \eqref{eq26}, choose $g_+,g_-\in\SG$ such that by viewing $g_\pm$ as in \eqref{eq2}, and by viewing $\wtd I,\wtd J$ as subsets of the universal cover $\Rbb$ of $\Rbb/2\pi\Zbb=\Sbb^1$, we have
\begin{gather*}
g_+(\theta)=\theta/2\qquad \text{for each }\theta\in \wtd I\\
g_-(\theta)=\theta/2-\pi\qquad \text{for each }\theta\in \wtd J
\end{gather*}
Note that $(g_+,g_-)O=\sqrt O$. Lift $g_+,g_-$ to elements of $\Gc$, still denoted by $g_+,g_-$. By the conformal covariance in Thm. \ref{lb72}, we can define (spatial) isomorphisms of von Neumann algebras
\begin{gather*}
\Qbf_{\sqrt O}:\Bcl(O)\rightarrow\Bcl(\sqrt O)\qquad A\mapsto \Ucl(g_+,g_-)A\Ucl(g_+,g_-)^*
\end{gather*}
This definition is independent of the choice $g_\pm$. 
\end{df}

\begin{proof}[Explanation]
To prove the independence of $\Qbf_{\sqrt O}$ on $g_\pm$, assume that $h_\pm$ satisfy the same properties as $g_\pm$. Then $h_+^{-1}g_+$ fixes each point of $\wtd I$. Hence $h_+^{-1}g_+\in\Gc(I')$ by Rem. \ref{lb129}. Similarly, $h_-^{-1}g_-\in\Gc(J')$. Then, by choosing suitable arg-functions on $I',J'$, the spacelike complement $O'$ of $O$ corresponds to $((I',\arg_{I'}),(J',\arg_{J'}))$. By the conformal covariance in Thm. \ref{lb72}, we have
\begin{align*}
\Ucl(h_+,h_-)^{-1}\Ucl(g_+,g_-)\in\Bcl(O')
\end{align*}
Therefore, by the locality in Thm. \ref{lb72},  $\Ad_{\Ucl(g_+,g_-)}$ equals $\Ad_{\Ucl(h_+,h_-)}$ when applied to $\Bcl(O)$.
\end{proof}

\begin{rem}\label{lb128}
Clearly, if $O_1\subset O_2$ are in $\Opp$, then $\Qbf_{\sqrt{O_2}}\big|_{\Bcl(O_1)}=\Qbf_{\sqrt{O_1}}$.
\end{rem}

\begin{df}\label{lb120}
For each $O\in\Opp$, which corresponds uniquely to $(\wtd I,\wtd J)$ satisfying \eqref{eq26}, choose $\wtd K_+,\wtd K_-\in\Jtd$ such that
\begin{align}\label{eq62}
\sqrt{\wtd I},\wtd K_+, -\sqrt{\wtd J},\wtd K_- \qquad\text{are in clockwise order}
\end{align}
Choose $K_+$-unitary $\xi_0\in\MH_i(K_+)$ and $K_-$-unitary $\eta_0\in\MH_{\ovl j}(K_-)$, and define a faithful normal representation
\begin{align*}
\pmb{\varpi_{\opij,\sqrt O}}:\Bcl(\sqrt O)\rightarrow\fk L(\Hopij)
\end{align*}
abbreviated to $\pmb{\varpi_{\sqrt O}}$ when no confusion arises, such that by setting
\begin{align*}
V(\xi_0,\eta_0)=L(\xi_0,\wtd K_+)R(\eta_0,\wtd K_-)\equiv R(\eta_0,\wtd K_-)L(\xi_0,\wtd K_+): \MH_0\rightarrow\Hopij
\end{align*}
and by choosing any $\wtd K\in\Jtd$ containing $\wtd I,\wtd J$ (which exists by Rem. \ref{lb56}), the following diagram commutes:
\begin{equation}\label{eq61}
\begin{tikzcd}[row sep=1cm,column sep=3cm, ampersand replacement=\&]
\fk A(\sqrt I,-\sqrt J) \arrow[r,"{\Ad_{V(\xi_0,\eta_0)}}"] \arrow[d,"{\pi^L_{\wtd K}\big|_{\MH_\sqz\boxtimes\MH_{\ovl\sqz}}}"',"\simeq"]     \& \fk L(\MH_i\boxtimes\MH_{\ovl j}) \\
 \Bcl(O) \arrow[r,"\Qbf_{\sqrt O}","\simeq"'] \& \Bcl(\sqrt O) \arrow[u,"\varpi_{\sqrt O}"']
\end{tikzcd}
\end{equation}
\end{df}

Note that the isomorphism from the upper-left corner to the lower-left corner is given by Def. \ref{lb64}. 

\begin{rem}\label{lb148}
The definition of $\varpi_{\sqrt O}$ is independent of $\wtd K$, as indicated in the explanation to Def. \ref{lb64}, and also by Rem. \ref{lb128}. It is independent of $\wtd K_+,\wtd K_-$ due to the isotony property for $L,R$ operators (cf. Cor. \ref{lb31}). It is also independent of the choice of $\xi_0,\eta_0$, as explained in the following proof.
\end{rem}

\begin{proof}
Choose $K_+$-unitary $\xi_1\in\MH_i(K_+)$ and $K_-$-unitary $\eta_1\in\MH_{\ovl j}(K_-)$. Then $L(\xi_0,\wtd K_+)^*L(\xi_1,\wtd K_+)|_{\MH_0}$ commutes with $\MA(K_+')$, and hence belongs to $\MA(K_+)$ by the Haag duality for $\MA$. Thus
\begin{align*}
L(\xi_1,\wtd K_+)|_{\MH_0}=L(\xi_0,\wtd K_+)x|_{\MH_0}
\end{align*} 
for some (necessarily unitary) $x\in\MA(K_+)$. Thus
\begin{align*}
V(\xi_1,\eta_0)=V(\xi_0,\eta_0)x
\end{align*}
By Def. \ref{lb121}, elements of $\MA(K_+)$ commute with those of $\fk A(\sqrt I,-\sqrt J)$. Therefore $\Ad_{V(\xi_1,\eta_0)}=\Ad_{V(\xi_0,\eta_0)}$ on $\fk A(\sqrt I,-\sqrt J)$. A similar argument shows $\Ad_{V(\xi_1,\eta_1)}=\Ad_{V(\xi_1,\eta_0)}$ on $\fk A(\sqrt I,-\sqrt J)$. Hence $\Ad_{V(\xi_1,\eta_1)}=\Ad_{V(\xi_0,\eta_0)}$ on $\fk A(\sqrt I,-\sqrt J)$.
\end{proof}


\subsubsection{Basic properties of the action $\Bcl\big|_{\ZOpp}\curvearrowright\Hopij$}

\begin{pp}\label{lb124}
The following properties hold for each $O,O_1,O_2\in\Opp$.
\begin{enumerate}[label=(\arabic*)]
\item (Isotony) If $O_1\subset O_2$, then
\begin{align*}
\varpi_{\sqrt{O_2}}\big|_{\Bcl(\sqrt{O_1})}=\varpi_{\sqrt{O_1}}
\end{align*}
\item (Bulk-bulk locality) Suppose that $O_1,O_2$ are spacelike separated (and hence $\sqrt{O_1},\sqrt{O_2}$ are spacelike separated). Then
\begin{align*}
\big[\varpi_{\sqrt{O_1}}(\Bcl(\sqrt{O_1})),\varpi_{\sqrt{O_2}}(\Bcl(\sqrt{O_2}))\big]=0
\end{align*}
\item (Extension property) Let $O$ correspond to $(\wtd I,\wtd J)$. For each $x\in\MA(\sqrt I)$ and $y\in\MA(-\sqrt J)$, noting that
\begin{align*}
\pi^+_{\sqz\boxtimes\ovl\sqz,\sqrt I}(x)\in\Bcl(\sqrt O)\qquad \pi^-_{\sqz\boxtimes\ovl\sqz,-\sqrt J}(y)\in\Bcl(\sqrt O)
\end{align*}
by the extension property in Thm. \ref{lb72}, we have
\begin{gather*}
\varpi_{\sqrt O}\big(\pi^+_{\sqz\boxtimes\ovl\sqz,\sqrt I}(x)\big)=\pi_{i\boxtimes \ovl j,\sqrt I}(x)\\
\varpi_{\sqrt O}\big(\pi^-_{\sqz\boxtimes\ovl\sqz,-\sqrt J}(y)\big)=\pi_{i\boxtimes \ovl j,-\sqrt J}(y)
\end{gather*}
\item (Conformal covariance) Suppose that $g\in\Gc$ and $O\in\Opp$ satisfy $(g,g)\sqrt O\in\ZOpp$. Then for each $A\in\Bcl(\sqrt O)$, noting that $\Ucl(g,g)A\Ucl(g,g)^*$ belongs to $\Bcl((g,g)\sqrt O)$ by the conformal covariance in Thm. \ref{lb72}, we have
\begin{align*}
U_{i\boxtimes\ovl j}(g)\varpi_{\sqrt O}(A)U_{i\boxtimes\ovl j}(g)^*=\varpi_{(g,g)\sqrt O}(\Ucl(g,g)A\Ucl(g,g)^*)
\end{align*}
\item (Bulk-boundary Haag duality) Let $\wtd E_+$ and $\wtd E_-$ be the intersections of the spacelike complement of $\sqrt O$ with $\partial_+\Rop$ and $\partial_-\Rop$ respectively, i.e., for each $\pm\in\{+,-\}$,
\begin{gather*}
\wtd E_\pm=\Int\{p\in\partial_\pm\Rop:p\text{ is spacelike separated from elements of }\sqrt O\}
\end{gather*}
where the interior is with respect to $\partial_\pm\Rop$. Then $\wtd E_+\in\Jtd_+$ and $\wtd E_-\in\Jtd_-$. Moreover, by viewing $\wtd E_+,\wtd E_-\in\Jtd$ via the equivalence \eqref{eq49}, we have
\begin{align*}
\varpi^L_{\wtd E_+}\big(\Bopi(\wtd E_+)\big)\vee \varpi^R_{\wtd E_-}\big(\Bopj'(\wtd E_-)\big)=\varpi_{\sqrt O}\big(\Bcl(\sqrt O)\big)'
\end{align*}
\end{enumerate}
\end{pp}

\begin{proof}
Isotony: This is obvious.

Bulk-bulk locality: Since $\varpi_{\sqrt O}$ is normal for each $O$, by the additivity of $\Bcl$ (Cor. \ref{lb75}) and Lem. \ref{lb77}, it suffices to assume that $O_1,O_2\subset O$ for some standard $O\in\Opp$. In that case, we have $\varpi_{\sqrt{O_i}}(\Bcl(\sqrt{O_i}))=\varpi_{\sqrt O}(\Bcl(\sqrt{O_i}))$. Thus the commutativity follows from the locality in Thm. \ref{lb72}.

Extension property: See Lem. \ref{lb136}.

Conformal covariance: See Lem. \ref{lb125}.

Bulk-boundary Haag duality: Let $(\wtd I,\wtd J)$ satisfy \eqref{eq26} and correspond to $O$. Choose $\wtd K_+,\wtd K_-\in\Jtd$ such that $K_+,K_-$ are the two components of the interior of $\Sbb^1\setminus(\sqrt I\cup-\sqrt J)$, and satisfying \eqref{eq62}. By viewing $\wtd K_+\in\Jtd_+$ and $\wtd K_-\in\Jtd_-$, one checks (e.g. by projecting $\wtd E_+,\wtd E_-$ onto the $u^-$-axis, cf. \eqref{eq63b}) that
\begin{align}\label{eq79}
\wtd E_+=\wtd K_+\qquad \wtd E_-=\wtd K_-
\end{align}
This proves the first part of the Haag duality.

By Def. \ref{lb120}, we have
\begin{align*}\label{eq64}
\varpi_{\sqrt O}\big(\Bcl(\sqrt O)\big)=\Ad_{V(\xi_0,\eta_0)}\big(\fk A(\sqrt I,-\sqrt J)\big)  \tag{a}
\end{align*}
By Prop. \ref{lb34}, we have
\begin{align*}
\End_{\MA(E_+')}(\MH_i)=\Ad_{L(\xi_0,\wtd E_+)|_{\MH_0}}(\MA(E_+))
\end{align*}
Therefore, since $\wtd E_-$ is clockwise to $\wtd E_+$, by \eqref{eq15} we have
\begin{align*}
\pi^L_{\wtd E_+}\big(\End_{\MA(E_+')}(\MH_i)\big)\big|_{\MH_i\boxtimes\MH_{\ovl j}}=\Ad_{R(\eta_0,\wtd E_-)}\big(\End_{\MA(E_+')}(\MH_i)\big)=\Ad_{V(\xi_0,\eta_0)}(\MA(E_+))
\end{align*}
It follows from Def. \ref{lb118} that
\begin{align*}\label{eq65}
\varpi^L_{\wtd E_+}\big(\Bopi(\wtd E_+)\big)=\Ad_{V(\xi_0,\eta_0)}(\MA(E_+))  \tag{b}
\end{align*}
A similar argument shows that
\begin{align*}\label{eq66}
\varpi^R_{\wtd E_-}\big(\Bopj'(\wtd E_-)\big)=\Ad_{V(\xi_0,\eta_0)}(\MA(E_-))  \tag{c}
\end{align*}
Combining \eqref{eq64}, \eqref{eq65}, \eqref{eq66} with the relation (cf. Def. \ref{lb121})
\begin{align*}
\MA(E_+)\vee\MA(E_-)=\fk A(\sqrt I,-\sqrt J)'
\end{align*} 
the proof of the second part of the Haag duality is finished.
\end{proof}

\begin{lm}\label{lb136}
The extension property in Prop. \ref{lb124} holds.
\end{lm}

\begin{proof}
We prove the relation for $x$. The relation for $y$ can be proved in a similar way. Choose $g_+,g_-\in\Gc$ as in Def. \ref{lb127}. Let $\wht x\in\MA(I)$ such that the isomorphism $\Qbf_{\sqrt{\wtd I}}:\MA(I)\rightarrow\MA(\sqrt I)$ in Def. \ref{lb50} sends $\wht x$ to $x$, that is,
\begin{align*}\label{eq69}
\Qbf_{\sqrt{\wtd I}}(\wht x)\equiv U_0(g_+)\wht xU_0(g_+)^*=x \tag{$\star$}
\end{align*}
Choose $\xi_0,\eta_0$ and $\wtd K$ as in Def. \ref{lb120}. In view of the commutative diagram \eqref{eq61}, the proof will be finished if we can show for the element $x\in\MA(\sqrt I)$ that
\begin{align*}
\begin{tikzcd}[row sep=1cm,column sep=3cm, ampersand replacement=\&]
x \arrow[r,maps to,"{\Ad_{V(\xi_0,\eta_0)}}"] \arrow[d,maps to,"{\pi^L_{\wtd K}\big|_{\MH_\sqz\boxtimes\MH_{\ovl\sqz}}}"']     \& \pi_{i\boxtimes\ovl j,\sqrt I}(x) \\
\pi^+_{\sqz\boxtimes\ovl\sqz,I}(\wht x) \arrow[r,maps to,"\Qbf_{\sqrt O}"] \& \pi^+_{\sqz\boxtimes\ovl\sqz,\sqrt I}(x) 
\end{tikzcd}
\end{align*}

The top horizontal arrow is due to \eqref{eq68}. By \eqref{eq69} and Def. \ref{lb52}, we have
\begin{align*}
\pi_{\sqz,\wtd I}(\wht x\otimes 1)=x
\end{align*}
By viewing $x\in \fk A(\sqrt I,-\sqrt J)\subset\End_{\Maa(\wtd K')}(\MH_\sqz)$ (cf. \eqref{eq70}), it follows that
\begin{align*}
&\pi^L_{\wtd K}(x)\big|_{\MH_\sqz\boxtimes\MH_{\ovl\sqz}}=\pi^L_{\wtd K}(\pi_{\sqz,\wtd I}(\wht x\otimes 1))\big|_{\MH_\sqz\boxtimes\MH_{\ovl\sqz}}\\
\xlongequal{\text{Rem.\ref{lb61}}}&\pi^L_{\wtd I}(\pi_{\sqz,\wtd I}(\wht x\otimes 1))\big|_{\MH_\sqz\boxtimes\MH_{\ovl\sqz}}\xlongequal{\eqref{eq28}}\pi_{\sqz\boxtimes\ovl\sqz,I}(\wht x\otimes 1)
\end{align*}
where the last term is exactly $\pi^+_{\sqz\boxtimes\ovl\sqz,I}(\wht x)$ by Def. \ref{lb83}. This proves the left vertical arrow. By Def. \ref{lb127},
\begin{align*}
\Qbf_{\sqrt O}\big(\pi^+_{\sqz\boxtimes\ovl\sqz,I}(\wht x)\big)=\Ad_{\Ucl(g_+,g_-)}\big(\pi^+_{\sqz\boxtimes\ovl\sqz,I}(\wht x)\big)
\end{align*}
Recall from Def. \ref{lb134} that $\Ucl(g_+,g_-)=U^+_{\sqz\boxtimes\ovl\sqz}(g_+)U^-_{\sqz\boxtimes\ovl\sqz}(g_-)$. By Rem. \ref{lb104}, $U^-_{\sqz\boxtimes\ovl\sqz}(g_-)$ commutes with $\pi^+_{\sqz\boxtimes\ovl\sqz,I}(\wht x)$. Hence
\begin{align*}
&\Qbf_{\sqrt O}\big(\pi^+_{\sqz\boxtimes\ovl\sqz,I}(\wht x)\big)=\Ad_{U^+_{\sqz\boxtimes\ovl\sqz}(g_+)}\big(\pi^+_{\sqz\boxtimes\ovl\sqz,I}(\wht x)\big)\\
\xlongequal{\text{Cor.\ref{lb30}}}&\pi^+_{\sqz\boxtimes\ovl\sqz,\sqrt I}(U_0(g_+)\wht x U_0(g_+)^*)\xlongequal{\eqref{eq69}}\pi^+_{\sqz\boxtimes\ovl\sqz,\sqrt I}(x)
\end{align*}
This proves the bottom horizontal arrow.
\end{proof}

\begin{lm}\label{lb125}
The conformal covariance in Prop. \ref{lb124} holds.
\end{lm}

\begin{proof}
Step 1. In this step, we prove the relation
\begin{align*}
U_{i\boxtimes\ovl j}(g)\varpi_{\sqrt O}(A)U_{i\boxtimes\ovl j}(g)^*=\varpi_{(g,g)\sqrt O}(\Ucl(g,g)A\Ucl(g,g)^*)  \tag{a}\label{eq77}
\end{align*}
when $O\in\Opp$ corresponds to $(\wtd I,\wtd J)$, $A\in\Bcl(\sqrt O)$, and $g\in\Gc$ fixes pointwise $\sqrt {\wtd I}\cup -\sqrt {\wtd J}$ (by viewing $\Sbb^1=\Rbb/2\pi\Zbb$), i.e., $g\in\Gc((\sqrt I)')\cap\Gc((-\sqrt J)')$). Note that $(g,g)\sqrt O=\sqrt O\in\ZOpp$.

By the conformal covariance in Thm. \ref{lb72}, we have $\Ucl(g,g)\in\Bcl(\sqrt O')$. Therefore, by the locality in Thm. \ref{lb72}, the RHS of \eqref{eq77} equals $\varpi_{\sqrt O}(A)$. Let us prove that the LHS equals $\varpi_{\sqrt O}(A)$.

Let $\wtd K_\pm$ and $\xi_0\in\MH_i(K_+),\eta_0\in\MH_{\ovl j}(K_-)$ be as in Def. \ref{lb120}, and assume more over that the interior of $\Sbb^1\setminus(\sqrt I\cup-\sqrt J)$ is the disjoint union of $K_+$ and $K_-$. By the commuting diagram \eqref{eq61}, there exists $\fk a\in \fk A(\sqrt I,-\sqrt J)$ such that
\begin{align*}
\varpi_{\sqrt O}(A)=\Ad_{V(\xi_0,\eta_0)}(\fk a)
\end{align*}
By the conformal covariance in Thm. \ref{lb15} and the fact that $gK_+=K_+$ and $gK_-=K_-$, we have
\begin{align*}
U_{i\boxtimes\ovl j}(g)V(\xi_0,\eta_0)=V(\xi_1,\eta_1)U_0(g)
\end{align*}
where $\xi_1=g\xi_0g^{-1}\in\MH_i(K_+)$ is $K_+$-unitary, and $\eta_1=g\eta_0g^{-1}\in\MH_{\ovl j}(K_-)$ is $K_-$-unitary, cf. Exp. \ref{lb147}. On the other hand, since $g$ fixes points of $\sqrt I\cup -\sqrt J$, we can write $g=g_1g_2$ where $g_1\in\Gc(K_+)$ and $g_2\in\Gc(K_-)$. So $U_0(g)$ belongs to $\MA(K_+)\vee\MA(K_-)$, and hence commutes with any element of $\fk A(\sqrt I,-\sqrt J)$. Thus
\begin{align*}
U_{i\boxtimes\ovl j}(g)\varpi_{\sqrt O}(A)U_{i\boxtimes\ovl j}(g)^*=\Ad_{U_{i\boxtimes\ovl j}(g)}\Ad_{V(\xi_0,\eta_0)}(\fk a)=\Ad_{V(\xi_1,\eta_1)}\Ad_{U_0(g)}(\fk a)=\Ad_{V(\xi_1,\eta_1)}(\fk a)
\end{align*}
where the last term equals $\varpi_{\sqrt O}(A)$ by Rem. \ref{lb148}.\\[-1ex]

Step 2. In this step, we prove \eqref{eq77} for $O\in\Opp$ and $A\in\Bcl(\sqrt O)$ and $g\in\Gc$, under the additional assumption that $O$ is contained in a standard double cone $D$ corresponding to $(\wtd K,\wtd K)$ where $\wtd K\in\Jtd$, and that $g$ fixes points outside $(\sqrt {\wtd K}\cup-\sqrt {\wtd K})+2\pi\Zbb$ (by viewing $\Sbb^1=\Rbb/2\pi\Zbb$). Note that in this case, we have $(g,g)\sqrt O\subset (g,g)\sqrt{D}=\sqrt{D}\in\ZOpp$. By the isotony in Prop. \ref{lb124}, it suffices to prove
\begin{align*}\label{eq78}
U_{i\boxtimes\ovl j}(g)\varpi_{\sqrt{D}}(A)U_{i\boxtimes\ovl j}(g)^*=\varpi_{\sqrt{D}}(\Ucl(g,g)A\Ucl(g,g)^*)  \tag{b}
\end{align*}

By assumption, we can write $g=g_1g_2=g_2g_1$ where $g_1\in\Gc(\sqrt K)$ and $g_2\in\Gc(-\sqrt K)$. Therefore, by Thm. \ref{lb12} and the extension property in Prop. \ref{lb124} (proved in Lem. \ref{lb136}), we have
\begin{align*}
U_{i\boxtimes\ovl j}(g_1)=\pi_{i\boxtimes\ovl j,\sqrt K}(U_0(g_1))=\varpi_{\sqrt D}\big(\pi^+_{\sqz\boxtimes\ovl\sqz,\sqrt K}(U_0(g_1))\big)=\varpi_{\sqrt D}(\Ucl^+(g_1))
\end{align*} 
and, similarly,
\begin{align*}
 U_{i\boxtimes\ovl j}(g_2)=\pi_{i\boxtimes\ovl j,-\sqrt K}(U_0(g_2))=\varpi_{\sqrt D}\big(\pi^-_{\sqz\boxtimes\ovl\sqz,-\sqrt K}(U_0(g_2))\big)=\varpi_{\sqrt D}(\Ucl^-(g_2))
\end{align*}
Therefore
\begin{align*}
U_{i\boxtimes\ovl j}(g)=\varpi_{\sqrt D}(\Ucl^+(g_1)\Ucl^-(g_2))=\varpi_{\sqrt D}(\Ucl(g_1,g_2))
\end{align*}

By the conformal covariance in Thm. \ref{lb72}, the element $\Ucl(g_2,g_1)$ belongs to $\Bcl(\sqrt D')$, and hence commutes with $A$ by the locality in Thm. \ref{lb72}. Therefore, from the fact that $(g,g)=(g_1,g_2)(g_2,g_1)$ we conclude
\begin{align*}
\Ad_{\Ucl(g_1,g_2)}(A)=\Ad_{\Ucl(g,g)}(A)
\end{align*}
and hence the LHS of \eqref{eq78} equals
\begin{align*}
\Ad_{\varpi_{\sqrt D}(\Ucl(g_1,g_2))}(\varpi_{\sqrt D}(A))=\varpi_{\sqrt D}(\Ad_{\Ucl(g_1,g_2)}(A))=\varpi_{\sqrt D}(\Ad_{\Ucl(g,g)}(A))
\end{align*}
where the last term is the RHS of \eqref{eq78}.\\[-1ex]

Step 3. Let $X$ be a (smooth) $2\pi$-periodic vector field on $\Rbb$. The flow $\alpha^X:s\in\Rbb\mapsto\alpha^X_s$ generated by this vector field is a one-parameter subgroup of $\SG$, cf. \eqref{eq2} for the description of elements of $\SG$. For each $s$, we choose a lift of $\alpha^X_s$ to $\Gc$, also denoted by $\alpha^X_s$.

In this step, we prove that for each $O\in\Opp$, there exists $\eps>0$ such that for each $-\eps\leq s\leq\eps$, we have $(\alpha^X_s,\alpha^X_s)\sqrt O\in\ZOpp$, and that \eqref{eq77} holds for any $A\in\Bcl(\sqrt O)$ and $g=\alpha^X_s$. The structure of the proof is similar to that of Lem. \ref{lb69}.

Let $O$ correspond to $(\wtd I,\wtd J)$. Since $O$ is compactly contained in a standard double cone, by viewing $\Sbb^1=\Rbb/2\pi\Zbb$, there exists $\wtd K\in\Jtd$ such that $\wtd I,\wtd J\Subset \wtd K$. Let
\begin{align*}
\wtd E_3=\sqrt{\wtd K}\qquad \wtd F_3=-\sqrt{\wtd K}
\end{align*}
and choose $\wtd E_1,\wtd E_2,\wtd F_1,\wtd F_2\in\Jtd$ and $\eps>0$ such that
\begin{gather*}
\wtd E_1\Subset\wtd E_2\Subset \wtd E_3\qquad \wtd F_1\Subset\wtd F_2\Subset \wtd F_3\\
\ovl{\bigcup_{|s|\leq\eps}\alpha^X_s\big(\sqrt{\wtd I}}\big)\subset \wtd E_1\qquad  \ovl{\bigcup_{|s|\leq\eps}\alpha^X_s\big(-\sqrt{\wtd J}}\big)\subset \wtd F_1
\end{gather*}
In particular, we have $\sqrt{\wtd I}\subset\wtd E_1$ and $-\sqrt{\wtd J}\subset\wtd F_1$.

Let $Y$ be a $2\pi$-periodic vector field on $\Rbb$ which equals $X$ on $\wtd E_2\cup\wtd F_2$, and which vanishes outside $(\wtd E_3\cup\wtd F_3)+2\pi\Zbb$. For each $s$, we choose a lift of $\alpha^Y_s$ to $\Gc$, also denoted by $\alpha^Y_s$. 

By shrinking $\eps$, we have
\begin{gather*}
\ovl{\bigcup_{|s|\leq\eps}\alpha^X_s(\wtd E_1)}\subset \wtd E_2\qquad \ovl{\bigcup_{|s|\leq\eps}\alpha^X_s(\wtd F_1)}\subset \wtd F_2
\end{gather*}
We now fix $s$ satisfying $|s|\leq\eps$. Then
\begin{gather*}
h:=(\alpha_s^Y)^{-1}\alpha_s^X
\end{gather*}
fixes pointwise $\wtd E_1\cup\wtd F_1$, and hence satisfies the assumption in Step 1. Therefore, for each $A\in\Bcl(\sqrt O)$, Step 1 implies
\begin{align*}
\Ad_{U_{i\boxtimes\ovl j}(h)}\big(\varpi_{\sqrt O}(A)\big)=\varpi_{(h,h)\sqrt O}(\Ad_{\Ucl(h,h)}(A))
\end{align*}
and $(h,h)\sqrt O=\sqrt O$.

Since $\alpha_s^Y$ satisfies the assumption in Step 2, we have
\begin{align*}
&\Ad_{U_{i\boxtimes\ovl j}(\alpha_s^Y)}\big(\varpi_{(h,h)\sqrt O}(\Ad_{\Ucl(h,h)}(A))\big)=\varpi_{(\alpha_s^X,\alpha_s^X)\sqrt O}\big(\Ad_{\Ucl(\alpha_s^Y,\alpha_s^Y)}\Ad_{\Ucl(h,h)}(A))\big)\\
=&\varpi_{(\alpha_s^X,\alpha_s^X)\sqrt O}\big(\Ad_{\Ucl(\alpha_s^X,\alpha_s^X)}(A))\big)
\end{align*}
where $(\alpha_s^X,\alpha_s^X)\sqrt O\in\ZOpp$. Combining the above two computations together, we obtain \eqref{eq77} for $g=\alpha^X_s$.\\[-1ex]

Step 4. Let $\mbf G$ be the set of all $g\in\Gc$ such that the relation \eqref{eq77} holds for each $O\in\Opp$ satisfying $(g,g)\sqrt O\in\ZOpp$, and for each $A\in\Bcl(\sqrt O)$. The proof of the conformal covariance in Prop. \ref{lb124} will be finished by showing that $\mbf G=\Gc$.

By \cite{Eps70,Her71,Thu74}, $\DiffS$ is a simple group. Thus, since the subsemigroup generated (algebraically) by the exponentials of vector fields of $\Sbb^1$ is a normal subgroup, it is equal to $\DiffS$. It follows that the subsemigroup generated (algebraically) by the exponentials of $2\pi$-periodic vector fields of $\Rbb$ is equal to $\SG$. Therefore, it suffices to show that $\mbf G$ is a subsemigroup of $\Gc$ containing any lift of $\alpha^X_s$ where $X$ is a $2\pi$-periodic vector field of $\Rbb$, and $s\in\Rbb$.

Suppose that $g_1,g_2\in\mbf G$. By uniform continuity, there exists $\delta>0$ such that for each $O\in\Opp$ with side lengths $<\delta$, the double cones $(g_2,g_2)\sqrt O$ and $(g_1g_2,g_1g_2)\sqrt O$ have side lengths $<\pi$ in the $u^+u^-$-coordinates, and hence belong to $\ZOpp$, cf. Rem. \ref{lb79}. Then \eqref{eq77} holds for all $O\in\Opp$ with side lengths $<\delta$, and for $A\in\Bcl(\sqrt O)$ and $g=g_1g_2$. By the additivity of $\Bcl$ (Cor. \ref{lb75}), \eqref{eq77} holds for all $O\in\Opp$ with $(g_1g_2,g_1g_2)\sqrt O\in\ZOpp$, and for $A\in\Bcl(\sqrt O)$ and $g=g_1g_2$. Thus $g_1g_2\in\mbf G$. This proves that $\mbf G$ is a subsemigroup.

Now let $X$ a $2\pi$-periodic vector field of $\Rbb$. For each $s\in\Rbb$, choose an arbitrary lift of $\alpha^X_s$, also denoted by $\alpha^X_s$. (Note that we have $\alpha_{s_1}^X\alpha_{s_2}^X=\alpha_{s_1+s_2}^X$ up to a phase, and any phase does not affect the proof of \eqref{eq77}.) Choose any $R>0$. Let us prove that $\alpha^X_s\in\mbf G$ for each $|s|\leq R$. Choose any $O\in\Opp$. 

In the special case where $(\alpha^X_s,\alpha^X_s)\sqrt O\in\ZOpp$ for each $|s|\leq R$, Step 3 implies that for each $s\in[-R,R]$, there exists $\eps_s>0$ such that \eqref{eq77} holds for $g=\alpha^X_t$ whenever $|t|\leq\eps_s$, and for any $A\in\Bcl((\alpha_s,\alpha_s)\sqrt O)$. Thus \eqref{eq77} holds for $g=\alpha^X_{s_2-s_1}$ and $A\in\Bcl((\alpha_{s_1},\alpha_{s_1})\sqrt O)$ whenever $s_1,s_2\in[s-\eps_s,s+\eps_s]$. It follows from the compactness of $[-R,R]$ (and the Lebesgue number lemma) that \eqref{eq77} holds for $g=\alpha^X_s$ and $A\in\Bcl(\sqrt O)$, where $s\in[-R,R]$.

In the general case, each $p\in O$ is contained in a double cone $O_p\subset O$ such that $(\alpha_s^X,\alpha^X_s)\sqrt{O_p}$ has side lengths $<\pi$ (and hence belongs to $\ZOpp$) for each $s\in[-R,R]$. Therefore, by the previous special case, \eqref{eq77} holds for $g=\alpha^X_s$ whenever $s\in[-R,R]$, and for each $A\in\Bcl(\sqrt{O_p})$. By the additivity of $\Bcl$ (Cor. \ref{lb75}), \eqref{eq77} also holds for all $A\in\Bcl(\sqrt O)$.
\end{proof}

\subsection{Extending the action $\Bcl\big\vert_{\ZOpp}\curvearrowright\Hopij$ to $\Bcl\big\vert_{\Ozp}\curvearrowright\Hopij$}\label{lb154}

Fix $\MH_i,\MH_j\in\RepA$.

\begin{thm}\label{lb149}
The collection $\varpi_{\opij}=(\varpi_{\opij,O})_{O\in\ZOpp}$ defined in Def. \ref{lb120} can be extended uniquely to a collection of normal representations
\begin{align*}
\varpi_{\opij,O}:\Bcl(O)\rightarrow\fk L(\MH_{\opij})\qquad\text{for each }O\in\Ozp
\end{align*}
abbreviated to $\pmb{\varpi_O}$ when no confusion arises, such that the following property holds:
\begin{enumerate}
\item[(1)] (Isotony) If $O_1,O_2\in\Ozp$ and $O_1\subset O_2$, then $\varpi_{O_2}\big|_{\Bcl(O_1)}=\varpi_{O_1}$.
\end{enumerate}
Moreover, the following properties hold for each $O,O_1,O_2\in\Ozp$.
\begin{enumerate}[label=(\arabic*),start=2]
\item (\textbf{Bulk-bulk locality}) Suppose that $O_1,O_2$ are spacelike separated. Then
\begin{align*}
\big[\varpi_{O_1}(\Bcl(O_1)),\varpi_{O_2}(\Bcl(O_2))\big]=0
\end{align*}
\item (Extension property) Let $O$ correspond to $(\wtd I,\wtd J)$. For each $x\in\MA(I)$ and $y\in\MA(J)$, noting that
\begin{align*}
\pi^+_{\sqz\boxtimes\ovl\sqz,I}(x)\in\Bcl(O)\qquad \pi^-_{\sqz\boxtimes\ovl\sqz,J}(y)\in\Bcl(O)
\end{align*}
by the extension property in Thm. \ref{lb72}, we have
\begin{gather*}
\varpi_{O}\big(\pi^+_{\sqz\boxtimes\ovl\sqz,I}(x)\big)=\pi_{i\boxtimes \ovl j,I}(x)\\
\varpi_{O}\big(\pi^-_{\sqz\boxtimes\ovl\sqz,J}(y)\big)=\pi_{i\boxtimes \ovl j,J}(y)
\end{gather*}
\item (Conformal covariance) Suppose that $g\in\Gc$ and $O\in\Ozp$. Then for each $A\in\Bcl(O)$, noting that $\Ucl(g,g)A\Ucl(g,g)^*$ belongs to $\Bcl((g,g)O)$ by the conformal covariance in Thm. \ref{lb72}, we have
\begin{align*}
U_{i\boxtimes\ovl j}(g)\varpi_{O}(A)U_{i\boxtimes\ovl j}(g)^*=\varpi_{(g,g)O}(\Ucl(g,g)A\Ucl(g,g)^*)
\end{align*}
\item (\textbf{Bulk-boundary Haag duality}) Let $\wtd E_+$ and $\wtd E_-$ be the intersections of the spacelike complement of $O$ with $\partial_+\Rop$ and $\partial_-\Rop$ respectively, i.e., for each $\pm\in\{+,-\}$,
\begin{gather*}
\wtd E_\pm=\Int\{p\in\partial_\pm\Rop:p\text{ is spacelike separated from elements of }O\}
\end{gather*}
where the interior is with respect to $\partial_\pm\Rop$. Then $\wtd E_+\in\Jtd_+$ and $\wtd E_-\in\Jtd_-$. Moreover, by viewing $\wtd E_+,\wtd E_-\in\Jtd$ via the equivalence \eqref{eq49}, we have
\begin{align*}
\varpi^L_{\wtd E_+}\big(\Bopi(\wtd E_+)\big)\vee \varpi^R_{\wtd E_-}\big(\Bopj'(\wtd E_-)\big)=\varpi_{O}\big(\Bcl(O)\big)'
\end{align*}
\end{enumerate}
\end{thm}

The bulk-boundary Haag duality implies the \textbf{bulk-boundary locality}: If $\wtd E_+\in\Jtd_+$ and $\wtd E_-\in\Jtd_-$ are spacelike separated from $O$, then both $\varpi^L_{\wtd E_+}\big(\Bopi(\wtd E_+)\big)$ and $\varpi^R_{\wtd E_-}\big(\Bopj'(\wtd E_-)\big)$ commute with $\varpi_O\big(\Bcl(O)\big)$.

\begin{proof}
The uniqueness follows directly from the isotony, and from the additivity of $\Bcl$ (Cor. \ref{lb75}). To prove the existence, for each $O\in\Ozp$, choose any $g\in\Gc$ such that $(g,g)O\in\ZOpp$, and define $\varpi_O$ by
\begin{align*}
\varpi_{O}(A)=U_{i\boxtimes\ovl j}(g)^*\varpi_{(g,g)O}(\Ucl(g,g)A\Ucl(g,g)^*)U_{i\boxtimes\ovl j}(g)
\end{align*}
for each $A\in\Bcl(O)$. The conformal covariance in Prop. \ref{lb124} implies that this is well-defined (i.e. independent of the choice of $g$), and that the conformal covariance in the present theorem holds true. Moreover, if $O\in\ZOpp$, this definition agrees with the one in Def. \ref{lb120} since one can choose $g=1$ in this case.

The isotony in this theorem follows from that in Prop. \ref{lb124} and our construction of $\varpi_O(A)$ in the previous paragraph. The extension property in this theorem follows from that in Prop. \ref{lb124}, together with the additivity of $\MA$. 

Bulk-bulk locality: By the additivity of $\Bcl$ (Cor. \ref{lb75}), it suffices to assume that the closures of $O_1$ and $O_2$ are spacelike separated. Let $(\wtd I_1,\wtd J_1)$ correspond to $O_1$ and $(\wtd I_2,\wtd J_2)$ correspond to $O_2$. By Rem. \ref{lb76} (or Fig. \ref{fig1}) and a possible exchange of the subscripts $1$ and $2$, the following arg-valued intervals are in clockwise order:
\begin{align*}
\wtd I_2,\wtd I_1,\wtd J_1,\wtd J_2
\end{align*}
and $I_2,I_1,J_1,J_2$ have mutually disjoint closures. By the conformal covariance in this theorem, after applying some $g\in\Gc$ to these four arg-valued intervals, we assume that $O_1=\sqrt{D_1}$ and $O_2=\sqrt{D_2}$ where $D_1,D_2\in\Opp$ have spacelike separated closures. Thus, the bulk-bulk locality in this theorem follows from that in Prop. \ref{lb124}.

Bulk-boundary Haag duality: Let $(\wtd I,\wtd J)$ correspond to $O$ such that $\wtd J$ is clockwise to $\wtd I$ (cf. Rem. \ref{lb79}). Choose $\wtd K_+,\wtd K_-\in\Jtd$ such that $K_+,K_-$ are the two components of the interior of $\Sbb^1\setminus(I\cup J)$, and that
\begin{align*}
\wtd I,\wtd K_+, \wtd J,\wtd K_- \qquad\text{are in clockwise order}
\end{align*}
Similar to \eqref{eq79}, one has
\begin{align}
\wtd E_+=\wtd K_+\qquad \wtd E_-=\wtd K_-
\end{align}
This proves the first part of the Haag duality. The second part follows from the Haag duality in Prop. \ref{lb124}, together with the conformal covariance in this theorem and the conformal covariance in Thm. \ref{lb119}.
\end{proof}

\begin{figure}[h]
	\centering
\begin{gather*}
\includegraphics[height=3.5cm]{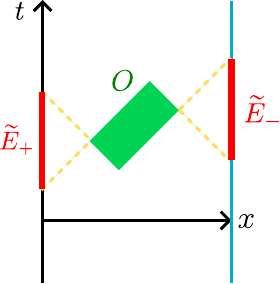} \qquad  \qquad  \qquad \includegraphics[height=3.5cm]{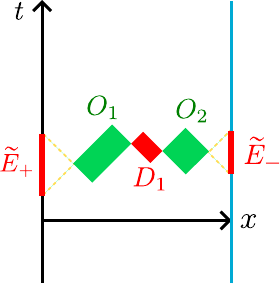}
   \end{gather*}
\caption{.~~Bulk-boundary Haag duality in $\Rop$}\label{fig5}
\end{figure}

\begin{thm}[\textbf{Bulk-boundary Haag duality}]\label{lb153}
Let $n\geq 1$. Let $O_1,\dots,O_n\in\Ozp$ whose closures are mutually spacelike separated, and which are contained in a common double cone $O\in\Ozp$. Let $\Sigma$ be the spacelike complement of ${O_1},\dots,{O_n}$ in $\Rop$. Let $D_1,\dots,D_{n-1}\in\Ozp$ such that ${D_1},\dots,{D_{n-1}}$ are the connected components of $\Sigma$ that are contained in $\Rzp$. Let
\begin{align*}
\wtd E_+=\Sigma\cap \partial_+\Rop\qquad \wtd E_-=\Sigma\cap \partial_-\Rop
\end{align*}
Then $\wtd E_+\in\Jtd_+$ and $\wtd E_-\in\Jtd_-$, and the following two von Neumann algebras on $\Hopij$ are commutants of each other:
\begin{gather*}
\bigvee_{j=1}^n\varpi_{{O_j}}\big(\Bcl({O_j})\big)\\
\varpi^L_{\wtd E_+}\big(\Bopi(\wtd E_+)\big)\vee \varpi^R_{\wtd E_-}\big(\Bopj'(\wtd E_-)\big)\vee\bigvee_{j=1}^{n-1}\varpi_{{D_j}}\big(\Bcl({D_j})\big)
\end{gather*}
\end{thm}


\begin{proof}
The case $n=1$ follows from Thm. \ref{lb149}. Now assume $n>1$. We assume that $O_1,\dots,O_n$ are listed from left to right. By shrinking $O$, we assume that the left corners of $O_1$ and $O$ are the same, and the right corners of $O_n$ and $O$ are the same. Then $\wtd E_\pm$ is the intersection of the spacelike complement of $ O$ (in $\Rop$) with $\partial_\pm\Rop$, which belongs to $\Jtd_\pm$ by Thm. \ref{lb149}. Moreover, by Thm. \ref{lb149}, the commutant of $\varpi_{ O}\big(\Bcl( O)\big)$ is generated by $\varpi^L_{\wtd E_+}\big(\Bopi(\wtd E_+)\big)$ and $\varpi^R_{\wtd E_-}\big(\Bopj'(\wtd E_-)\big)$. Therefore, it suffices to show that
\begin{align*}
\varpi_{ O}\big(\Bcl( O)\big)\cap\Big(\bigvee_{j=1}^{n-1}\varpi_{ O}\big(\Bcl({D_j})\big)\Big)'=\bigvee_{j=1}^n\varpi_{ O}\big(\Bcl({O_j})\big)
\end{align*}
This follows from the faithfulness of the normal representation $\varpi_{ O}$, together with the fact that the relative commutant of $\bigvee_{j=1}^{n-1}\Bcl({D_j})$ in $\Bcl( O)$ is equal to $\bigvee_{j=1}^n\Bcl({O_j})$ (as implied by Thm. \ref{lb115} or by its proof).
\end{proof}

\subsection{Nets of boundary algebroids changing the boundary conditions}\label{lb156}

For each $\MH_i\in\RepA$, in Sec. \ref{lb131} and \ref{lb130}, we have considered the nets $\Bopi,\Bopi'$ of (left and right) boundary operators preserving the boundary condition $i$. From Def. \ref{lb112}, for each $\wtd I\in\Jtd$ we have isomorphisms
\begin{align*}
\Bopi(\wtd I)\simeq\End_{\MA(I')}(\MH_i)\qquad \Bopi'(\wtd I)\simeq\End_{\MA(I')}(\MH_{\ovl i})
\end{align*}

More generally, for each $\MH_i,\MH_k\in\RepA$, we view
\begin{align*}
\wtd I\in\Jtd\mapsto\pmb{\Bop{i\rightarrow k}(\wtd I)}=\Hom_{\MA(I')}(\MH_i,\MH_k)
\end{align*}
as the net of von Neumann algebroids of left boundary operators changing the boundary condition from $i$ to $k$. Similarly, for each $\MH_j,\MH_l\in\RepA$, we view
\begin{align*}
\wtd I\in\Jtd\mapsto\pmb{\Bop{j\rightarrow l}'(\wtd I)}=\Hom_{\MA(I')}(\MH_{\ovl j},\MH_{\ovl l})
\end{align*}
as the net of von Neumann algebroids of right boundary operators changing the boundary condition from $j$ to $l$.

The nets $\Bop{i\rightarrow k}$ and $\Bop{j\rightarrow l}'$ act on boundary state spaces as follows: For each $\MH_b\in\RepA$, define
\begin{gather*}
\varpi_{\wtd I}^L:\Bop{i\rightarrow k}(\wtd I)\rightarrow \fk L(\Hop{i,b},\Hop{k,b})\\
A\mapsto R(\psi,\wtd I')AR(\psi,\wtd I')^*|_{\MH_i\boxtimes\MH_{\ovl b}}
\end{gather*}
where $\psi\in\MH_{\ovl b}(I')$ is chosen to be $I'$-unitary. This definition is independent of the choice of $\psi$. Indeed, if we choose $I$-unitary vectors $\xi_i\in\MH_i(I)$ and $\xi_k\in\MH_k(I)$, then by Prop. \ref{lb34},
\begin{align*}
A=L(\xi_k,\wtd I)xL(\xi_i,\wtd I)^*|_{\MH_i}
\end{align*}
for some $x\in\MA(I)$. Hence, by the locality in Thm. \ref{lb15},
\begin{align*}
\varpi_{\wtd I}^L\big(L(\xi_k,\wtd I)xL(\xi_i,\wtd I)^*|_{\MH_i}\big)\Big|_{\MH_i\boxtimes\MH_{\ovl b}}=L(\xi_k,\wtd I)xL(\xi_i,\wtd I)^*\Big|_{\MH_i\boxtimes\MH_{\ovl b}}
\end{align*}

Similarly, we define
\begin{gather*}
\varpi_{\wtd I}^R:\Bop{j\rightarrow l}'(\wtd I)\rightarrow \fk L(\Hop{b,j},\Hop{b,l})\\
B\mapsto L(\phi,\bpr\wtd I)BL(\phi,\bpr\wtd I)^*|_{\MH_b\boxtimes\MH_{\ovl j}}
\end{gather*}
where $\phi\in\MH_b(I')$ is $I'$-unitary. For each $x\in\MA(I)$ and each $I$-unitary $\eta_{\ovl j}\in\MH_{\ovl j}(I),\eta_{\ovl l}\in\MH_{\ovl l}(I)$, one checks
\begin{align*}
\varpi_{\wtd I}^R\big(R(\eta_{\ovl l},\wtd I)xR(\eta_{\ovl j},\wtd I)^*|_{\MH_{\ovl j}}\big)\Big|_{\MH_b\boxtimes\MH_{\ovl j}}=R(\eta_{\ovl l},\wtd I)xR(\eta_{\ovl j},\wtd I)^*\Big|_{\MH_b\boxtimes\MH_{\ovl j}}
\end{align*}

The following trick allows us to generalize directly many properties in the previous sections of this chapter to $\Bop{i\rightarrow k}$ and $\Bop{j\rightarrow l}'$: By \eqref{eq67}, the fusion product of two isometric morphisms of $\MA$-modules is isometric. Therefore, we may view $\MH_i\boxtimes\MH_{\ovl j}$ and $\MH_k\boxtimes\MH_{\ovl l}$ as $\MA$-submodules of $\MH_{i\oplus k}\boxtimes\MH_{\ovl j\oplus \ovl l}$. By viewing
\begin{align*}
\Bop{i\rightarrow k}(\wtd I)\subset \Bop{i\oplus k}(\wtd I)\qquad \Bop{j\rightarrow l}'(\wtd I)\subset \Bop{j\oplus l}'(\wtd I)
\end{align*}
the representation $\varpi^L_{\wtd I}$ of $\Bop{i\oplus k}(\wtd I)$ restricts to that of $\Bop{i\rightarrow k}(\wtd I)$, and the representation $\varpi^R_{\wtd I}$ of $\Bop{j\oplus l}'(\wtd I)$ restricts to that of $\Bop{j\rightarrow l}'(\wtd I)$. Using this observation, one can for instance generalize the bulk-boundary locality in Thm. \ref{lb149} to $\Bop{i\rightarrow k}$, $\Bop{j\rightarrow l}'$, and the bulk net $\Bcl$.

\begin{appendices}

\section{Consequences of Eq. \eqref{eq68} and conditions (1)--(4) of Thm. \ref{lb15}}\label{lb140}

Let $\MA$ be a conformal net, and let $G\rightarrow\Aut(\MA)$ be a group homomorphism. Throughout this chapter, we do not assume Convention \ref{lb10}, except in Sec. \ref{lb139}.

The goal of this chapter is to show that Eq. \eqref{eq68} together with properties (1)--(4) of Thm. \ref{lb15} already implies the remaining properties in that theorem, as well as the compatibility conditions in Rem. \ref{lb7} and the balancing condition. This provides a more algebraic approach to these results, in contrast to the more geometric proofs in \cite{MS26a,MS26b}, which rely on path continuation.

Neither the remaining properties of Thm. \ref{lb15}, nor the compatibility conditions in Rem. \ref{lb7}, nor the balancing condition is needed for the construction of open/closed CFTs carried out in this paper. The only exception is the conformal covariance property in Thm. \ref{lb15}, which is used in Ch. \ref{lb150} to verify the conformal covariance of the open CFT. Readers interested only in the construction of open/closed CFTs may therefore safely skip this chapter.

\subsection{$G$-covariance}\label{lb24}

We begin the proof of the $G$-covariance by proving the following special case:

\begin{lm}\label{lb23}
Let $\MH_i\in\RepGA$ and $\phi\in G$. Then for each $\wtd I\in\Jtd$ and $\xi\in\MH_i(I)$, the following identities hold as bounded linear operators $\MH_0\rightarrow\MH_{\phi i}$.
\begin{align}\label{eq13}
\Gamma_\phi L(\xi,\wtd I)\big|_{\MH_0}=L(\Gamma_\phi\xi,\wtd I)\phi\big|_{\MH_0}\qquad \Gamma_\phi R(\xi,\wtd I)\big|_{\MH_0}=R(\Gamma_\phi\xi,\wtd I)\phi\big|_{\MH_0}
\end{align}
\end{lm}

\begin{proof}
By Rem. \ref{lb13}, the operator $\Gamma_\phi L(\xi,\wtd I)\phi^{-1}|_{\MH_0}$ belongs to $\Hom_{\MA(\wtd I')}(\MH_0,\MH_{\phi i})$, and hence is the unique element in that set sending $\Omega$ to
\begin{align*}
\Gamma_\phi L(\xi,\wtd I)\phi^{-1}\Omega=\Gamma_\phi L(\xi,\wtd I)\Omega=\Gamma_\phi\xi
\end{align*}
So it equals $L(\Gamma_\phi\xi,\wtd I)|_{\MH_0}$. The second identity can be proved in the same way.
\end{proof}

\begin{rem}\label{lb27}
In Lem. \ref{lb23}, the operators $\Gamma_\phi L(\xi,\wtd I)\big|_{\MH_0}$ and $L(\Gamma_\phi\xi,\wtd I)\phi\big|_{\MH_0}$ originally have codomains $\MH_{\phi(i\boxtimes 0)}$ and $\MH_{\phi i}\boxtimes\MH_0$, respectively. To ensure that the first identity in \eqref{eq13} holds, we identify $\MH_{\phi(i\boxtimes 0)}\simeq\MH_{\phi i}\simeq\MH_{\phi i\boxtimes0}$ using the unitors in \eqref{eq12d}. Similarly, to obtain the second identity in \eqref{eq13}, we identify $\MH_{\phi(0\boxtimes i)}\simeq\MH_{\phi i}\simeq\MH_{0\boxtimes\phi i}$ using the unitors in \eqref{eq12c}. This observation will be needed in the proofs of \eqref{eq12c} and \eqref{eq12d} in Subsec. \ref{lb25}.
\end{rem}

\begin{thm}\label{lb28}
For each $\MH_i,\MH_j\in\RepGA$ and $\phi\in G$, there is a (necessarily unique) unitary isomorphism of $G$-twisted $\MA$-modules
\begin{align*}
\mbf D_{\phi,i,j}:\MH_{\phi i}\boxtimes\MH_{\phi j}\rightarrow \MH_{\phi(i\boxtimes j)}
\end{align*}
satisfying Def. \ref{lb16}.
\end{thm}

\begin{proof}
Step 1. It suffices to construct $\mbf D_{\phi,i,j}$ satisfying the first relation of \eqref{eq11}, since the second one will follow immediately from Prop. \ref{lb20}.  

For each $\xi_1,\xi_2\in\MH_i(I)$ and $\eta_1,\eta_2\in\MH_j$, by Prop. \ref{lb21},
\begin{align*}
&\bk{L(\Gamma_\phi\xi_2,\wtd I)\Gamma_\phi\eta_2|L(\Gamma_\phi\xi_1,\wtd I)\Gamma_\phi\eta_1}=\bk{\Gamma_\phi\eta_2|L(\Gamma_\phi\xi_2,\wtd I)^*L(\Gamma_\phi\xi_1,\wtd I)\Gamma_\phi\eta_1}\\
=&\bk{\Gamma_\phi\eta_2|L(L(\Gamma_\phi\xi_2,\wtd I)^*\Gamma_\phi\xi_1,\wtd I)\Gamma_\phi\eta_1}
\end{align*}
Set $x=L(\Gamma_\phi\xi_2,\wtd I)^*L(\Gamma_\phi\xi_1,\wtd I)|_{\MH_0}$, which belongs to $\MA(I)$. Then the above expression equals
\begin{align*}
\bk{\eta_2|\Gamma_\phi^{-1}L(x\Omega,\wtd I)\Gamma_\phi\eta_1}=\bk{\eta_2|\Gamma_\phi^{-1}\pi_{\wtd I}(x)\Gamma_\phi\eta_1}=\bk{\eta_2|\pi_{\wtd I}(\phi^{-1}x\phi)\eta_1}
\end{align*}
where Prop. \ref{lb22} and \eqref{eq4} are used. By Lem. \ref{lb23},
\begin{align*}
\phi^{-1}x\phi=(L(\Gamma_\phi\xi_2,\wtd I)\phi|_{\MH_0})^*(L(\Gamma_\phi\xi_1,\wtd I)\phi|_{\MH_0})=L(\xi_2,\wtd I)^*L(\xi_1,\wtd I)|_{\MH_0}
\end{align*}
Thus
\begin{align*}
\bk{L(\Gamma_\phi\xi_2,\wtd I)\Gamma_\phi\eta_2|L(\Gamma_\phi\xi_1,\wtd I)\Gamma_\phi\eta_1}=\bk{\eta_2|\pi_{\wtd I}(L(\xi_2,\wtd I)^*L(\xi_1,\wtd I)|_{\MH_0})\eta_1}
\end{align*}
A similar (and indeed simpler) computation shows
\begin{align*}
\bk{L(\xi_2,\wtd I)\eta_2|L(\xi_1,\wtd I)\eta_1}=\bk{\eta_2|\pi_{\wtd I}(L(\xi_2,\wtd I)^*L(\xi_1,\wtd I)|_{\MH_0})\eta_1}
\end{align*}
Thus, by the density of fusion products in Thm. \ref{lb15}, there is a unitary map
\begin{gather*}
\mbf D^{\wtd I}_{\phi,i,j}:\MH_{\phi i}\boxtimes\MH_{\phi j}\rightarrow \MH_{\phi(i\boxtimes j)}\qquad L(\Gamma_\phi\xi,\wtd I)\Gamma_\phi\eta\mapsto \Gamma_\phi L(\xi,\wtd I)\eta
\end{gather*}
for each $\xi\in\MH_i(I),\eta\in\MH_j$.\\[-1ex]

Step 2. If $\wtd I_1,\wtd I_2\in\Jtd$ are close in the sense that there exists $\wtd I_0\in\Jtd$ satisfying $\wtd I_0\subset\wtd I_1$ and $\wtd I_0\subset\wtd I_2$, then $\mbf D^{\wtd I_1}_{\phi,i,j}$ and $\mbf D^{\wtd I_2}_{\phi,i,j}$ agree on vectors of the form $L(\MH_{\phi i}(I_0),\wtd I_0)\MH_{\phi j}$, and hence agree on $\MH_{\phi i}\boxtimes\MH_{\phi j}$ by the density of fusion products. If $\wtd I_1$ and $\wtd I_2$ are not close, one can connect them by a chain of arg-valued intervals such that each interval in the chain is close to the preceding one. It follows that $\mbf D^{\wtd I_1}_{\phi,i,j}$ and $\mbf D^{\wtd I_2}_{\phi,i,j}$ also agree in this general case. We have thus proved that $\mbf D^{\wtd I}_{\phi,i,j}$ is independent of the choice of $\wtd I$, and hence can be denoted by $\mbf D_{\phi,i,j}$.\\[-1ex]

Step 3. For each $\wtd J\in\Jtd$ and $y\in\MA(J)$, and for each $\wtd I$ anticlockwise to $\wtd J$ and each $\xi\in\MH_i(I),\eta\in\MH_j$, we compute using \eqref{eq4} that
\begin{align*}
&\mbf D_{\phi,i,j}\pi_{\wtd J}(y)L(\Gamma_\phi\xi,\wtd I)\Gamma_\phi\eta=\mbf D_{\phi,i,j}L(\Gamma_\phi\xi,\wtd I)\pi_{\wtd J}(y)\Gamma_\phi\eta\\
=&\mbf D_{\phi,i,j}L(\Gamma_\phi\xi,\wtd I)\Gamma_\phi\pi_{\wtd J}(\phi^{-1}y\phi)\eta=\Gamma_\phi L(\xi,\wtd I)\pi_{\wtd J}(\phi^{-1}y\phi)\eta\\
=&\Gamma_\phi \pi_{\wtd J}(\phi^{-1}y\phi)L(\xi,\wtd I)\eta=\pi_{\wtd J}(y)\Gamma_\phi L(\xi,\wtd I)\eta=\pi_{\wtd J}(y)\mbf D_{\phi,i,j}L(\Gamma_\phi\xi,\wtd I)\Gamma_\phi\eta
\end{align*}
This proves that $\mbf D_{\phi,i,j}$ is a morphism of $G$-twisted $\MA$-modules.
\end{proof}

\subsection{The compatibility conditions in Rem. \ref{lb7}}\label{lb25}

\subsubsection{Equations \eqref{eq12}}

\begin{proof}[\textbf{Proof of \eqref{eq12a}}]
Choose any $\wtd I\in\Jtd,\xi\in\MH_i(I),\eta\in\MH_j$, we have
\begin{align*}
&(\fk T_\phi(S)\boxtimes\fk T_\phi(T))L(\Gamma_\phi\xi,\wtd I)\Gamma_\phi\eta=(\Gamma_\phi S\Gamma_\phi^{-1}\boxtimes\Gamma_\phi T\Gamma_\phi^{-1})L(\Gamma_\phi\xi,\wtd I)\Gamma_\phi\eta\\
=&L(\Gamma_\phi S\xi,\wtd I)\Gamma_\phi T\eta
\end{align*}
where Prop. \ref{lb26} is used. Thus, by Def. \ref{lb16},
\begin{align*}
\mbf D_{\phi,\wtd i,\wtd j}(\fk T_\phi(S)\boxtimes\fk T_\phi(T))L(\Gamma_\phi\xi,\wtd I)\Gamma_\phi\eta=\Gamma_\phi L(S\xi,\wtd I)T\eta
\end{align*}
Def. \ref{lb16} and Prop. \ref{lb26} also imply
\begin{align*}
&\fk T_\phi(S\boxtimes T)\mbf D_{\phi,i,j}L(\Gamma_\phi\xi,\wtd I)\Gamma_\phi\eta=\Gamma_\phi (S\boxtimes T)\Gamma_\phi^{-1}\cdot\Gamma_\phi L(\xi,\wtd I)\eta\\
=&\Gamma_\phi L(S\xi,\wtd I)T\eta
\end{align*}
This proves \eqref{eq12a}, thanks to the density of fusion products.
\end{proof}

\begin{proof}[\textbf{Proof of \eqref{eq12b}}]
Choose any $\wtd I,\wtd K\in\Jtd$ with $\wtd K$ clockwise to $\wtd I$. Then for each $\xi\in\MH_i(I),\eta\in\MH_j,\chi\in\MH_k(K)$,
\begin{align*}
&\mbf D_{\phi,i,j\boxtimes k}(\idt\boxtimes\mbf D_{\phi,j,k})L(\Gamma_\phi\xi,\wtd I)R(\Gamma_\phi\chi,\wtd K)\Gamma_\phi\eta\\
=&\mbf D_{\phi,i,j\boxtimes k}L(\Gamma_\phi\xi,\wtd I)\mbf D_{\phi,j,k}R(\Gamma_\phi\chi,\wtd K)\Gamma_\phi\eta\\
=&\mbf D_{\phi,i,j\boxtimes k}L(\Gamma_\phi\xi,\wtd I)\Gamma_\phi R(\chi,\wtd K)\eta=\Gamma_\phi L(\xi,\wtd I)R(\chi,\wtd K)\eta
\end{align*}
where the naturality in Thm. \ref{lb15} is used. Similarly,
\begin{align*}
&\mbf D_{\phi,i\boxtimes j,k}(\mbf D_{\phi,i,j}\boxtimes \idt)R(\Gamma_\phi\chi,\wtd K)L(\Gamma_\phi\xi,\wtd I)\Gamma_\phi\eta\\
=&\mbf D_{\phi,i\boxtimes j,k}R(\Gamma_\phi\chi,\wtd K)\mbf D_{\phi,i,j}L(\Gamma_\phi\xi,\wtd I)\Gamma_\phi\eta\\
=&\mbf D_{\phi,i\boxtimes j,k}R(\Gamma_\phi\chi,\wtd K)\Gamma_\phi L(\xi,\wtd I)\eta=\Gamma_\phi R(\chi,\wtd K)L(\xi,\wtd I)\eta
\end{align*}
This finishes the proof, thanks to the locality and the density of fusion products in Thm. \ref{lb15}.
\end{proof}

\begin{proof}[\textbf{Proof of \eqref{eq12c}}]
For each $\wtd I\in\Jtd,\xi\in\MH_i(I),\chi\in\MH_0$, we compute
\begin{align*}
(\mbf B_\phi\boxtimes\idt)R(\Gamma_\phi\xi,\wtd I)\phi\chi=R(\Gamma_\phi\xi,\wtd I)\mbf B_\phi\phi\chi=R(\Gamma_\phi\xi,\wtd I)\Gamma_\phi\chi
\end{align*}
where the naturality in Thm. \ref{lb15} and the definition of $\mbf B_\phi$ are used. By Def. \ref{lb16},
\begin{align*}
\mbf D_{\phi,0,i}(\mbf B_\phi\boxtimes\idt)R(\Gamma_\phi\xi,\wtd I)\phi\chi=\mbf D_{\phi,0,i}R(\Gamma_\phi\xi,\wtd I)\Gamma_\phi\chi=\Gamma_\phi R(\xi,\wtd I)\chi
\end{align*}
Identify $\MH_0\boxtimes\MH_{\phi i}\simeq\MH_{\phi i}\simeq\MH_{\phi(0\boxtimes i)}$ using the unitors. Then, by Lem. \ref{lb23} and Rem. \ref{lb27}, we have
\begin{align*}
\mbf D_{\phi,0,i}(\mbf B_\phi\boxtimes\idt)R(\Gamma_\phi\xi,\wtd I)\phi\chi=R(\Gamma_\phi\xi,\wtd I)\phi\chi
\end{align*}
By the density of fusion products, $\mbf D_{\phi,0,i}(\mbf B_\phi\boxtimes\idt)$ equals the identity.
\end{proof}

\begin{proof}[\textbf{Proof of \eqref{eq12d}}]
This is similar to the proof of \eqref{eq12c}, and hence is omitted.
\end{proof}

\subsubsection{Equations \eqref{eq73}}

\begin{proof}[\textbf{Proof of \eqref{eq73a}}]
Choose any $\wtd I\in\Jtd$ and $\xi\in\MH_i(I),\eta\in\MH_j$. 
By the definition of $\mbf A$ (in Rem. \ref{lb7}) and Prop. \ref{lb26},
\begin{align*}
(\mbf A_i\boxtimes\mbf A_j)L(\xi,\wtd I)\eta=L(\mbf A_i\xi,\wtd I)\mbf A_j\eta=L(\Gamma_e\xi,\wtd I)\Gamma_e\eta
\end{align*}
By Def. \ref{lb16},
\begin{align*}
\mbf D_{e,i,j}(\mbf A_i\boxtimes\mbf A_j)L(\xi,\wtd I)\eta=\Gamma_eL(\xi,\wtd I)\eta
\end{align*}
Since $L(\xi,\wtd I)\eta\in\MH_{i\boxtimes j}$, by the definition of $\mbf A$,
\begin{align*}
\mbf A_{i\boxtimes j}L(\xi,\wtd I)\eta=\Gamma_eL(\xi,\wtd I)\eta
\end{align*}
The proof is thus finished by the density of fusion products in Thm. \ref{lb15}.
\end{proof}

\begin{proof}[\textbf{Proof of \eqref{eq73b}}]
By the definitions of $\mbf A$ and $\mbf B$ in Rem. \ref{lb7},
\begin{align*}
\mbf B_e=\Gamma_e\circ e^{-1}|_{\MH_0}=\Gamma_e|_{\MH_0}=\mbf A_0
\end{align*}
This finishes the proof.
\end{proof}

\subsubsection{Equations \eqref{eq72}}

\begin{proof}[\textbf{Proof of \eqref{eq72a}}]
Choose any $\wtd I\in\Jtd$ and $\xi\in\MH_i(I),\eta\in\MH_j$. Then the vector
\begin{align*}
\chi:=L(\Gamma_\phi\Gamma_\psi\xi,\wtd I)\Gamma_\phi\Gamma_\psi\eta
\end{align*}
belongs to $\MH_{\phi(\psi i)}\boxtimes\MH_{\phi(\psi j)}$, the top left corner of \eqref{eq72a}. By the definition of $\mbf C$ in Rem. \ref{lb7}-(c), we have
\begin{align*}
\mbf C_{\phi,\psi,i}\Gamma_\phi\Gamma_\psi\xi=\Gamma_{\phi\psi}\xi\qquad \mbf C_{\phi,\psi,j}\Gamma_\phi\Gamma_\psi\eta=\Gamma_{\phi\psi}\eta
\end{align*}
Therefore, by Prop. \ref{lb26}, we have
\begin{align*}
(\mbf C_{\phi,\psi,i}\boxtimes\mbf C_{\phi,\psi,j})\chi=L(\Gamma_{\phi\psi}\xi,\wtd I)\Gamma_{\phi\psi}\eta
\end{align*}
By Def. \ref{lb16},
\begin{align*}
\mbf D_{\phi\psi,i,j}(\mbf C_{\phi,\psi,i}\boxtimes\mbf C_{\phi,\psi,j})\chi=\Gamma_{\phi\psi}L(\xi,\wtd I)\eta
\end{align*}

On the other hand, by Def. \ref{lb16},
\begin{align*}
\mbf D_{\phi,\psi i,\psi j}\chi=\Gamma_\phi L(\Gamma_\psi\xi,\wtd I)\Gamma_\psi\eta
\end{align*}
Therefore, since $\fk T_\phi(\mbf D_{\psi,i,j})=\Gamma_\phi\circ\mbf D_{\psi,i,j}\circ\Gamma_\phi^{-1}$, Def. \ref{lb16} implies
\begin{align*}
\fk T_\phi(\mbf D_{\psi,i,j})\mbf D_{\phi,\psi i,\psi j}\chi=\Gamma_\phi\circ\Gamma_\psi L(\xi,\wtd I)\eta
\end{align*}
By the definition of $\mbf C$,
\begin{align*}
\mbf C_{\phi,\psi,i\boxtimes j}\fk T_\phi(\mbf D_{\psi,i,j})\mbf D_{\phi,\psi i,\psi j}\chi=\Gamma_{\phi\psi}L(\xi,\wtd I)\eta
\end{align*}
This finishes the proof, thanks to the density of fusion products in Thm. \ref{lb15}.
\end{proof}

\begin{proof}[\textbf{Proof of \eqref{eq72b}}]
Choose any $\chi\in\MH_0$. By the definition of $\mbf B$ in Rem. \ref{lb7}, we have $\mbf B_\phi \cdot\phi\psi\chi=\Gamma_\phi\cdot \psi\chi$, and hence
\begin{align*}
\fk T_\phi(\mbf B_\psi)\cdot\mbf B_\phi\cdot \phi\psi\chi=\fk T_\phi(\mbf B_\psi)\cdot\Gamma_\phi\cdot\psi\chi=\Gamma_\phi\mbf B_\psi\cdot\psi \chi=\Gamma_\phi\Gamma_\psi\chi
\end{align*}
By the definitions of $\mbf C$ and $\mbf B$,
\begin{align*}
\mbf C_{\phi,\psi,0}\cdot\fk T_\phi(\mbf B_\psi)\cdot\mbf B_\phi\cdot \phi\psi\chi=\mbf C_{\phi,\psi,0}\cdot\Gamma_\phi\Gamma_\psi\chi=\Gamma_{\phi\psi}\chi=\mbf B_{\phi\psi}\cdot\phi\psi\chi
\end{align*}
This finishes the proof.
\end{proof}


\subsection{Braiding}

The braiding in $\RepGA$ is defined to be the unique operation satisfying (5) of Thm. \ref{lb15}. In this subsection, we explain how (5) follows directly from \eqref{eq68} and (1)--(4) of Thm. \ref{lb15}.

\begin{lm}\label{lb18}
Suppose that $\MH_i\in\Rep^\phi(\MA)$. Then for each $\wtd I\in\Jtd$ and $\xi\in\MH_i(I)$, we have
\begin{align*}
L(\xi,\wtd I)\big|_{\MH_0}=R(\xi,\wtd I)\phi\big|_{\MH_0}
\end{align*}
\end{lm}

\begin{proof}
This is a direct reformulation of Lem. \ref{lb17}.
\end{proof}

\begin{thm}\label{lb29}
For each $\phi\in G,\MH_i\in\Rep^\phi(\MA),\MH_j\in\RepGA$, there is a unique unitary map (which is also a $G$-twisted $\MA$-module isomorphism) $\mbb B_{i,j}:\MH_i\boxtimes\MH_j\rightarrow\MH_{\phi j}\boxtimes\MH_i$ such that
\begin{align*}
\mbb B_{i,j}L(\xi,\wtd I)\eta=R(\xi,\wtd I)\Gamma_\phi\eta\qquad\text{for each }\xi\in\MH_i(I),\eta\in\MH_j
\end{align*}
\end{thm}

\begin{proof}
The uniqueness is clear from the density of fusion products in Thm. \ref{lb15}. By the computation in the proof of Thm. \ref{lb28}, for each $\xi_1,\xi_2\in\MH_i(I)$ and $\eta_1,\eta_2\in\MH_j$, we have
\begin{subequations}\label{eq22}
\begin{gather}
\bk{L(\xi_2,\wtd I)\eta_2|L(\xi_1,\wtd I)\eta_1}=\bk{\eta_2|\pi_{\wtd I}(L(\xi_2,\wtd I)^*L(\xi_1,\wtd I)|_{\MH_0})\eta_1}\\
\bk{R(\xi_2,\wtd I)\eta_2|R(\xi_1,\wtd I)\eta_1}=\bk{\eta_2|\pi_{\wtd I}(R(\xi_2,\wtd I)^*R(\xi_1,\wtd I)|_{\MH_0})\eta_1}
\end{gather}
\end{subequations}
Therefore, by \eqref{eq4} and Lem. \ref{lb18},
\begin{align*}
&\bk{R(\xi_2,\wtd I)\Gamma_\phi\eta_2|R(\xi_1,\wtd I)\Gamma_\phi\eta_1}=\bk{\eta_2|\Gamma_\phi^{-1}\pi_{\wtd I}(R(\xi_2,\wtd I)^*R(\xi_1,\wtd I)|_{\MH_0})\Gamma_\phi\eta_1}\\
=&\bk{\eta_2|\pi_{\wtd I}(\phi^{-1}R(\xi_2,\wtd I)^*R(\xi_1,\wtd I)\phi|_{\MH_0})\eta_1}=\bk{\eta_2|\pi_{\wtd I}(L(\xi_2,\wtd I)^*L(\xi_1,\wtd I)|_{\MH_0})\eta_1}\\
=&\bk{L(\xi_2,\wtd I)\eta_2|L(\xi_1,\wtd I)\eta_1}
\end{align*}
Therefore, there is a unitary map $\mbb B_{i,j}^{\wtd I}:\MH_i\boxtimes\MH_j\rightarrow\MH_{\phi j}\boxtimes\MH_i$ sending each $L(\xi,\wtd I)\eta$ to $R(\xi,\wtd I)\Gamma_\phi\eta$. As in Step 2 of the proof of Thm. \ref{lb28}, this map is independent of $\wtd I$, and hence can be written as $\Bbb_{i,j}$.

For each $x\in\MA(I)$, by Prop. \ref{lb22} and \ref{lb21} we have
\begin{align*}
&\Bbb_{i,j}\pi_{\wtd I}(x)L(\xi,\wtd I)\eta=\Bbb_{i,j}L(\pi_{\wtd I}(x)\xi,\wtd I)\eta=R(\pi_{\wtd I}(x)\xi,\wtd I)\Gamma_\phi\eta\\
=&\pi_{\wtd I}(x)R(\xi,\wtd I)\Gamma_\phi\eta=\pi_{\wtd I}(x)\Bbb_{i,j}L(\xi,\wtd I)\eta
\end{align*} 
This proves that $\Bbb_{i,j}$ is a $G$-twisted $\MA$-module morphism.
\end{proof}

\subsection{Conformal covariance}

\begin{proof}[\textbf{Proof of Thm. \ref{lb15}-(6)}]
Fix $\wtd I\in\Jtd$. By Lem. \ref{lb14}, it suffices to prove \eqref{eq14} when $\xi\in\MH_i(I)$ and $g\in\GA(J)$, where $\wtd J\in\Jtd$ is such that $\wtd I$ and $\wtd J$ can be covered by some $\wtd K\in\Jtd$. By Rem. \ref{lb13}, we have
\begin{align*}
U(g)L(\xi,\wtd I)U(g)^{-1}|_{\MH_0}\in\Hom_{\MA(g\wtd I')}(\MH_0,\MH_i)
\end{align*}
and hence
\begin{align*}
g\xi g^{-1}:=U(g)L(\xi,\wtd I)U(g)^{-1}\Omega\in\MH_i(gI)
\end{align*}
By Thm. \ref{lb12} and Prop. \ref{lb22}, we have
\begin{gather*}
U(g)=\pi_{\wtd K}(U_0(g))=L(U_0(g)\Omega,\wtd K)
\end{gather*}
when acting on any object of $\RepGA$. Therefore, by Prop. \ref{lb21} and isotony (Cor. \ref{lb31}),
\begin{align*}
&U(g)L(\xi,\wtd I)U(g)^{-1}=L(U_0(g)\Omega,\wtd K)L(\xi,\wtd K)L(U_0(g^{-1})\Omega,\wtd K)\\
=&L(U_0(g)\Omega,\wtd K)L(L(\xi,\wtd K)U_0(g^{-1})\Omega,\wtd K)=L(L(U_0(g)\Omega,\wtd K)L(\xi,\wtd K)U_0(g^{-1})\Omega,\wtd K)\\
=&L(U(g)L(\xi,\wtd K)U(g)^{-1}\Omega,\wtd K)=L(g\xi g^{-1},\wtd K)
\end{align*}
This proves the first
identity in \eqref{eq14}. The second identity follows by a similar argument.
\end{proof}

\subsection{Balancing}\label{lb139}

Assume Convention \ref{lb10}. (In fact, it suffices to assume the identifications $\MH_{\phi(\psi i)}=\MH_{(\phi\psi)i}$ and $\MH_{\phi i\boxtimes\phi j}=\MH_{\phi(i\boxtimes j)}$ via $\mbf C$ and $\mbf D$, respectively. Accordingly, we will use the identity $\Gamma_\phi\Gamma_\omega=\Gamma_{\phi\omega}$ as well as \eqref{eq75}.)

The following proposition generalizes its untwisted counterpart in \cite[Lem. 1.5]{Gui25}.

\begin{pp}\label{lb60}
Let $\phi,\omega\in G$, $\MH_i\in\Rep^\phi(\MA)$, $\MH_j\in\Rep^\omega(\MA)$, $\wtd I\in\Jtd$, $\xi\in\MH_{i}(I)$, and $\eta\in\MH_{j}$. Then, as elements of $\MH_{\phi\omega i}\boxtimes\MH_{\phi\omega j}$, we have
\begin{align*}
\Gamma_{\phi\omega}L(\xi,\varrho(2\pi)\wtd I)\eta=\Bbb_{\phi\omega j,\phi i}\Bbb_{\phi i,\omega j}L(\Gamma_\phi\xi,\wtd I)\Gamma_\omega\eta
\end{align*}
\end{pp}

\begin{proof}
Let $\wtd I_1=\varrho(2\pi)\wtd I$. Assume without loss of generality that $\eta\in\MH_j(I')=\MH_j(I_1')$. Then by Prop. \ref{lb20} and the braiding in Thm. \ref{lb15}, and noting $\Gamma_\phi\Gamma_\omega=\Gamma_{\phi\omega}$,
\begin{align*}
&\Bbb_{\phi\omega j,\phi i}\Bbb_{\phi i,\omega j}L(\Gamma_\phi\xi,\wtd I)\Gamma_\omega\eta=\Bbb_{\phi\omega j,\phi i}R(\Gamma_\phi\xi,\wtd I)\Gamma_{\phi\omega}\eta=\Bbb_{\phi\omega j,\phi i}R(\Gamma_\phi\xi,\wtd I_1'')\Gamma_{\phi\omega}\eta\\
=&\Bbb_{\phi\omega j,\phi i}L(\Gamma_{\phi\omega}\eta,\wtd I_1')\Gamma_\phi\xi=R(\Gamma_{\phi\omega}\eta,\wtd I_1')\Gamma_{\phi\omega\phi^{-1}}\Gamma_\phi\xi=R(\Gamma_{\phi\omega}\eta,\wtd I_1')\Gamma_{\phi\omega}\xi\\
=&L(\Gamma_{\phi\omega}\xi,\wtd I_1)\Gamma_{\phi\omega}\eta=\Gamma_{\phi\omega}L(\xi,\varrho(2\pi)\wtd I)\eta
\end{align*}
where the $G$-covariance in Thm. \ref{lb15} is used for the last equality.
\end{proof}

Recall that for each $\MH_i\in\Rep^\phi(\MA)$, the balancing $\vartheta_i\in\Hom_\MA(\MH_i,\MH_{\phi i})$ is defined by $\Gamma_\phi\circ U_i(\varrho_\MA(2\pi))$.

\begin{thm}
For each $\MH_i\in\Rep^\phi(\MA)$ and $\MH_j\in\Rep^\omega(\MA)$, we have 
\begin{gather*}
\fk T_\phi(\vartheta_i)=\vartheta_{\phi i}\qquad \vartheta_{i\boxtimes j}=\Bbb_{\phi\omega j,\phi i}\Bbb_{\phi i,\omega j}(\vartheta_i\boxtimes\vartheta_j)
\end{gather*}
\end{thm}

\begin{proof}
By Rem. \ref{lb59},
\begin{align*}
\fk T_\phi(\vartheta_i)=\Gamma_\phi\vartheta_i\Gamma_\phi^{-1}|_{\MH_{\phi i}}=\Gamma_\phi\Gamma_\phi U_i(\varrho_\MA(2\pi))\Gamma_\phi^{-1}|_{\MH_{\phi i}}=\Gamma_\phi U_{\phi i}(\varrho_\MA(2\pi))|_{\MH_{\phi i}}=\vartheta_{\phi i}
\end{align*}
For each $\xi,\eta$ as in Prop. \ref{lb60}, by Prop. \ref{lb26},
\begin{align*}
&\Bbb_{\phi\omega j,\phi i}\Bbb_{\phi i,\omega j}(\vartheta_i\boxtimes\vartheta_j)L(\xi,\wtd I)\eta=\Bbb_{\phi\omega j,\phi i}\Bbb_{\phi i,\omega j}L(\vartheta_i\xi,\wtd I)\vartheta_j\eta\\
=&\Bbb_{\phi\omega j,\phi i}\Bbb_{\phi i,\omega j}L(\Gamma_\phi\varrho_\MA(2\pi)\xi,\wtd I)\Gamma_\omega\varrho_\MA(2\pi)\eta\\
=&\Gamma_{\phi\omega}L(\varrho_\MA(2\pi)\xi,\varrho(2\pi)\wtd I)\varrho_\MA(2\pi)\eta=\Gamma_{\phi\omega}\varrho_\MA(2\pi)L(\xi,\wtd I)\eta
\end{align*}
where the conformal covariance in Thm. \ref{lb15} is used in the last equality.
\end{proof}

\end{appendices}

\noindent {\small \sc Yau Mathematical Sciences Center, Tsinghua University, Beijing, China.}

\noindent {\textit{E-mail}}: binguimath@gmail.com\qquad bingui@tsinghua.edu.cn
\end{document}